\documentclass[12pt]{article}

\usepackage{graphicx}
\usepackage{caption}
\usepackage{subcaption}

\usepackage{float}

\usepackage{booktabs}

\usepackage{multirow}

\usepackage{amsmath} 
\usepackage{amsthm} 
\usepackage{amssymb} 
\usepackage{physics} 
\usepackage{amsfonts} 
\usepackage{bm}

\allowdisplaybreaks

\usepackage[title]{appendix}

\usepackage[backend=biber, style=apa, sorting=nty,  doi=false]{biblatex}

\DeclareNameFormat{labelname}{%
	\nameparts{#1}%
	\usebibmacro{name:family}
	{\namepartfamily}
	{\namepartgiven}
	{\namepartprefix}
	{\namepartsuffix}%
	\usebibmacro{name:andothers}%
}

\usepackage{hyperref}
\hypersetup{
	colorlinks=true, 
	linkcolor=blue, 
	urlcolor=blue, 
	citecolor=blue 
}

\usepackage{xcolor}
\usepackage{color}
{}           

\usepackage{CJK, CJKnumb}

\usepackage{indentfirst}
\usepackage{mathrsfs}

\makeatletter
 \newcommand{\Rmnum}[1]{\expandafter\@slowromancap\romannumeral #1@}
\makeatother

\usepackage{chngcntr}

\numberwithin{equation}{section}

\usepackage{cleveref}
\crefname{Theorem}{Theorem}{Theorem}
\Crefname{Theorem}{Theorem}{Theorem}

\crefname{figure}{fig.}{figs.}     
\Crefname{figure}{Fig.}{Figs.}     

\newtheorem{Definition}{Definition}[section]
\newtheorem{Theorem}{Theorem}[section]
\newtheorem{Lemma}{Lemma}[section]
\newtheorem{Corollary}{Corollary}[section]

\numberwithin{equation}{section}
\newtheorem{Remark}{Remark}[section]

\makeatletter
\long\def\@makefntext#1{\parindent 1em\noindent *\ #1}
\makeatother
\begin{document}
	\title{Equilibrium analysis in a multi-agent reinsurance chain}

	\author{Kaizheng Wang$^1$, \quad Wei Liu$^{1,*}$, \quad Zhuo Jin$^2$, \quad Wenyuan Wang$^3$\\
		\\
		$^1$ College of Mathematics and System Science, Xinjiang University, \\
		Urumqi 830017, China 
		\\		
		\\
		$^2$ Department of Actuarial Studies and Business Analytics, Macquarie University, \\
		2109, NSW, Australia
	   \\
	   \\
		$^3$ School of Mathematics and Statistics, 
		 Fujian Normal University, \\
		  Fuzhou 350007, China
	}

	\date{}
	
	\maketitle
	
		\footnotetext{Email address: kzwang\underline{ }math@stu.xju.edu.cn  (K.Wang); 
			liuwei.math@xju.edu.cn (W.Liu);\\
			zhuo.jin@mq.edu.au (Z. Jin); 
			wwywang@xmu.edu.cn (W. Wang). }

	\def\sA{{\mathscr A}}
	\def\sB{{\mathscr B}}
	\def\sC{{\mathscr C}}
	\def\sD{{\mathscr D}}
	\def\sF{{\mathscr F}}
	\def\sG{{\mathscr G}}
	\def\sH{{\mathscr H}}
	\def\sM{{\mathscr M}}
	\def\sP{{\mathscr P}}
	\def\sQ{{\mathscr Q}}
	\def\sR{{\mathscr R}}
	\def\sS{{\mathscr S}}
	\def\sT{{\mathscr T}}
	\def\sW{{\mathscr W}}
	\def\sX{{\mathscr X}}
	\def\sY{{\mathscr Y}}

	\begin{abstract}
This paper investigates a multi-layer reinsurance chain within a stochastic differential game framework involving $m$ competing insurers and $n$ reinsurers. Specifically, Stackelberg differential games are employed to characterize the strategic interactions between reinsurance buyers and sellers at each layer of the chain. In addition, a non-zero-sum game model is established to capture the competitive behavior among insurers. Both insurers and reinsurers are allowed to invest in a risk-free asset and a risky asset. To examine the heterogeneity of reinsurance chains under different contract types, the analysis is conducted separately for proportional reinsurance and excess-of-loss reinsurance. By combining dynamic programming and game theory, closed-form equilibrium strategies for investment and reinsurance are derived by solving the extended Hamilton–Jacobi–Bellman (HJB) systems under the mean-variance (MV) criterion. Numerical analysis is conducted to explore the impact of key parameters on the equilibrium strategies. The results indicate that intensified competition in the insurance market leads to a reduction in the safety loadings of reinsurance contracts at each layer of the reinsurance chain.
	\end{abstract}

	\vspace{0.2cm}
	

	\vspace{0.2cm}
	
	
	%

	\newpage
	\section{Introduction}
	Since the seminal work of \textcite{browne1995optimal}, the study of optimal reinsurance has attracted substantial academic attention. Early research primarily examined the optimal strategies for a single insurer, as exemplified by \textcite{hojgaard1998optimal}, \textcite{david2005minimizing}, and \textcite{cao2009optimal}. More recently, driven by competitive dynamics in insurance markets, game-theoretic approaches have become a central theme in actuarial optimization. To model insurer competition, \textcite{zeng2010}, \textcite{taksar2011}, and \textcite{meng2015} adopted zero-sum game frameworks. In contrast, \textcite{bensoussan2014}, \textcite{pun2016a}, and \textcite{pun2016b} proposed non-zero-sum game models to capture strategic interactions among insurers whose objectives are not strictly antagonistic. 
	However, these contributions largely neglect the strategic role of reinsurers, whose  decisions also play an important role in contract design. Addressing this gap, \textcite{chen2018new} pioneered the application of Stackelberg differential games to reinsurance-investment problems, modeling the reinsurer as the leader and the insurer as the follower in a hierarchical decision structure. Building on this framework, subsequent studies such as \textcite{gu2020}, \textcite{wang2020}, and \textcite{yuan2022robust} have further enriched the literature by examining optimal strategies under varying assumptions regarding market structure, information asymmetry, and robustness.

	Most existing studies primarily consider single-layer reinsurance structures. However, with the increasing prevalence of cross-border risks driven by economic globalization, traditional reinsurance frameworks are often inadequate to capture the complexities of modern risk-sharing arrangements. To address the multilayered and interdependent nature of insurer–reinsurer relationships, recent research has adopted network-based models to study strategic interactions within reinsurance markets. \textcite{lin2015reinsurance} analyzed the impact of reinsurance networks on decision-making by building a stylized theoretical model, complemented by empirical validation using real-world data. \textcite{chen2020reinsurancenetwork} examined the microstructure of the U.S. property reinsurance market, highlighting the dynamics of default risk and contagion within the network. \textcite{klages2020cascading} introduced a cascading contagion model to explore systemic losses resulting from interconnected exposures. 
	By building a global reinsurance network,  international reinsurance companies have made it possible to spread risks more rationally and efficiently around the world. 
	Within these complex network structures,   reinsurance chains constitute fundamental  structures. 
	This paper concentrates on the study of reinsurance chains, a topic with many unresolved issues that deserves significant attention.

	In reality, reinsurers, when unable to bear excessive risks, will choose to transfer the risks further.
    The National Association of Insurance Commissioners
	states in its definition of reinsurance that reinsurers may also buy reinsurance protection, which is
	called "retrocession". Multiple layers of reinsurance and retrocession form a reinsurance chain in which reinsurers simultaneously act as leaders in the upstream risk transfer  layer and followers in the subsequent layer. This hierarchical structure can be captured by the Stackelberg game framework.  
	Reinsurance chains are common in practice. 
	For example, during Hurricane Katrina in 2005, multiple insurers and reinsurers faced massive claims and were on the verge of bankruptcy. 
	Some of these reinsurers had pre-arranged retrocession agreements to transfer a portion of the hurricane risk to other international reinsurers. Consequently, they were able to effectively alleviate their financial burden after the storm, ensuring their financial stability and continued operational capability. 
	In such scenarios, insurers provide coverage to residents or businesses in affected areas, while first-layer reinsurers receive the risk ceded by the insurers, and retrocessionaires receive retrocession by the reinsurers 
	to further diversify the transferred losses 
	and prevent insurers and reinsurers from going bankrupt due to the concentration of risk. 
	\textcite{davison2019counterparty}  points out that, after retrocession, the losses caused by hurricane Katrina to reinsurers are reduced by half.
	Therefore, the study of the equilibrium strategies of insurers and reinsurers  under the reinsurance chain is of great significance in promoting the development of the economy and maintaining the stability of the market.

	The study of reinsurance chains traces back to  \textcite{gerber1984chains}, who pioneered the derivation of optimal reinsurance cession and pricing strategies under the assumption that insurance claims follow a normal distribution.  Subsequent research by scholars such as \textcite{d1985chains} and \textcite{lemaire1986chains} further examined the structural properties of reinsurance chains.  
	Despite the prevalence of reinsurance chains in practice, a substantial literature gap emerged following early contributions in the 1980s, and the topic remained largely unexplored for decades.  Recently, \textcite{chen2020continuous} revitalized this line of inquiry by analyzing optimal strategies within a continuous-time Stackelberg game framework, investigating how structural variations in reinsurance chains influence optimal decision-making.  
	Building on this, \textcite{cao2023-two-reinsurance} compared the performance of reinsurance chain structures with that of reinsurance trees, while in another study, \textcite{cao2023-n-reinsurance}  explored the impact of model ambiguity on the optimal strategies under reinsurance chain structure. 
	
	Existing studies in this domain predominantly focus on settings involving a single insurer. In practice, however, 
	the number of insurers often substantially exceeds the number of reinsurers, which frequently necessitates that a reinsurer simultaneously manages multiple reinsurance contracts from various insurers.  
	A primary driver for the emergence of reinsurance chains is the substantial volume of risk ceded by insurers, which often requires retrocession  to achieve further risk diversification. 
	At the same time, intense competition among insurers compels them to incorporate the strategies of their competitors into their own decision-making processes. 
	 \textcite{han2019stochastic} examined optimal strategies within a two-layer Stackelberg game framework comprising two insurers and one reinsurer, employing a bifractional Brownian motion environment. 
	 \textcite{he2024stochastic} 
	 investigate equilibrium strategies for two insurers and one reinsurer in a delayed stochastic Stackelberg game under value-at-risk (VaR) constraints.  
	 \textcite{wang2024optimal} examines the optimal investment and reinsurance strategies of $n$ insurers with competitive and cooperative interactions and a reinsurer under the Stackelberg game framework.
    
	\begin{figure}[h]
			\includegraphics[width=\textwidth]{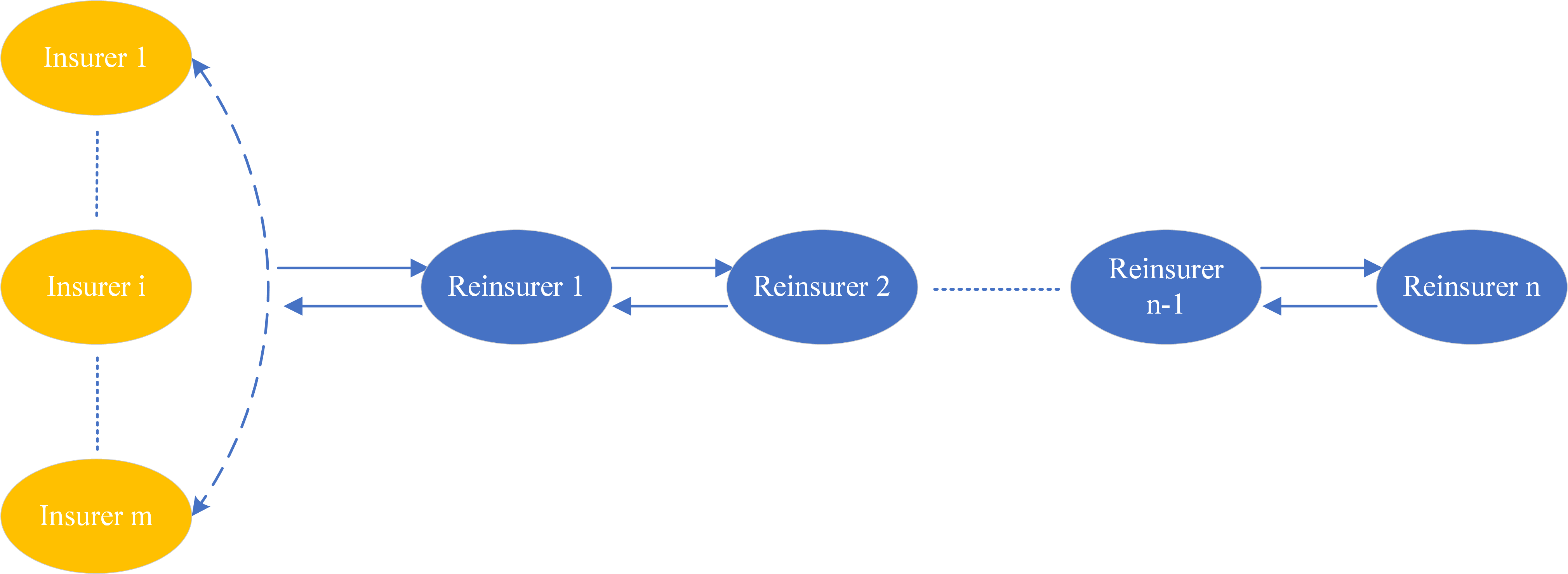}
		\caption{The reinsurance chain of $m$ insurers and $n$ reinsurers}	
		\label{fig:Reinsurance-chain}
	\end{figure}

	This paper constructs a multi-layer stochastic Stackelberg differential game model to analyze a reinsurance chain comprising $m$ insurers and $n$ reinsurers (see \Cref{fig:Reinsurance-chain} as an example). The model captures the joint propagation of systemic risk emerging from both market competition and the chain structure itself. It offers a precise representation of the complex risk network formed when a limited number of reinsurers engage with multiple primary insurers, thereby elucidating the influence of competitive dynamics in the primary insurance market on reinsurance demand.  
	Extending the frameworks developed in  \textcite{chen2020continuous} and \textcite{cao2023-two-reinsurance,cao2023-n-reinsurance}, this study generalizes the reinsurance chain model from a single-insurer context to a multi-insurer setting and provides a rigorous analysis of how competition among insurers affects the reinsurance chain. In contrast to \textcite{wang2024optimal}, where a single-layer Stackelberg game with multiple insurers is considered, the present framework generalizes the hierarchical structure to a multi-layer reinsurance chain and further examines the role of such chain configurations in shaping reinsurance contracts. 
	The proposed model not only provides a more accurate characterization of multi-layer reinsurance but also explicitly accounts for the empirical market structure, in which numerous insurers interact with a limited number of reinsurers. In practice, a small set of reinsurers typically undertakes the risk-transfer obligations of multiple insurers, while competitive and cooperative interactions among insurers significantly influence reinsurance strategies. To capture these dynamics, the model incorporates a multi-layer reinsurance chain framework, more faithfully reflecting decision-making interactions and risk-sharing processes among market participants. Consequently, the model exhibits strong generalizability, offering a rigorous framework applicable to various insurance products and market scenarios, and providing a theoretical basis for risk management and strategy optimization.
	
			Building  upon the  multi-dimensional reinsurance chain model above, we next turn to the methodological challenges inherent in extending the reinsurance chain framework to a multi-insurer and multi-layer setting. Specifically, the inclusion of multiple insurers introduces higher-dimensional interactions that necessitate new solution approaches. 
		Recent studies on reinsurance chains, such as \textcite{chen2020continuous},  \textcite{cao2023-n-reinsurance} and \textcite{cao2023-two-reinsurance},  have primarily focused on optimal strategies within a single-insurer framework. In contrast, the present study  extends the reinsurance chain framework from a low-dimensional to a high-dimensional setting, aiming to capture the more complex interactions among multiple insurers and reinsurers. In this extension, we systematically examine the key challenges arising from the transition, including intensified strategy coupling, increased complexity of risk transfer mechanisms, and higher computational demands for solution methods. 
		
		To address these challenges, we develop a recursive analytical framework that enables tractable characterization of equilibrium strategies in high-dimensional environments. 
		The inclusion of $m$ insurers renders the derivation and verification of optimal reinsurance contracts considerably more challenging. In order to obtain the explicit solution in high dimensions, we propose a backward recursive approach. The essence of this approach lies in establishing recursive relations among hierarchical equilibrium strategies. 
		Unlike \textcite{chen2020continuous} and \textcite{cao2023-n-reinsurance}, which examine only the single-insurer setting, this paper introduces a structured matrix that systematizes the transmission of strategies across layers, thereby substantially mitigating the computational burden inherent in high-dimensional optimization. 
	    In contrast to the previous studies limited to a single Stackelberg game, our approach is capable of capturing multi-layer reinsurance interactions in real markets with higher precision. Based on this recursive formulation, we derive equilibrium strategies under both proportional and excess-of-loss reinsurance contracts, thereby providing an effective methodological reference for decision-making in competitive reinsurance markets. 
	
	Beyond methodological advancements, the proposed framework also yields deeper insights into the strategic interdependencies characterizing multi-insurer systems.  Extending the research framework from a single-insurer model to a multi-insurer setting introduces tighter strategic interconnections and more intricate interdependencies within the risk-transfer mechanism. In this extended framework, each insurer’s decision not only affects its own expected utility but also influences the equilibrium strategies of other insurers, thereby generating a multi-layer network of strategic interactions transmitted through the reinsurance chain. Distinct from prior studies that primarily focus on a single Stackelberg structure, the recursive formulation proposed in this paper captures the multi-layer interactions among reinsurers with higher analytical precision and stronger empirical relevance. This methodological innovation provides a rigorous theoretical foundation for strategic decision-making in competitive reinsurance markets and establishes a feasible analytical framework for dynamic optimization problems in multi-agent and multi-layer insurance systems. 
	
	The main contributions of this paper are threefold.
	
	First, we extend the classical reinsurance framework to a multi-layer setting with $m$ insurers and $n$ reinsurers, formulating the problem as a coupled stochastic control system. By deriving and solving the associated extended  Hamilton–Jacobi–Bellman equations, we obtain explicit closed-form solutions for the optimal investment and reinsurance strategies, providing tractable expressions for equilibrium controls in this complex multi-agent system.
	
	Second, to the best of our knowledge, this is the first study to address the stochastic optimal investment and reinsurance problem within a multi-layer reinsurance chain that simultaneously incorporates both proportional and excess-of-loss treaties. This generalization captures the heterogeneous contractual structures observed in practice and unifies multiple reinsurance mechanisms within a single analytical framework.
	
	Third, we rigorously analyze the strategic interactions among insurers and the structural properties of the reinsurance network. Our results show that competition among insurers propagates through the entire reinsurance chain, affecting the design, pricing, and risk allocation of subsequent contracts. A systematic comparative statics analysis characterizes how key model parameters—such as risk aversion, claim intensity, and reinsurance costs—influence equilibrium strategies and resulting risk-sharing outcomes.
	
	In contrast to single-layer models, the equilibrium strategies derived in Theorem \ref{Theorem-Proportional reinsurance-equilibrium strategy} demonstrate that interactions at the initial level of the chain generate endogenous adjustments throughout all layers, revealing a mathematically tractable mechanism by which competitive dynamics and network structure jointly shape the global equilibrium of a reinsurance market.
	
	The rest of the paper is organized as follows. In Section 2, 
	the structure of the reinsurance-chain model is given. Section 3 presents the equilibrium strategies under two scenarios: proportional reinsurance and excess-of-loss reinsurance. In Section 4, the paper examines the relevant properties of reinsurance chains. Numerical 
	analysis and economic explanation are illustrated in Section 5. Section 6 summarizes the paper.

	\section{Model formulation}
	
	Let $(\Omega, \mathcal{F}, \mathbb{P})$ be a complete filtered probability space, where the filtration $\mathcal{F} = \{\mathcal{F}_t\}_{t \in [0,T]}$ is generated by a collection of mutually independent processes, including  a standard Brownian motion $B(t)$, Poisson processes $N(t)$ and $N_i(t)$, and a family of independent random variables $\{Z_k^i : k \geq 1\}$ for $i \in \{1,2, \cdots, m\}$. Here, all processes and random variables are assumed to be mutually independent. The probability measure $\mathbb{P}$ denotes the real-world probability measure, and the interval $[0,T]$ represents a fixed finite time horizon.

	\subsection{The wealth process of insurers}
		 Let $R_j (j=1, 2, \cdots, n)$ denote the $n$ reinsurers. In the absence of reinsurance arrangements, the reserve process of insurer
    $i (i=1, 2, \cdots, m)$ follows the stochastic differential equation
	\begin{equation*}
		\begin{aligned}
			\mathrm{d}U_{i}(t)&=c^{0}_{i}\mathrm{d}t-\mathrm{d}\sum_{k=1}^{N_{i}(t)+N(t)}Z^{i}_{k}, 
		\end{aligned}
	\end{equation*}
	where $c^{0}_{i}$ is a constant representing the premium rate of insurer $i$. The claim sizes $\bigl\{Z^{i}_{k}, k=1,2,\cdots \bigr\}$ form a sequence of independent and identically distributed (i.i.d.) non-negative random variables with common distribution function $F_i(z)$, finite first and second moments $E[Z^{i}]=\mu_{i}$ and  $E[(Z^{i})^2]=b_{i}^2$. The process
	$N_{i}(t)+N(t)$ 
	represents the total number of claims experienced by insurer $i$ up to time $t$. $N_{i}(t)$ and $N(t)$ are mutually independent Poisson processes with intensities $\lambda_i$ and $\lambda$, respectively. The sequence $\left\{Z^{i}_{k}, k=1,2,\cdots\right\}$ is independent of both $N_{i}(t)$ and $N(t)$. Thus, the interdependence among the $m$ insurers arises from their exposure to the common shock process $N(t)$.
    The  correlation coefficient between the aggregate claims $\sum_{k=1}^{N_{x}(t)+N(t)}Z^{x}_{k}$ and  $\sum_{k=1}^{N_{y}(t)+N(t)}Z^{y}_{k}$ is
	\begin{equation*}
		\rho_{xy}=\frac{\lambda \mu_x\mu_y}
		{\sqrt{\left(\lambda_x+\lambda\right)b_x^2}\sqrt{\left(\lambda_y+\lambda\right)b_y^2}}. 
	\end{equation*}
	
	Under the expected value premium principle, the premium rate $c^{0}_{i}$ for insurer $i$ is 
	 \begin{equation*}
	 c^{0}_{i}=\left(1+\eta_{i}\right)\left(\lambda_{i}+\lambda\right) \mu_{i},
	 \end{equation*}
	  where $\eta_{i}>0$ is the safety loading of insurer $i$. For each claim, insurer $i$ retains $l^{0}_{i}(Z^i_k,t)$, while the residual amount $Z^i_k-l^{0}_{i}(Z^i_k,t)$ is ceded to the reinsurer $R_1$ at a premium rate of $c^{1}_{i}(t)$.  After reinsurance, the surplus process of insurer $i$ satisfies
	\begin{equation}
		\begin{aligned}
			\mathrm{d}U_{i}(t)&=\left(c^{0}_{i}-c^{1}_{i}(t)\right)\mathrm{d}t
			-\mathrm{d} \left(\sum_{k=1}^{N_{i}(t)+N(t)}l^{0}_{i}(Z^{i}_{k}, t)\right).
		\end{aligned}
	\end{equation}

    All insurers and reinsurers invest in a financial market consisting of a risk-free asset $S_0(t)$ with a constant interest rate $r$ and a risky asset $S_1(t)$. The price processes of $S_0(t)$ and $S_1(t)$ follow the stochastic differential equations
    \begin{equation*}
    	\left\{
    	\begin{aligned}
    		&
    		\frac{\mathrm{d}S_1(t)}{S_1(t)}=\mu\mathrm{d}t+\sigma \mathrm{d}B(t),
    		\\
    		&
           \frac{\mathrm{d}S_0(t)}{S_0(t)}=r\mathrm{d}t,
    	\end{aligned}
    	\right.
    \end{equation*}
    where $ \mu > r $ is expected return, $ \sigma > 0$ is the  volatility, and $ B(t) $ is a standard Brownian motion.

	 Assume that insurer $i$ invests $\pi_{i}(t)$  in the risky asset $S_{1}(t)$ and allocates the remaining wealth to the risk-free asset $S_{0}(t)$. Thus, the wealth process $X_{i}(t)$ of  insurer $i$ follows
		\begin{equation}
			\label{Weathy process of insurer i}
			\mathrm{d}X_{i}(t)
			=\left[(\mu-r)\pi_{i}(t)+rX_{i}(t)+  c^{0}_{i}-c^{1}_{i}(t) \right] \mathrm{d}t 
			-\mathrm{d}\sum_{k=1}^{N_{i}(t)+N(t)}l^{0}_{i}(Z^{i}_{k},t)
			+\pi_{i}(t)\sigma\mathrm{d}B(t). 
		\end{equation}
	
	Relative performance measures an insurer's wealth relative to that of its competitors.  Specifically, the relative performance process $\widehat X_i(t)$ of insurer $i$ is defined as
	\begin{align}
		\mathrm{d}\widehat X_i(t)
		&=\mathrm{d}\left(X_{i}(t)-\omega_{m,i}\sum_{k\neq i}^{m}X_k(t)\right) \notag\\
		&=\Bigg[r\widehat X_i(t)
		+(\mu-r)\left(\pi_{i}(t)-\omega_{m,i}\sum_{k\neq i}^{m}\pi_{k}(t)\right) \notag\\
		&\quad+\left(c^{0}_{i}-c^{1}_{i}(t)\right)
		-\omega_{m,i}\sum_{k\neq i}^{m}\left(c^0_k-c^1_k(t)\right)\Bigg]\mathrm{d}t 
		\label{Relative performance}\\
		&\quad-\mathrm{d}\left[\sum_{v=1}^{N_{i}(t)+N(t)}l^{0}_{i}(Z^{i}_{v},t)
		-\omega_{m,i}\sum_{k\neq i}^{m}\sum_{v=1}^{N_{k}(t)+N(t)}l^{0}_{k}(Z^{k}_{v},t)\right] \notag\\
		&\quad+\left(\pi_{i}(t)-\omega_{m,i}\sum_{k\neq i}^{m}\pi_k(t)\right)\sigma\mathrm{d}B(t), \notag
	\end{align}
	where $\omega_{m,i} = \frac{\omega_i}{m - 1}$.  $\omega_i \in [0, 1]$ is the competition weight parameter. A higher value of $\omega_i$ indicates that insurer $i$ places more weight on relative performance, thereby intensifying strategic interactions in the game.

	\subsection{The wealth processes of the reinsurers}
	    Reinsurer $R_1$  receives claims $\sum_{k=1}^{N_{i}(t)+N(t)}\left(Z^{i}_{k}-l^{0}_{i}(Z^{i}_{k},t)\right)$ 
	    from insurer $i$ at premium rate $c^{1}_{i}(t)$, and subsequently transfers claims $\sum_{k=1}^{N_{i}(t)+N(t)}\left(Z^i_k-l^{0}_{i}(Z^{i}_{k},t)-l^{1}_{i}(Z^{i}_{k},t)\right)$ to reinsurer $R_{2}$ at premium rate $c^{2}_{i}(t)$. For any $ j\in\left\{2,\cdots,n-1\right\}$, 
		reinsurer $R_j$  receives claims $\sum_{k=1}^{N_{i}(t)+N(t)}$ $\big(Z^i_k-\sum_{v=0}^{j-1}l^{v}_{i}(Z^{i}_{k},t)\big)$ 
		from reinsurer $R_{j-1}$ at the premium rate $c^{j}_{i}(t)$, and further cedes claims $\sum_{k=1}^{N_{i}(t)+N(t)}$\quad$\big(Z^i_k-\sum_{v=0}^{j}l^{v}_{i}(Z^{i}_{k},t)\big)$ to reinsurer $R_{j+1}$ at premium rate $c^{j+1}_{i}(t)$. Consequently, the surplus process $U_{R_j}$ of reinsurer $R_j$ follows
	\begin{equation}
		\begin{aligned}
			\mathrm{d}U_{R_j}(t)&=\sum_{i=1}^{m}\left(c^{j}_{i}(t)-c^{j+1}_{i}(t)\right)
			\mathrm{d}t
			-\mathrm{d}\sum_{i=1}^{m}\sum_{k=1}^{N_{i}(t)+N(t)}l^{j}_{i}(Z^{i}_{k},t).
		\end{aligned}
	\end{equation}
	Under the expected value premium principle, the premium rate $c^{j}_{i}(t)$ is given by
	\begin{equation*}
		c^{j}_{i}(t)=\left(1+\theta^{j}_{i}(t)\right)\left(\lambda_{i}+\lambda\right)
		\mathbb{E}\left[Z^{i}-\sum_{k=0}^{j-1}l^{k}_{i}(Z^{i},t)\right],
	\end{equation*}
    where $\theta^{j}_{i}(t)$ is the reinsurance safety loading strategy of reinsurer $R_{j}$ when providing coverage to reinsurer $R_{j-1}$. 
    
  Reinsurer $R_j$ invests an amount $\pi_{R_j}(t)$ in $S_1(t)$, holding the remainder in $S_0(t)$. Its wealth process $Y_j(t)$  therefore satisfies 
	\begin{align}
		\mathrm{d}Y_j(t)
		&=\pi_{R_j}(t)\frac{\mathrm{d}S_1(t)}{S_1(t)}
		+\left(Y_{j}(t)-\pi_{R_j}(t)\right)\frac{\mathrm{d}S_0(t)}{S_0(t)}
		+\mathrm{d}U_{R_j}(t) 
		\notag\\
		&=\Big[rY_j(t)+(\mu-r)\pi_{R_j}(t)
		+\sum_{i=1}^{m}\big(c^{j}_{i}(t)-c^{j+1}_{i}(t)\big)\Big]\mathrm{d}t
		\label{Weathy process of Reinsurance}\\
		&\quad-\mathrm{d}\sum_{i=1}^{m}\sum_{k=1}^{N_{i}(t)+N(t)}
		l^{j}_{i}(Z^{i}_{k},t)
		+\pi_{R_j}(t)\sigma\mathrm{d}B(t). \notag
	\end{align}
	
	\subsection{The optimization problem}
	
	Define the strategy of insurer $i$ as $u_{i}(t):=(\pi_{i}, l^{0}_{i})(t)$, and the strategy of reinsurer $R_j$ as $u_{R_j}(t):=(\pi_{R_j},  l^{j}_{1}, l^{j}_{2}, \cdots,l^{j}_{m},  \theta^{j}_{1}, \theta^{j}_{2}, \cdots,\theta^{j}_{m})(t)$.  The aggregated strategy vector is denoted by $u(t)=\left(u_1, u_2, \cdots, u_{m}, u_{R_1}, u_{R_2}, \cdots, u_{R_n}\right)(t)$. For notational convenience, define $u_{i-}(t)=(u_{1}, \cdots, $ $u_{i-1}, u_{i+1}, \cdots, u_{m}, u_{R_1}, u_{R_2}, \cdots, u_{R_n})(t)$, and
	$u_{R_j-}(t)=(u_{1}, \cdots, u_m, u_{R_{1}}, \cdots, u_{R_{j-1}}$, $ u_{R_{j+1}}, \cdots, u_{R_{n}})(t)$. 
	
	\begin{Definition}
		(Admissible strategy) A strategy $\{u(t)\}$ is said to be admissible if it satisfies the following conditions
		
		(i) The strategy processes $\{u_{i}(t)\}$ and $\{u_{R_j}(t)\}$ are non-negative and  $\left\{\mathcal{F}_t\right\}$ progressively measurable. 
		
		(ii) For any $i\in\left\{1,2,\cdots,m\right\}$,  $c^{n}_{i}(t) \leq  \cdots \leq c^{1}_{i}(t) \leq c^{0}_{i}$,
		$0 \leq Z^{i}-\sum_{k=0}^{n-1}l^{k}_{i}(Z^{i}) \leq \cdots \leq Z^{i}-\sum_{k=0}^{j}l^{k}_{i}(Z^{i}) \leq \cdots \leq Z^{i}-l^{0}_{i}(Z^{i})\leq Z^i$. 
        
		(iii)
		Under the strategy $u(t)$, SDEs (\ref{Weathy process of insurer i}) and (\ref{Weathy process of Reinsurance}) have unique strong solutions $X_i^{u}(t)$ and $Y^{u}_{R_j}(t)$ that
		are càdlàg, $\left\{\mathcal{F}_t\right\}$-adapted, and satisfy the square‑integrability conditions $\mathbb{E}\left[\sup_{t\in[0,T]}\abs{X_{i}^u(t)}^2\right]<\infty$,
		$\mathbb{E}\left[\sup_{t\in[0,T]} \abs{Y_j^u(t)}^2\right]<\infty$.
	\end{Definition}
	
	Denote by $\mathcal{A}=\mathcal{A}_1\times\mathcal{A}_2 \times \cdots \times \mathcal{A}_{m}\times\mathcal{A}_{R_1}\times\cdots\times\mathcal{A}_{R_n}$ the set of all admissible strategies, with  $\mathcal{A}_{i}$ and $\mathcal{A}_{R_j}$ being the admissible strategy sets for insurer $i$  and reinsurer $R_j$, respectively.
	
	Within the mean–variance framework, the objective functionals for insurer $i$ and reinsurer $R_{j}$ are given by
	
	   \begin{equation*}
	   	\left\{
	   	\begin{aligned}
	   	&J^{u_{i}, u_{i-}}_{i}(t, x_{i})
	   	=\mathbb{E}_{t, x_{i}}\left[\widehat{X_{i}(T)}\right]-\frac{\gamma_{i}}{2}\mathrm{Var}_{t, x_{i}}\left[\widehat{X_{i}(T)}\right],
	   	\\
	   	&J^{u_{R_j}, u_{{R_j}-}}_{R_j}(t, y_j)=\mathbb{E}_{t, y_{j}}\left[Y_j(T)\right]-\frac{\gamma_{R_j}}{2}\mathrm{Var}_{t, y_{j}}\left[Y_j(T)\right].
	   	\end{aligned}
	   	\right.
	   \end{equation*}
	   where $\gamma_{i}$ and $\gamma_{R_j}$ are the risk aversion coefficients of insurer $i$ and reinsurer $R_j$, respectively. 
	   
	   The optimization problem of reinsurance chain can be decomposed into two coupled problems: one
	   centering on the insurers and the other on the reinsurers. 
	   
	   The value function of insurer $i$ is given by
	   		\begin{equation}
	   			\label{Insurance-value function}
	   			V^{i}(t, x_{i})=\sup_{u_{i}\in\mathcal{A}_{i}}
	   			J^{u_{i}, u_{i-}}_{i}(t, x_{i}).  
	   		\end{equation}	
	   		
	   		The  value function of reinsurer $R_j$ is given by
	   		\begin{equation}
	   			\label{reinsurance-value function}
	   			V^{R_j}(t, y_{j})=
	   			\sup_{u_{R_j}\in \mathcal{A}_{R_j}}J^{u_{R_j}, u_{{R_j}-}}_{R_j}(t,y_{j}).
	   		\end{equation}

\section{Nash equilibrium strategies in the reinsurance chain}
    In the reinsurance chain framework,  each reinsurer $R_j$ acts as a leader for its immediate downstream reinsurer $R_{j-1}$, and as a follower of its upstream reinsurer $R_{j+1}$. The Nash equilibrium strategies for insurers and reinsurers are defined as follows.
	\begin{Definition}
		(Nash equilibrium) For any $t\in[0,T]$, consider the strategy $\left(u_{i}^*(t), u_{i-}^*(t) \right)$. For each insurer \(i \in \{1,2, \cdots, m\}\), any $\varepsilon>0$ and any admissible control $u_{i}(t)\in\mathcal{A}_{i}$, define $u_{i}^\varepsilon(t)$ by
		\begin{equation}
			u_{i}^{\varepsilon}(s)=
			\left\{
			\begin{aligned}
				&u_{i}(s), s\in[t, t+\varepsilon], 
				\\
				&u^*_{i}(s), s\in[t+\varepsilon, T]. 
			\end{aligned}
			\right. 
		\end{equation}
		Then $u^*_{i}(t)$ is called an equilibrium strategy  for insurer $i$ if for every $u_{i}(t)\in\mathcal{A}_{i}$
		\begin{equation*}
        \liminf_{\varepsilon\rightarrow 0}\frac{J^{u^*_{i}, u^*_{i-}}_{i}(t, x_{i})-J^{u^{\varepsilon}_{i}, u_{i-}^*}_{i}(t, x_{i})}{\varepsilon}\ge0. 
		\end{equation*}

		For the reinsurance chain, fixed $ j\in\left\{1,2,\cdots,n\right\}$ and assume that $ u_{R_{j-1}}^*$
$\big(u^*_{R_{0}}=(u^*_1,u^*_2,\cdots,$ $ u^{*}_{n})\big)$ is known. For any $ u_{R_j}\in\mathcal{A}_{R_j}$, define
		\begin{equation*}
			u^{\varepsilon}_{R_j}(s)=
			\left\{	
			\begin{aligned}
				&u_{R_j}(s) , s\in [t, t+\varepsilon], 
				\\
				&u_{R_j}^*(s) , s\in [t+\varepsilon, T]. 
			\end{aligned}
			\right. 
		\end{equation*}
	 Then $u_{R_j}^*(t)$ is an equilibrium strategy for reinsurer $R_j$ if for any admissible $u_{R_j}\in\mathcal{A}_{R_{j}}$,
	\begin{equation*}
		\liminf_{\varepsilon\rightarrow 0}\frac{J^{u^*_{R_j},  u^*_{{R_j}-}}_{R_j}(t, y_j)-J^{u^{\varepsilon}_{R_j}, u^*_{{R_j}-}}_{R_j}(t, y_j)}{\varepsilon}\ge0. 
	\end{equation*}
	\end{Definition}
	
	Next, we define the corresponding infinitesimal generators. Let $\mathcal{L}^{u_{i}, u_{i-}}$ and  $\mathcal{L}^{u_{R_{j}}, u_{R_{j-}}}$ be given by
		\begin{align*}
			\quad\mathcal{L}^{u_{i}, u_{i-}}W(t, x_{i})
            &=W_t+
		\Biggl[
		rx_{i}+(\mu-r)\left(\pi_{i}(t)-\omega_{m,i}\sum_{k\neq i}^{m}\pi_{k}(t)\right)+\left(c^{0}_{i}-c^{1}_{i}(t)\right)
		\\
		&\quad
		-
		\omega_{m,i}\sum_{k\neq i}^{m}\left(c^0_k-c^1_k(t)\right)
		\Biggr]W_{x_{i}}
		+
		\frac{1}{2}\left[\left(\pi_{i}(t)-\omega_{m,i}\sum_{k\neq i}^{m}\pi_k(t)\right)^2\sigma^2\right]W_{x_{i}x_{i}}
		\\
		&\quad
		+\lambda_{i}\mathbb{E}\left[W\left(t,x_{i}-l^{0}_{i}(Z^i)\right)-W(t,x_{i})\right]
		\\
		&\quad
		+
		\sum_{k\neq i}^{m}\lambda_{k} \mathbb{E}\left[W\left(t,x_{i}+\omega_{m,i} l^{0}_{k}(Z^k)\right)-W(t,x_{i})\right]
		\\
		&\quad+\lambda \mathbb{E}\left[W\left(t,x_{i}-l^{0}_{i}(Z^i)+\omega_{m,i}\sum_{k\neq i}^{m} l^{0}_{k}(Z^k)\right)-W(t,x_{i})\right].
		\\
		\mathcal{L}^{u_{R_j}, u_{R_j-}}W(t, y_j)&=W_t
			+\Big[
			ry_j(t)+\left(\mu-r\right)\pi_{R_j}(t)+\sum_{i=1}^{m}\left(c^{j}_{i}(t)-c^{j+1}_{i}(t)\right)
			\Big]W_{y_j}
			\\
			&
			+
			\frac{1}{2}\pi_{R_j}^2(t)\sigma^2
			W_{y_jy_j}
			+
			\sum_{i=1}^{m}\lambda_{i} \mathbb{E}\left[W\left(t,y_{i}-l^{j}_{i}(Z^i)\right)-W(t,y_{i})\right]
			\\
			&
			+
			\lambda \mathbb{E}\left[W\left(t,y_{j}-\sum_{i=1}^{m}l^{j}_{i}(Z^i)\right)-W(t,y_{j})\right].
		\end{align*}
	\begin{Theorem}
		(Verification Theorem)
		If real functions $w^{i}(t,x_{i}), g^{i}(t,x_{i}), w^{R_j}(t,y_j), g^{R_j}(t,y_j)$ satisfies the following extended HJB equation: 
		\begin{equation*}
			\left\{
			\begin{aligned}
				&
				\sup_{u_{i}\in\mathcal{A}_{i}}
				\left\{
				\mathcal{L}^{u_{i}, u_{i-}}
				w^{i}(t, x_{i})-\frac{\gamma_{i}}{2}\mathcal{L}^{u_{i}, u_{i-}}g^{i}(t, x_{i})^2+\gamma_{i}
				g^{i}(t, x_{i})\mathcal{L}^{u_{i}, u_{i-}}g^{i}(t, x_{i})
				\right\}=0,
				\\
				&
				w^{i}(T, x_{i})=x_{i},\quad  
				g^{i}(T, x_{i})=x_{i},\quad  
				\mathcal{L}^{u^*_{i}, u^*_{i-}}g^{i}(t, x_{i})=0.  
			\end{aligned}
			\right.
		\end{equation*}
		
		\begin{equation*}
			\left\{
			\begin{aligned}
				&
				\sup_{u_{R_j}\in\mathcal{A}_{R_j}}
				\Bigl\{
				\mathcal{L}^{u_{R_j}, u_{R_j-}}w^{R_j}(t, y_j)
				-\frac{\gamma_{R_j}}{2}\mathcal{L}^{u_{R_j}, u_{R_j-}}g^{R_{j}}(t, y_j)^2
				\\
				&\qquad\qquad
				+\gamma_{{R_j}}
				g^{R_{j}}(t, y_{j})\mathcal{L}^{u_{R_j}, u_{R_j-}}g^{R_j}(t, y_{j})
				\Bigr\}=0,
				\\
				&
				w^{R_j}(T, y_j)=y_j,\quad
				g^{R_j}(T, y_j)=y_j,\quad
				\mathcal{L}^{u^*_{R_j}, u^*_{R_j-}}g^{R_j}(t, y_j)=0. 
			\end{aligned}
			\right.
		\end{equation*}
		and satisfy
		  \begin{equation*}
			\begin{aligned}
			u_{i}^*=&\arg \sup_{u_{i}\in\mathcal{A}_{i}}
			\left\{
			\mathcal{L}^{u_{i}, u_{i-}}w^{i}(t, x_{i})-\frac{\gamma_{i}}{2}\mathcal{L}^{u_{i}, u_{i-}}g^{i}(t, x_{i})^2+\gamma_{i}
			g^{i}(t, x_{i})\mathcal{L}^{u_{i}, u_{i-}}g^{i}(t, x_{i})
			\right\}, 
			\\
			u_{R_j}^*=&\arg\sup_{u_{R_j}
				\in\mathcal{A}_{R_j}}
			\Bigr\{
			\mathcal{L}^{u_{R_j}, u_{R_j-}}w^{R_j}(t, y_{j})
			-
			\frac{\gamma_{R_j}}{2}\mathcal{L}^{u_{R_j}, u_{R_j-}}g^{R_{j}}(t, y_j)^2
			\\
			&\qquad\qquad\quad
			+
			\gamma_{{R_j}}
			g^{R_{j}}(t, y_{j})\mathcal{L}^{u_{R_j}, u_{R_j-}}g^{R_j}(t, y_{j})
			\Bigr\},
		 \end{aligned}
		\end{equation*}
		then $w^{i}(t,x_{i})=V^{i}(t,x_{i}),w^{R_j}(t,y_j)=V^{R_j}(t,y_j)$, $u^*(t)\in\mathcal{A}$ is the equilibrium strategy.
	\end{Theorem}
	\begin{proof}
		The proof follows the verification argument developed in \textcite{chen2019stochastic}.
	\end{proof}
	\subsection{The case of proportional reinsurance}
	Assume that both insurers and reinsurers adopt proportional reinsurance contracts to manage their risk exposures. The amount of retained claims in the reinsurance chain is given by
	\begin{equation*}
		l^{j}_{i}(Z^{i},t)=\left(q^{j}_{i}(t)-q^{j+1}_{i}(t)\right)Z^{i},
	\end{equation*} 
	where  $q^1_{i}(t)$ denotes the proportion of insurer  $i$'s risk transferred to the first reinsurer $R_{1}$, $q^{j+1}_{i}(t)$ 
	denotes the proportion of risk ceded by reinsurer $R_{j}$ to the next reinsurer $R_{j+1}$.  We set $q^0_{i}(t)=1$.
	
	The corresponding premium rates $c^{0}_{i}$ and $c^{j}_{i}(t)$ become
	
	\begin{equation*}
		\left\{
		\begin{aligned}
			&c^{0}_{i}=\left(1+\eta_{i}\right)\left(\lambda_{i}+\lambda\right)\mu_{i},
			\\
			&c^{j}_{i}(t)=\left(1+\theta^{j}_{i}(t)\right)\left(\lambda_{i}+\lambda\right)q^{j}_{i}(t)\mu_{i}.
		\end{aligned}
		\right.
	\end{equation*}
	For convenience, denote $o_{i}=\left(\lambda_{i}+\lambda\right)\mu_{i}$, $\sigma_{i}=\sqrt{\left(\lambda_{i}+\lambda\right)b_{i}^2}$.
	
	\begin{Theorem}
		\label{Theorem-Proportional reinsurance-equilibrium strategy}
		$\forall i,k\in\left\{1,2,\cdots, m\right\}$ satisfies $i\neq k$. Define m-dimensional matrices $\mathbf{\Psi}_{j}(t), \mathbf{H}_{j}(t), \mathbf{A}_{j}(t) \text{ and } \mathbf{B}_{j}(t)$,  and m-dimensional column vectors $\vec{C}_{j}$ and $\vec{G}_{j}(t)$, such that
		
		\begin{align*}
			& \Lambda^{j+1}_{i}(t) =
			\left\{
			\begin{aligned}
				&
				\frac{1}{\left(1+\omega_{m,i} \frac{\lambda \mu_{i}^2}{\sigma_{i}^2}\right)} 
				\cdot
				\frac{ o_i}{\sigma_{i}^2 \gamma_i \mathrm{e}^{r(T-t)}},\quad j=0,
				\\
				&
				\frac{1}{\left(1-\frac{\lambda \mu_{i}^2}{\sigma_{i}^2}\right)}
				\cdot
				\frac{o_i}{\sigma_{i}^2\gamma_{{R_j}}\mathrm{e}^{r(T-t)}},\quad j=1,\cdots,n-1.
			\end{aligned}
			\right.
			\\
			& \Phi^{j+1}_{i} =
			\left\{
			\begin{aligned}
				&
				\frac{ \omega_{m,i}
					\frac{\lambda \mu_{i}}{\sigma_{i}^2} }{\left(1+\omega_{m,i} \frac{\lambda \mu_{i}^2}{\sigma_{i}^2}\right)}
				,\quad j=0,
				\\
				&
				-
				\frac{1}{\left(1-\frac{\lambda \mu_{i}^2}{\sigma_{i}^2}\right)}
				\cdot
				\frac{\lambda \mu_{i}}{\sigma_{i}^2},\quad j=1,\cdots,n-1.
			\end{aligned}
			\right.
			\\
			& \Psi^{j+1}_{ii}(t) = \Lambda^{j+1}_{i}(t)
			+ \frac{ \mu_{i} \Phi^{j+1}_{i} \Lambda^{j+1}_{i}(t) }
			{1-\sum_{v=1}^{m} \mu_{v} \Phi^{j+1}_{v}} .
			\\
			& \Psi^{j+1}_{ik}(t) = \Phi^{j+1}_{i} 
			\frac{ \mu_{k} \Lambda^{j+1}_{k}(t)}
			{1-\sum_{v=1}^{m} \mu_{v} \Phi^{j+1}_{v}} .
			\\
			& \vec{G}_{j+1}(t) =
			\left\{
			\begin{aligned}
				&0, \quad j=0,
				\\
				&\vec{G}_{j}(t)+\mathbf{H}_{j}(t)\left(\mathbf{A}_{j}\right)^{-1}
				\vec{C}_{j}, \quad j=1, \cdots, n-1.
			\end{aligned}
			\right.
			\\
			& C^{j}_{i}(t) =
			\left\{
			\begin{aligned}
				&(1-G^{j}_{i}(t))o_i, \quad j=1,\cdots,n-1,
				\\
				&(1-G^{n}_{i}(t))o_i
				+
				\mathrm{e}^{r(T-t)}\gamma_{R_n}
				\Bigg[
				\left(1-G^{n}_{i}(t)\right)\left(H^{n}_{ii}(t)\sigma_i^2
				+
				\sum_{k\neq i}^{m}
				\rho_{ki} H^{n}_{ki}(t)\sigma_i \sigma_k\right)
				\\
				&
				+
				\sum_{k\neq i}^{m}
				\left(1-G^{n}_{k}(t)\right)
				\left(H^{n}_{ki}(t)\sigma_k^2+
				\sum_{v\neq k}^{m}\rho_{kv}H^{n}_{vi}(t)
				\sigma_k\sigma_v\right)
				\Bigg], \quad j=n.
			\end{aligned}
			\right.
			\\
			& \mathbf{H}_{j+1}(t) =
			\left\{
			\begin{aligned}
				&\mathbf{\Psi}_{1}(t), 
				\quad j=0,
				\\
				&
				\mathbf{H}_{j}(t)\left(\mathbf{A}_{j}\right)^{-1}\mathbf{B}_{j}+\mathbf{\Psi_{j+1}}(t), 
				\quad j=1,2,\cdots,n-1 .
			\end{aligned}
			\right.
			\\
			& A^{j}_{ii} =
			\left\{
			\begin{aligned}
				&2H^{j}_{ii}(t)o_i, \quad j=1, \cdots, n-1, 
				\\
				&2H^{n}_{ii}(t)o_i+
				\mathrm{e}^{r(T-t)}\gamma_{R_n}
				\Bigg[ H^{n}_{ii}(t)
				\left(H^{n}_{ii}(t)\sigma_i^2
				+
				\sum_{k\neq i}^{m}
				\rho_{ki} H^{n}_{ki}(t)\sigma_i \sigma_k\right)
				\\
				&
				+ 
				\sum_{k\neq i}^{m}
				H^{n}_{ki}(t)
				\left(H^{n}_{ki}(t)\sigma_k^2+
				\sum_{v\neq k}^{m}\rho_{kv}H^{n}_{vi}(t)
				\sigma_k\sigma_v\right)
				\Bigg],\quad j=n.
			\end{aligned}
			\right.
			\\
			& A^{j}_{ik} =
			\left\{
			\begin{aligned}
				&
				H^{j}_{ik}(t)o_{i}+H^{j}_{ki}(t)o_{k}, \quad j=1,\dots, n-1,
				\\
				&
				\left[
				H^{n}_{ik}(t)o_{i}+H^{n}_{ki}(t)o_{k}\right]
				+
				\mathrm{e}^{r(T-t)}\gamma_{R_n}
				\sum_{k \neq i}^{m} 				
				\Bigg[
				H^{n}_{ik}(t)
				\left(H^{n}_{ii}(t)\sigma_i^2
				+
				\sum_{k\neq i}^{m}
				\rho_{ki} H^{n}_{ki}(t)\sigma_i \sigma_k\right)
				\\
				&
				+
				\sum_{v \neq i}^{m} H^{n}_{vk}
				\left(H^{n}_{vi}(t)\sigma_v^2+
				\sum_{l\neq v}^{m}\rho_{vl}H^{n}_{li}(t)
				\sigma_l\sigma_v\right)
				\Bigg], \quad j=n
				.  
			\end{aligned}
			\right.
			\\
			& B^{j}_{ii} =
			\left\{
			\begin{aligned}
				&
				\mathrm{e}^{r(T-t)}\gamma_{R_j}
				\Bigg[ \Psi^{j+1}_{ii}(t)
				\left(H^{j}_{ii}(t)\sigma_i^2
				+
				\sum_{k\neq i}^{m}
				\rho_{ki} H^{j}_{ki}(t)\sigma_i \sigma_k\right)
				\\
				&
				+ 
				\sum_{k\neq i}^{m}
				\Psi^{j+1}_{ki}(t)
				\left(H^{j}_{ki}(t)\sigma_k^2+
				\sum_{v\neq k}^{m}\rho_{kv}H^{j}_{vi}(t)
				\sigma_k\sigma_v\right)
				\Bigg], \quad j=1, \cdots, n-1, 
				\\
				&
				0,\quad j=n.
			\end{aligned}
			\right.
			\\
			& B^{j}_{ik} =
			\left\{
			\begin{aligned}
				&
				\sum_{k \neq i}^{m} 				
				\Bigg[
				\Psi^{j+1}_{ik}(t)
				\left(H^{j}_{ii}(t)\sigma_i^2
				+
				\sum_{k\neq i}^{m}
				\rho_{ki} H^{j}_{ki}(t)\sigma_i \sigma_k\right)
				\\
				&
				+
				\sum_{v \neq i}^{m} \Psi^{j+1}_{vk}
				\left(H^{j}_{vi}(t)\sigma_v^2+
				\sum_{l\neq v}^{m}\rho_{vl}H^{j}_{li}(t)
				\sigma_l\sigma_v\right)
				\Bigg], \quad j=1,\dots, n-1,
				\\
				&
				0
				, \quad j=n
				.  
			\end{aligned}
			\right.
		\end{align*}

		The  equilibrium investment strategy of insurers and reinsurers are given by

		\begin{equation*}
			\left\{
			\begin{aligned}
				&\pi_{i}^*(t)=
				\left[
				\frac{
					\frac{\omega_{m,i}}{1+\omega_{m,i}}\sum_{v=1}^{m}
					\frac{1}{(1+\omega_{m,v})\gamma_{v}}
				}
				{1-\sum_{v=1}^{m}\frac{\omega_{m,v}}{1+\omega_{m,v}}}
				+
				\frac{1}{(1+\omega_{m,i})\gamma_{i}}
				\right]
				\cdot
				\frac{\mu-r}{\mathrm{e}^{r(T-t)}\sigma^2},
				\\
				&\pi_{R_j}^*(t)=\frac{\mu-r}{\gamma_{R_{j}}\mathrm{e}^{r(T-t)}\sigma^2}. 
			\end{aligned}
			\right. 
		\end{equation*}
		
		The equilibrium reinsurance strategy of reinsurer $R_j$ is 
		\begin{equation*}
			\begin{aligned}
				\vec{q}_{j+1}(t)^*=\mathbf{1}-\vec{G}_{j+1}(t)-\mathbf{H}_{j+1}(t)\vec{\theta}_{j+1}(t). 
			\end{aligned}
		\end{equation*}

		The equilibrium safety loading strategy of reinsurer $R_j$ is as follows.
		\begin{equation*}
			\vec{\theta}_{j}(t)^*=\left(\mathbf{A}_{j}\right)^{-1}\vec{C}_{j}
			+
			\left(\mathbf{A}_{j}\right)^{-1}\mathbf{B}_{j}\vec{\theta}_{j+1}(t).		\end{equation*}
	\end{Theorem}
	\begin{proof}
		 See Appendix \ref{proof-Proportional reinsurance-equilibrium strategy}. 
	\end{proof}	

	In contrast to existing studies restricted to single-layer reinsurance structures, the  equilibrium strategies established in the theorem above show that the competitive interactions among insurers do not terminate at the initial reinsurance stage. Instead, such interactions propagate along the reinsurance chain, inducing persistent modifications in the structural configuration and equilibrium characterization of subsequent reinsurance contracts.

	\subsection{The case of excess-of-loss reinsurance}
This section considers the case where insurers and reinsurers adopt excess-of-loss contracts. The retained claim amount $l^{j}_{i}(Z^{i},t)$ becomes
	\begin{equation}
		l^{j}_{i}(Z^{i},t)=\min\left(Z^{i}-\sum_{k=0}^{j-1}l^k_{i}(Z^i,t), a^{j+1}_{i}(t)\right),
	\end{equation}
	where $a^{j}_i(t)$ is  the retention level of reinsurer $R_{j-1}$, and $a^{1}_i(t)$ is  the retention level of insurer $i$.  
	
	The premium rates $c^{0}_{i}$ and $c^{j}_{i}(t)$ become
	\begin{equation*}
		\left\{
		\begin{aligned}
			&c^{0}_{i}=\left(1+\eta_{i}\right)\left(\lambda_{i}+\lambda\right)\mu_{i},
			\\
			&c^{j}_{i}(t)=\left(1+\theta^{j}_{i}(t)\right)\left(\lambda_{i}+\lambda\right)\mathbb{E}\left[Z^i-\sum_{k=0}^{j-1}l^k_{i}(Z^i)\right].
		\end{aligned}
		\right.
	\end{equation*}
	
	With reference to \textcite{li2022equilibrium},
	\begin{equation*}
		\begin{aligned}
			\mathbb{E}\left[\min\left(Z^{i},a_1(t)\right)\right]&=\int_{0}^{a}(1-F(z))\mathrm{d}z=\int_{0}^{a}\overline{F}(z)\mathrm{d}z,
			\\
			\mathbb{E}\left[\min(Z^{i}, a_1(t))^2\right] &= \int_{0}^{a}2z\overline{F}(z)\mathrm{d}z,
		\end{aligned} 	
	\end{equation*}
	where $\overline{F}(z)=1-F(z)$. For ease of notations, denote $h^{j}_{i}(t)=\sum_{k=1}^{j}a^k_{i}(t)$. Unlike the case with only one reinsurer, 
	$l^{j-1}_{i}(Z^{i})(j=2, \cdots, n)$ satisfies the following lemma:
	\begin{Lemma}
		\label{Lemma 3.1} For each \(j = 1, \dots, n\),
		\begin{equation*}
			\begin{aligned}
				\mathbb{E}\left[l^{j}_{i}(Z^i, t)\right]
				&=\int_{h^{j}_{i}(t)}^{h^{j+1}_{i}(t)}\overline{F}(z)\mathrm{d}z,
				\\
				\mathbb{E}\left[l^{j}_{i}(Z^i, t)^2\right]
				&=
				\int_{h^{j}_{i}(t)}^{h^{j+1}_{i}(t)}2\left(z-h^{j}_{i}(t)\right)\overline{F}(z)
				\mathrm{d}z.
			\end{aligned}
		\end{equation*}
	\end{Lemma}
	\begin{proof}
		See Appendix \ref{proof of Lemma 3.1}.
	\end{proof}

	\begin{Theorem}
		\label{Theorem-excess-of-loss reinsurance-n reinsurer}
		Under an excess‑of‑loss reinsurance arrangement, the equilibrium investment strategies for both insurers and reinsurers coincide with those given in  \cref{Theorem-Proportional reinsurance-equilibrium strategy}. The equilibrium retention levels and safety loadings are characterized as follows.
        
        (i) The equilibrium  retention level $a^{1}_{i}(t)^*$ of insurer $i$ is the solution to the following equation
		\begin{equation*}
			a^{1}_{i}(t)^*=\frac{\theta^{1}_{i}(t)}{\gamma_{i}\mathrm{e}^{r(T-t)}}+\frac{\lambda \omega_{m,i}}{\lambda+\lambda_{i}}
			\sum_{k\neq i}^{m}
			\int_{0}^{a^{1}_{k}(t)}\overline{F}_{k}(z)\mathrm{d}z.
		\end{equation*}
		
		(ii) The equilibrium  retention level $a^{j+1}_{i}(t)^*$ of reinsurer $R_{j}$ is the solution to
		\begin{equation*}
			\begin{aligned}
				a^{j+1}_{i}(t)^*
				&=\frac{\theta^{j+1}_{i}(t)}{\gamma_{R_j}\mathrm{e}^{r(T-t)}}
				-
				\frac{\lambda}{\lambda+\lambda_{i}}\int_{h^{j}_{k}(t)}^{h^{j+1}_{k}(t)}\overline{F}_{k}(z)\mathrm{d}z.
			\end{aligned}
		\end{equation*}
		
		(iii) The safety loading  $\theta^{j}_{i}(t)$ of reinsurer $R_{j} (j=1, 2, \cdots, n-1)$ is determined by
		\begin{align*}
			&\Bigg[
			(\lambda+\lambda_{i})\mu_{i}-(\lambda+\lambda_{i})\int_{0}^{h^{j}_{i}}\overline{F}_{i}(z)\mathrm{d}z
			\\
			&
			-
			\sum_{v=1}^{m}\theta^{j}_{v}(t)(\lambda+\lambda_{v})
			\overline{F}_{v}(h^{j}_{v})
			\frac{\partial h^{j}_{v}}{\partial \theta^{j}_{i}(t)}
			+
			\sum_{v=1}^{m}
			\theta^{j+1}_{v}(t)(\lambda+\lambda_{v})\overline{F}_{v}(h^{j+1}_{v})
			\frac{\partial h^{j}_{v}}{\partial \theta^{j}_{i}(t)}
			\Bigg]
			\\
			&
			-\mathrm{e}^{r(T-t)}\gamma_{R_j}
			\sum_{v=1}^{m}(\lambda_{v}+\lambda)
			\left[
			a^{j+1}_{v}\overline{F}_{v}(h^{j+1}_{v})
			-\int_{h^{j}_{v}}^{h^{j+1}_{v}}\overline{F}_{v}(z)\mathrm{d}z
			\right]\frac{\partial h^{j}_{v}}{\partial \theta^{j}_{i}(t)}
			\\
			&
			-\mathrm{e}^{r(T-t)}\gamma_{R_j}\lambda
			\Bigg[
			\sum_{k<p}^{m}
			\left(
			\overline{F}_{k}(h^{j+1}_{k})-\overline{F}_{k}(h^{j}_{k})\mathrm{d}z
			\right)
			\frac{\partial h^{j}_{k}}{\partial \theta^{j}_{i}(t)}
			\int_{h^{j}_{p}}^{h^{j+1}_{p}}\overline{F}_{p}(z)\mathrm{d}z
			\\
			&
			+
			\sum_{k<p}^{m}
			\int_{h^{j}_{k}}^{h^{j+1}_{k}}\overline{F}_{k}(z)\mathrm{d}z
			\left(
			\overline{F}_{p}(h^{j+1}_{p})-\overline{F}_{p}(h^{j}_{p})\mathrm{d}z
			\right)
			\frac{\partial h^{j}_{p}}{\partial \theta^{j}_{i}(t)}
			\Bigg]
			=0.
		\end{align*}
		
		(iv) The safety loading $\theta^{n}_{i}(t)$ of reinsurer $R_{n}$ is the solution to
		\begin{align*}
			&\Bigg[
			(\lambda+\lambda_{i})\mu_{i}-(\lambda+\lambda_{i})\int_{0}^{h^{n}_{i}}\overline{F}_{i}(z)\mathrm{d}z
			-
			\sum_{v=1}^{m}\theta^{n}_{v}(t)(\lambda+\lambda_{v})
			\overline{F}_{v}(h^{n}_{v})
			\frac{\partial h^{n}_{v}}{\partial \theta^{n}_{i}(t)}
			\Bigg]
			\\
			&
			-\mathrm{e}^{r(T-t)}\gamma_{R_1}
			\sum_{v=1}^{m}(\lambda_{v}+\lambda)
			\left[
			-
			\int_{h^{n}_{v}}^{\infty}\overline{F}_{v}(z)\mathrm{d}z
			\right]\frac{\partial h^{n}_{v}}{\partial \theta^{n}_{i}(t)}
			\\
			&
			-
			{e}^{r(T-t)}\gamma_{R_1}\lambda
			\Bigg[
			\sum_{k<p}^{m}
			\left(
			-\overline{F}_{k}(h^{n}_{k})
			\right)
			\frac{\partial h^{n}_{k}}{\partial \theta^{n}_{i}(t)}
			\int_{h^{n}_{p}}^{\infty}\overline{F}_{p}(z)\mathrm{d}z
			\\
			&
			+
			\sum_{k<p}^{m}
			\left(
			-\overline{F}_{p}(h^{n}_{p})
			\right)
			\frac{\partial h^{n}_{k}}{\partial \theta^{n}_{i}(t)}
			\int_{h^{n}_{k}}^{\infty}\overline{F}_{k}(z)\mathrm{d}z
			\Bigg]
			\\
			&
			=0.
		\end{align*}
	\end{Theorem}

	\begin{proof}
		See Appendix \ref{proof-excess-of-loss reinsurance-n reinsurer}. 
	\end{proof}
	
	Given the analytical challenges in establishing the existence and uniqueness of explicit solutions for excess-of-loss reinsurance within complex reinsurance chain structures, \Cref{Excess-of-loss reinsurance-economy analysis} introduces a simplified chain framework for further analysis.
	
\section{Economic implications}

    This section presents explicit forms of the equilibrium strategies for the special cases considered in \Cref{Theorem-Proportional reinsurance-equilibrium strategy} and \Cref{Theorem-excess-of-loss reinsurance-n reinsurer}. Under the assumption that all participants in the reinsurance chain behave as rational agents, we then conduct a systematic analysis of its structural characteristics and intrinsic properties, which reveals the mechanisms and dynamic patterns that govern its operation in equilibrium. 

    \subsection{Proportional reinsurance
    }

   The following analysis focuses on the structural characteristics and intrinsic properties of the reinsurance chain under proportional reinsurance.

    \subsubsection{The case of a single insurer}
    Assume only insurer $1$ exists in the market, the wealth process is as follows.
    \begin{equation*}
    	\begin{aligned}
    		\mathrm{d}X_{1}(t)
    		&=\left[(\mu-r)\pi_{1}(t)+rX_{1}(t)
    		+  \left(1+\eta_{1}\right)o_{1}
    		- \left(1+\theta_{1}(t)\right)q_{1}(t)o_{1}
    		\right] \mathrm{d}t 
    		\\
    		&\quad
    		-
    		\mathrm{d}\left( \sum_{k=1}^{N_{1}(t)+N(t)}(1-q_{1}(t))Z^{1}_{k} \right)
    		+\pi_{1}(t)\sigma\mathrm{d}B(t). 
    	\end{aligned}
    \end{equation*}
    
    The wealth process  $Y_j(t)$ of reinsurer $R_j$ satisfies the following stochastic differential equation
    \begin{equation*}
    	\begin{aligned}
    		\mathrm{d}Y_j(t)
    		&=\left[rY_j(t)+\left(\mu-r\right)\pi_{R_j}(t)
    		+\left(1+\theta_{j}(t)\right)q_{j}(t)o_{1}
    		-\left(1+\theta_{j+1}(t)\right)q_{j+1}(t)o_{1}
    		\right]\mathrm{d}t
    		\\
    		&\quad-\mathrm{d}\left(\sum_{k=1}^{N_{1}(t)+N(t)}
    		\left(q_{j}(t)-q_{j+1}(t)\right)Z^{1}_{k}\right)
    		+\pi_{R_j}(t)\sigma\mathrm{d}B(t). 
    	\end{aligned}
    \end{equation*}

    \begin{Theorem}
    	\label{one insurer-theorem}
    	In the special case where the market contains a single insurer, $H_{j+1}(t)$ and $G_{j+1}(t)$ reduce to 
    	\begin{equation}
    		\label{one insurer-H}
    		H_{j+1}(t)=\left\{
    		\begin{aligned}
    			&
    			\frac{o_1}{\gamma_1\mathrm{e}^{r(T-t)}\sigma_1^2}
    			,       \quad j=0,
    			\\
    			&\left[\frac{1}{2^{j}\gamma_{1}}+\sum_{k=1}^{j}\frac{1}{2^{j-k}\gamma_{R_k}}\right]\frac{o_1}{\mathrm{e}^{r(T-t)}\sigma_1^2},\quad j=1, 2, \dots, n-1.
    		\end{aligned}
    		\right. 
    	\end{equation}
    	
    	\begin{equation*}
    		G_{j+1}(t)=
    		\frac{2^{j}-1}{2^{j}}, \quad j=0, 1, 2, \dots, n-1.
    	\end{equation*}
    	
    	The equilibrium investment strategies of insurers and reinsurers are
    	\begin{equation*}
    		\left\{
    		\begin{aligned}
    			\pi_1^*(t)&=
    			\frac{\mu-r}{\gamma_1
    				\mathrm{e}^{r(T-t)}\sigma^2},
    			\\
    			\pi_{R_j}^*(t)&=\frac{\mu-r}{\gamma_{R_{j}}\mathrm{e}^{r(T-t)}\sigma^2}. 
    		\end{aligned}
    		\right. 
    	\end{equation*}

    	The equilibrium reinsurance strategy is
    	\begin{equation}
    		\label{one insurer-reinsurance}
    		q_{j+1}^*(t)=
    		\begin{aligned}
    			&\frac{1}{2^{j}}
    			-H_{j+1}(t)
    			\theta_{j+1}(t), \quad j=0, 1, \dots, n-1.  
    		\end{aligned}
    	\end{equation}

    	The equilibrium safety loading of reinsurer $R_j$ is
    	\begin{equation}
    		\label{one insurer-safety loadings}
    		\theta_{j}^*(t)=
    		\begin{aligned}
    			\frac{1}{2^{j}}\cdot\frac{1}{H_{j}(t)}+\frac{\theta_{j+1}(t)}{2}, \quad j=1, 2, \dots, n-1. 
    		\end{aligned}
    	\end{equation}
    	\begin{equation}
    		\label{one insurer-safety loadings-R_{n}}
    		\begin{aligned}
    			\theta_{n}^*(t)=&\frac{\left(1-G_{n}(t)\right)o_1+\mathrm{e}^{r(T-t)}\gamma_{{R_n}}
    				\left(1-G_{n}(t)\right)H_{n}(t)\sigma_1^2
    			}
    			{2H_{n}(t)o_{1}+\mathrm{e}^{r(T-t)}\gamma_{{R_n}}
    				H_{n}(t)^2\sigma_1^2}.
    		\end{aligned}
    	\end{equation}	
    	
    \end{Theorem}
    \begin{proof}
    	The proof follows directly from \cref{Theorem-Proportional reinsurance-equilibrium strategy}.
    \end{proof}
     Assuming that the reinsurers adopt the variance premium principle, and the wealth process of insurer $1$ satisfies the following stochastic differential equation：
    \begin{equation*}
    	\begin{aligned}
    		\mathrm{d}X_1(t)
    		&=\left[(\mu-r)\pi_1(t)+rX_1(t)+  \left(\eta_1o_1-\theta_1(t)q_{1}(t)^2\sigma_1^2\right) \right] \mathrm{d}t
    		\\
    		&\quad 
    		+ \left(1-q_{1}(t)\right)\sigma_1\mathrm{d}W_1(t)
    		+\pi_1(t)\sigma\mathrm{d}B(t). 
    	\end{aligned}
    \end{equation*}
    
   The wealth process of reinsurer $R_{j}$ satisfies    
    \begin{equation*}
    	\begin{aligned}
    		\mathrm{d}Y_j(t)
    		&=\left[rY_j(t)+\left(\mu-r\right)\pi_{R_j}(t)
    		+\left(\theta_{j}(t)q_{j}(t)^2-\theta_{j+1}(t)q_{j+1}(t)^2\right)\sigma_1^2
    		\right]\mathrm{d}t
    		\\
    		&\quad+\left(q_{j}(t)-q_{j+1}(t)\right)\sigma_1\mathrm{d}W_1(t)
    		+\pi_{R_j}(t)\sigma\mathrm{d}B(t). 
    	\end{aligned}
    \end{equation*}	
    
    \begin{Theorem}
    	\label{Theorem-one insurance-Var}
    	Under the variance premium principle, the equilibrium investment strategy remains unaffected and takes the same form as in the general equilibrium setting. Define
    	\begin{equation*}
    		H_j(t)=\frac{1}{\mathrm{e}^{r(T-t)}\gamma_{1}} + \sum_{k=1}^{j-1}\frac{1}{\mathrm{e}^{r(T-t)}\gamma_{R_k}}.
    	\end{equation*}
    	
    	The equilibrium reinsurance strategy is 
    	
    	\begin{equation*}
    		q_j(t)=\frac{1}{2^{j-1}}\cdot\frac{1}{2H_{j}(t)\theta_j(t)+1}, \quad j=1, \cdots, n.
    	\end{equation*}
    	
    	The equilibrium safety loading strategy is given by
    	
    	\begin{equation*}
    		\left\{
    		\begin{aligned}
    			&\theta_{j}^*(t)=\frac{1}{2H_{j}(t)}+
    			\frac{1}{ 
    				\frac{1}{\gamma_{R_j}\mathrm{e}^{r(T-t)}}
    				+\frac{1}{2\theta_{j+1}(t)} 
    			}
    			\quad j=1, \dots, n-1,
    			\\
    			&\theta_{n}^*(t)=\frac{1}{2H_{n}(t)}+\gamma_{R_n}\mathrm{e}^{r(T-t)}. 
    		\end{aligned}
    		\right.
    	\end{equation*}
    	
    \end{Theorem}
     \begin{proof}
     	The proof follows directly from \cref{Theorem-Proportional reinsurance-equilibrium strategy}.
     \end{proof}

    \subsubsection{Ordering properties of equilibrium strategies}
    In practice, leveraging informational and scale advantages, leading reinsurance firms often offer more favorable retrocession contracts to reinsurers to improve profitability and promote market stability. For analytical tractability, this subsection examines the ordering properties of equilibrium strategies within the framework of a single insurer. The following analysis relies on the explicit expressions derived in \Cref{one insurer-theorem}. 
    
    \begin{Theorem}
    	\label{Theorem 4.1}
    	The equilibrium reinsurance strategy satisfies $q_{j}^*(t)>q_{j+1}^*(t)$, $j= 1, \dots, n-1$. 
    \end{Theorem}
    
    \begin{proof}
    	
    	 Under the single-insurer market assumption, equation (\ref{Reinsurance ratio-recursion formula}) reduces to 
    	\begin{equation}
    		\begin{aligned}
    			q_{j+1}(t)&=
    			q_{j}(t)
    			-\frac{o_1}{\gamma_{R_j}\sigma_1^2\mathrm{e}^{r(T-t)}}\cdot
    			\theta_{j+1}(t).
    		\end{aligned}
    	\end{equation}
    	Since the second term on the right-hand side is strictly positive, the inequality $q_{j}^*(t)> q_{j+1}^*(t)$ follows immediately.
    	
        \end{proof}
        
   Based on the expected premium principle, reinsurers achieve an equilibrium between risk and return when determining their safety loading strategies. To attract downstream reinsurers, upstream reinsurers therefore need to offer more favorable reinsurance terms.
   
    \begin{Theorem}
    	The equilibrium safety loading strategy satisfies $\theta_{j}^*(t)>\theta_{j+1}^*(t)$, $j=1, \dots, n-1$.
    \end{Theorem}

        \begin{proof}
    	
    	From \eqref{one insurer-reinsurance} we observe that $q_{j}^*(t)$ is strictly decreasing in $\theta_{j}^*(t)$. 
    	Suppose, for contradiction, that $\theta_{j}^*(t) \leq \theta_{j+1}^*(t)$. 
    	Then there exists $\varepsilon \geq 0$ such that $\theta_{j+1}^*(t) = \theta_{j}^*(t) + \varepsilon$. 
    	Substituting this relation into (\ref{one insurer-safety loadings}) yields 
    	$\theta_{j}^*(t) = \frac{1}{2^{j-1}H_{j}(t)} + \varepsilon$. 
    	Inserting this expression for $\theta_{j}^*(t)$ back into (\ref{one insurer-reinsurance}) gives 
    	$q_{j}(t) = -H_{j}(t) \varepsilon \leq 0$, 
    	which contradicts the requirement that the retained proportion be strictly positive. 
    	Hence, the strict ordering $\theta_{j}^*(t) > \theta_{j+1}^*(t)$ must hold.
    	
    \end{proof}

    \begin{Corollary}
    	For all $j \in \left\{1, 2, \cdots, n\right\}$, if any risk aversion coefficient in the set  $\big\{\gamma_1, \gamma_{R_1},$ $ \cdots ,\gamma_{R_n}\big\}$ increases, the corresponding optimal safety loading $\theta_j(t)^*$ also increases.
    \end{Corollary}
    \begin{proof}
    	From (\ref{one insurer-H}), we note that $H_{j+1}(t)$ depends on the values of$\left\{\gamma_1, \gamma_{R_1}, \cdots, \gamma_{R_j}\right\}$, and is independent of    $\left\{\gamma_{R_{j+1}}, \cdots, \gamma_{R_n}\right\}$. 
    	 Moreover, 
    	 $H_{j+1}(t)$ is inversely proportional to each risk aversion coefficient in $\left\{\gamma_1, \gamma_{R_1}, \cdots, \gamma_{R_j}\right\}$.

    	First, consider $\gamma_{R_n}$, Only $\theta_{n}(t)$ is directly related to $\gamma_{R_n}$. From ($\ref{one insurer-safety loadings}$),
    	\begin{equation*}
    		\theta_{n}^*(t)=\frac{1}{2^{n-1}H_{n}(t)}\left(1-\frac{o_{1}}{2o_{1}+\mathrm{e}^{r(T-t)}\gamma_{R_n}H_{n}(t)\sigma_{1}^2}\right), 
    	\end{equation*}
    	where $\theta_{n}(t)$ is directly proportional to $\gamma_{R_{n}}$. Due to the direct proportion  relationship among the safety loading strategies, $\forall j\in\left\{1,2,\cdots,n\right\}, \theta_{j}(t)$ is directly proportional to $\gamma_{R_{n}}$.
    	 
    	  Consider $\gamma_{R_{n-1}}$, $\left\{H_{1}(t), H_{R_{1}}(t), \cdots, H_{R_{n-1}}(t)\right\}$ is not related to it, $H_{n}(t)$ is inversely proportional to $\gamma_{R_{n-1}}$. 
    	  
    	  \begin{equation*}
    	  	\frac{\partial \theta_{n}(t)}{\partial H_{n}(t)}
    	  	=
    	  	\frac{\mathrm{e}^{r(T-t)}\gamma_{R_n}H_{n}(t)\sigma_{1}^2o_1-\left(o_{1}+\mathrm{e}^{r(T-t)}\gamma_{R_n}H_{n}(t)\sigma_{1}^2\right)\left(2o_{1}+\mathrm{e}^{r(T-t)}\gamma_{R_n}H_{n}(t)\sigma_{1}^2\right)}
    	  	{2^{n-1}H_{n}(t)^2\left(2o_{1}+\mathrm{e}^{r(T-t)}\gamma_{R_n}H_{n}(t)\sigma_{1}^2\right)^2},
    	  \end{equation*}
    	It follows that $\frac{\partial \theta_{n}(t)}{\partial H_{n}(t)}<0$. Therefore,  $\theta_{n}(t)$ is directly proportional to $\gamma_{R_{n-1}}$. Due to the direct proportion  relationship among the safety loading strategies, for any $j\in\left\{1, 2, \cdots,n\right\}, \theta_{j}(t)$ is directly proportional to $\gamma_{R_{n-1}}$. 
    	
    	Similarly, it can be concluded that 
    	$\forall j\in\left\{1, 2, \cdots,n\right\}, \theta_{j}$ is directly proportional to  $\left\{\gamma_{1},\gamma_{R_{1}},\cdots,
    		\gamma_{R_{n}}\right\}$.
    \end{proof}

     \subsubsection{Optimal structure of reinsurance chains}
 In the preceding analysis of the reinsurance chain, the positions of reinsurers are assumed to be predetermined. We now relax this assumption to investigate how equilibrium strategies vary with the chain structure. For analytical simplicity, this subsection studies the optimal structure of the reinsurance chain under the expected value premium principle in the case of a single insurer. According to \cref{reinsurance-value function}, all reinsurers  follow a similar optimization problem, with the only difference being the risk aversion coefficient $\gamma_{R_j}$. Consequently, the problem reduces to identifying the appropriate ordering of these coefficients to attain the optimal chain structure. By treating the reinsurance chain as an integrated whole, the optimal chain is such that it achieves the highest possible risk-bearing capacity and imposes the lowest safety loading on the insurers.
      
      Assume that the risk aversion coefficients $\left\{\gamma_1, \gamma_{{R_1}}, \cdots, \gamma_{R_j}, \gamma_{R_n}\right\}$ are heterogeneous. To compare different chain structures, we construct two chains. Let \( l \in \{1, 2\} \) and \( j \in \{1, 2, \dots, n\} \). Denote by $\gamma_{R_{lj}}$ the risk aversion coefficient of the \(j\)-th reinsurer in the \(l\)-th chain. For any fix $k\in\left\{2,3,\cdots n\right\}$,  the coefficients are specified as follows:     $\gamma_{1j}=\gamma_{2j}=\gamma_j(j=1, 2, \cdots, k-2, k+1, \cdots, n)$, $\gamma_{1{k-1}}=\gamma_{2k}=\gamma_{k-1}$, $\gamma_{1{k}}=\gamma_{2k-1}=\gamma_{k}$.  The first chain has the same structure as the original reinsurance chain. Meanwhile, the second chain is constructed by interchanging the positions of the 
      	$k$-th and the 
      	$(k-1)$-th reinsurers in the original reinsurance chain.
      
      \begin{Theorem}
      	\label{Theorem-ChainChange}
      	Fix $k\in\left\{2, 3,\cdots ,n\right\}$, $\forall j \in\left\{1, 2, \cdots, k\right\}$,   $\theta_{1j}(t)<\theta_{2j}(t)$, if and only if $\gamma_{R_{k-1}}<\gamma_{R_k}$. In this case, $\forall j\in\left\{k+1, \cdots, n\right\}$, $\theta_{1j}(t)>\theta_{2j}(t)$.
      \end{Theorem}
      \begin{proof}
        	See Appendix \ref{proof-ChainChange}. 
      \end{proof}
     
     The entire reinsurance chain can be viewed as a whole, with its sole output being the strategy $\theta_{1}(t)$. $\theta_{1}(t)$ represents the overall risk-bearing capacity of the chain. The lower $\theta_{1}(t)$ is, the greater the reinsurance chain's risk-taking capacity is.
     If an optimal structure of the reinsurance chain exists such that the reinsurance capacity is maximized, i.e., the safety loading strategy $\theta_{1}(t)$ is minimized, then the optimal structure of the reinsurance chain is as follows:
      \begin{Theorem}
      	\label{Theorem-optimal reinsurance chain structure}
      	Under the setting of \cref{one insurer-theorem}, the optimal structure of the reinsurance chain satisfies $\gamma_{R_1}<\gamma_{R_2}<\cdots<\gamma_{R_n}$.
      \end{Theorem}
      \begin{proof}
      See Appendix \ref{proof of optimal reinsurance chain structure}. 
      \end{proof}

     In contrast to the variance‑premium framework of \textcite{chen2020continuous}, which adopts the variance premium principle, the adoption of the expected value premium principle leads to a distinctive outcome: when the risk aversion coefficient satisfies
      $\gamma_{R_{k-1}}>\gamma_{R_k}$
      , exchanging the positions of  reinsurer $R_{k}$ and  reinsurer $R_{k-1}$ results in an increase in the equilibrium safety loading strategy $\theta_{j}^{*}(t)$ for all  $j\in\left\{k+1,\cdots,n\right\}$. Therefore, in the optimal reinsurance chain, strategy  
      $\theta_{1}^{*}(t)$ 
      achieves the minimum and strategy 
      $\theta_{n}^{*}(t)$ 
      achieves the maximum among all admissible reinsurance chains.

      \begin{Remark}
      	It is assumed that each reinsurer has the discretion to select the subsequent-tier reinsurer to accept its reinsurance business, such that the final structure of the reinsurance chain still satisfies:
      	
      	\begin{equation*}
      		\gamma_{R_1}<\gamma_{R_2}<\cdots<\gamma_{R_n}.
      	\end{equation*}
      	
      	Fix $ k \in \{1, \cdots, n-1\} $, and assume that the risk aversion coefficients of the reinsurance companies $ R_j $ for $ j \in \{1, \cdots, k\} $ are already determined. In order to obtain the minimum safety load $ \theta_{k+1}(t) $ and maximize its value function, reinsurance company $ R_k $ will choose a reinsurance company $ R_{k+1} $ whose risk aversion coefficient $ \gamma_{R_{k}} $ satisfies:
      	\[
      	\gamma_{R_{k+1}} < \gamma_l \quad \text{for} \quad l \in \{k+2, \cdots, n\}.
      	\]
      \end{Remark}
      
     \subsection{Excess-of-loss reinsurance}
       \label{Excess-of-loss reinsurance-economy analysis}
       This section examines  a reinsurance chain consisting of the insurer $1$ and two reinsurers under excess of loss reinsurance.
       Assume $Z^1$ follows an exponential distribution with parameter $\delta_1$, denote $\gamma_{R_{0}}=\gamma_1$. At this point, the equilibrium strategy of the reinsurance chain is as follows:
       \begin{Theorem}
       	\label{Theorem-excess-of-loss reinsurance-two reinsurer}
       	The equilibrium excess-of-loss retention level	is as follows.
       	\begin{equation*}
       		a_{j}^*(t)=\frac{\theta_j(t)}{\gamma_{R_{j-1}}\mathrm{e}^{r(T-t)}}.
       	\end{equation*}
       	
       	The  equilibrium safety loading strategy $\theta_1(t)$ of reinsurer $R_1$ is
       	\begin{equation*}
       		\theta_{1}^*(t)=\frac{\gamma_1\mathrm{e}^{r(T-t)}}{\delta_1}
       		+\theta_2(t)\mathrm{e}^{-\delta_1a_2}
       		-\gamma_{R_1}\mathrm{e}^{r(T-t)}
       		\left(a_2\mathrm{e}^{-\delta_1a_2}
       		-\frac{(1-\mathrm{e}^{-\delta_1a_2})}{\delta_1}\right).
       	\end{equation*}
       	
       	The equilibrium equilibrium safety loading strategy $\theta_2(t)$ of reinsurer $R_2$ is the unique solution to the following equation.
       	\begin{equation*}
       		\label{excess-of-loss-safe load}
       		\begin{aligned}
       			\frac{1}{\delta_1}
       			+\left(\frac{\gamma_{R_{2}}\mathrm{e}^{r(T-t)}}{\delta_{1}}-\theta_{2}(t)\right)
       			\frac{\partial h_2}{\partial \theta_2}=0, 
       		\end{aligned}
       	\end{equation*}
       	where $\frac{\partial h_2}{\partial \theta_2}=
       	\frac{\mathrm{e}^{-\delta_1a_2}}{\gamma_{1}\mathrm{e}^{r(T-t)}}
       	+\frac{1}{\gamma_{R_{1}}\mathrm{e}^{r(T-t)}}$.

       \end{Theorem}
       \begin{proof}
       	See	Appendix \ref{proof of Excess-of-loss reinsurance-economy analysis}.
       \end{proof}

	\section{Numerical analysis}
	
	 \subsection{The case under proportional reinsurance
	}
	This section provides a numerical illustration of the effects of relevant parameter on 
	equilibrium reinsurance-investment strategy of the reinsurance chain between two insurers and five reinsurers 
	under proportional reinsurance. 
	According to \textcite{hao2022} and \textcite{li2023optimal}, unless otherwise stated, the relevant parameters are set as follows:
	$\gamma_{1}=\gamma_{2}=0.8, w_{1}=0.6,  w_{2}=0.8, \gamma_{R_{1}}=0.7,  \gamma_{R_{2}}=0.6, \gamma_{R_{3}}=0.5,  \gamma_{R_{4}}=0.4, \gamma_{R_{5}}=0.3,  \gamma_{R_{6}}=0.2,  \lambda=1.5, \lambda_{1}=3, \lambda_{2}=4,  \mu_{1}=\mu_{2}=0.1, b_{1} =b_{2}=0.2, r=0.07,  \mu=0.2, \sigma=0.5, t=5, T=10$.

	  \subsubsection{The impact of swapping the positions of reinsurer in the reinsurance chain}
	 In this section, we examine the impact of swapping the positions of reinsurers $R_2$ and $R_3$ within the reinsurance chain.

	 \begin{figure}[h]
	 	\begin{subfigure}[b]{0.5\textwidth}
	 		\includegraphics[width=\textwidth]{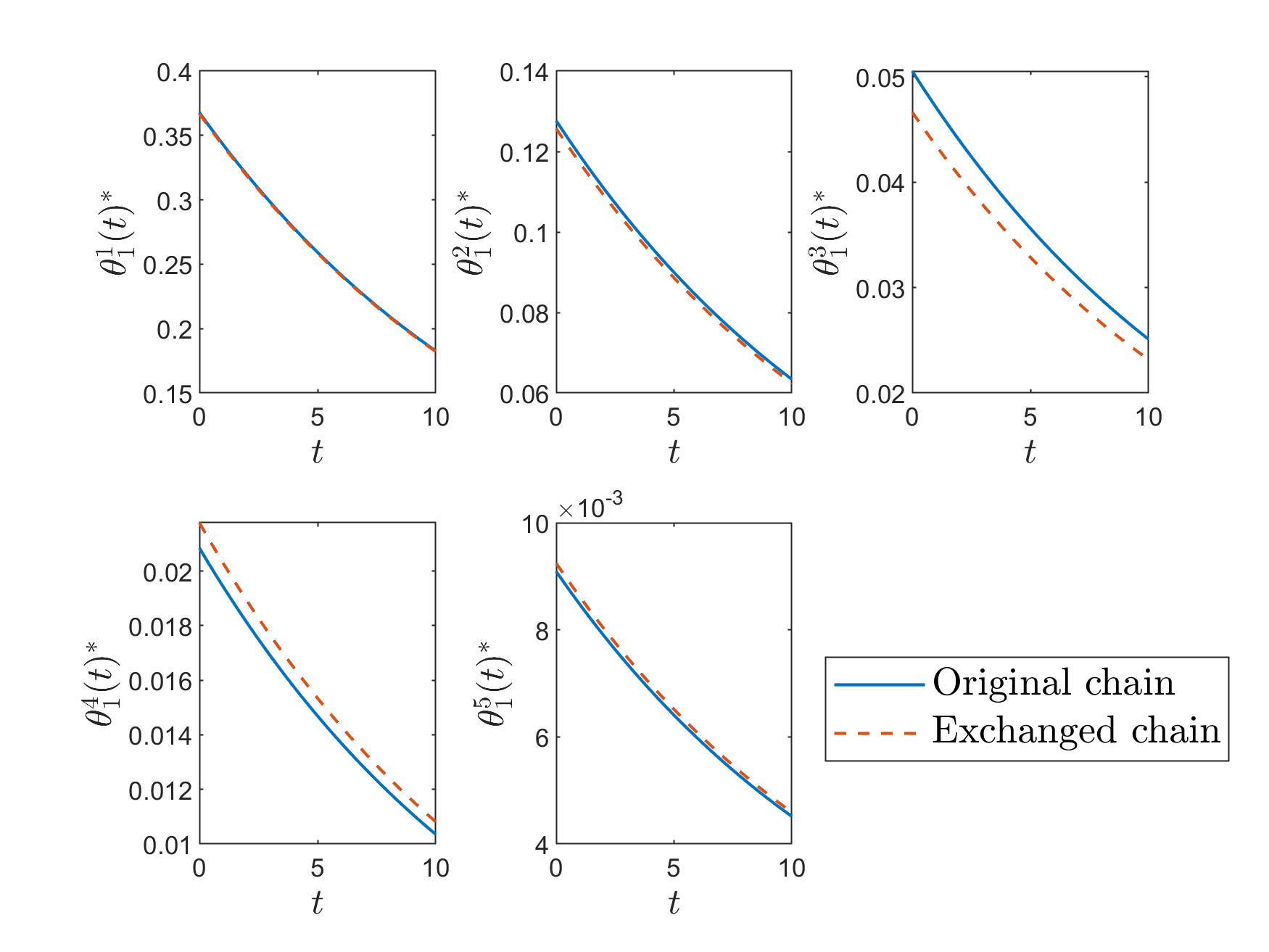}
	 		\caption{First chain}
	 		\label{fig:Reinsurance-chainchange-safetyloadings1}
	 	\end{subfigure}
	 	\begin{subfigure}[b]{0.5\textwidth}
	 		\includegraphics[width=\textwidth]{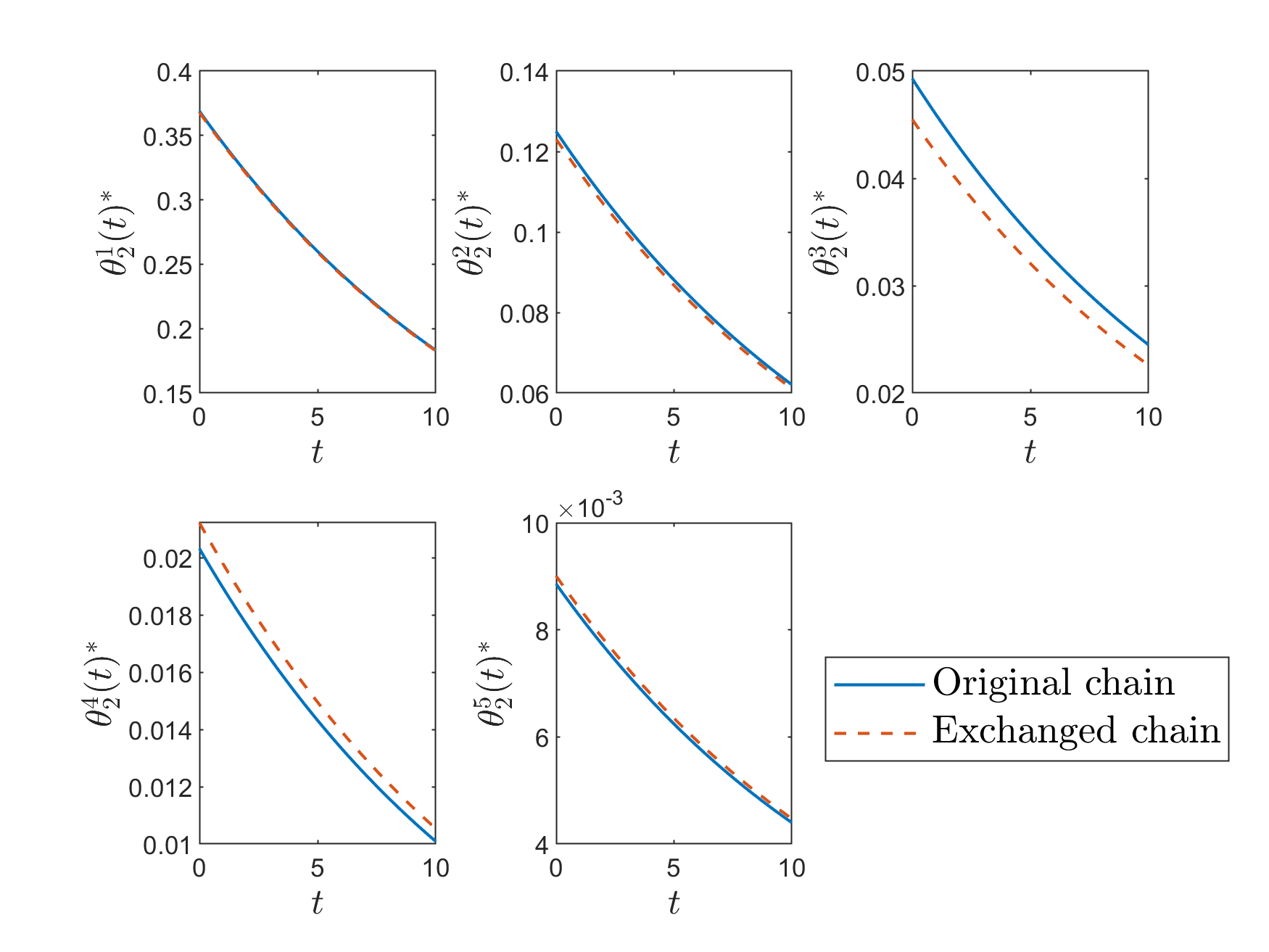}
	 		\caption{Second chain}
	 		\label{fig:Reinsurance-chainchange-safetyloadings2}
	 	\end{subfigure}		
	 	\caption{The impact of swapping the positions of  reinsurer $R_2$ and $R_3$ on the safety loading strategies}	
	 	\label{fig:Reinsurance-chainchange-safetyloadings}
	 \end{figure}

	    \Cref{fig:Reinsurance-chainchange-safetyloadings} depicts the variation of safety loading strategies under different reinsurance chain structures. After the exchange of the reinsurer $R_2$ and reinsurer $R_3$, the safety loading strategy of  reinsurer $R_1,R_2,R_3$ decreases,  whereas those of $R_4$ and $R_5$ increase. From the analysis of \cref{Theorem-ChainChange}, when $\gamma_{R_2}>\gamma_{R_3}$, the reinsurer $R_4$ is confronted with a more risk averse reinsurer after the exchange. At this point, the reinsurer $R_4$ will tend to increase its reinsurance business, which causes that reinsurer $R_5$ will increase its safety load in order to hedge its risk. 
	   And the exchange of $\gamma_{R_2},\gamma_{R_3}$ does not directly affect the relevant parameters in the expressions of  $\theta^{1}_1(t)^{*},\theta^{1}_2(t)^{*},  \theta^{2}_1(t)^{*}, \theta^{2}_2(t)^{*}$. Thus, the decline in $\theta_3(t)^{*}$ leads to a decrease in these equilibrium strategies, which in turn increases the corresponding reinsurance strategies.
	   
	   \Cref{fig:Reinsurance-chainchange-Reinsurance} depicts the variation of reinsurance strategies under different reinsurance chain structures. 
	   The  reinsurer $R_2$ tends to increase its reinsurance strategy $q^3_1(t)^*,q^3_2(t)^*$ due to a decrease in its own risk aversion.
	   On the other hand, The  reinsurer $R_3$ will increase its reinsurance strategy $q^4_1(t)^{*},q^4_2(t)^{*}$ due to an  increase in its own risk aversion.
	   At this point, in order to hedge risk, reinsurer $R_4$ will increase  the safety load strategy $\theta^{4}_{1}(t)^{*}$ and $\theta^{4}_{2}(t)^{*}$ accordingly, which results in the increase of  $\theta^{5}_{1}(t)^{*}$ and $\theta^{5}_{2}(t)^{*}$.

	   \begin{figure}[H]
	   	\begin{subfigure}[b]{0.5\textwidth}
	   		\includegraphics[width=\textwidth]{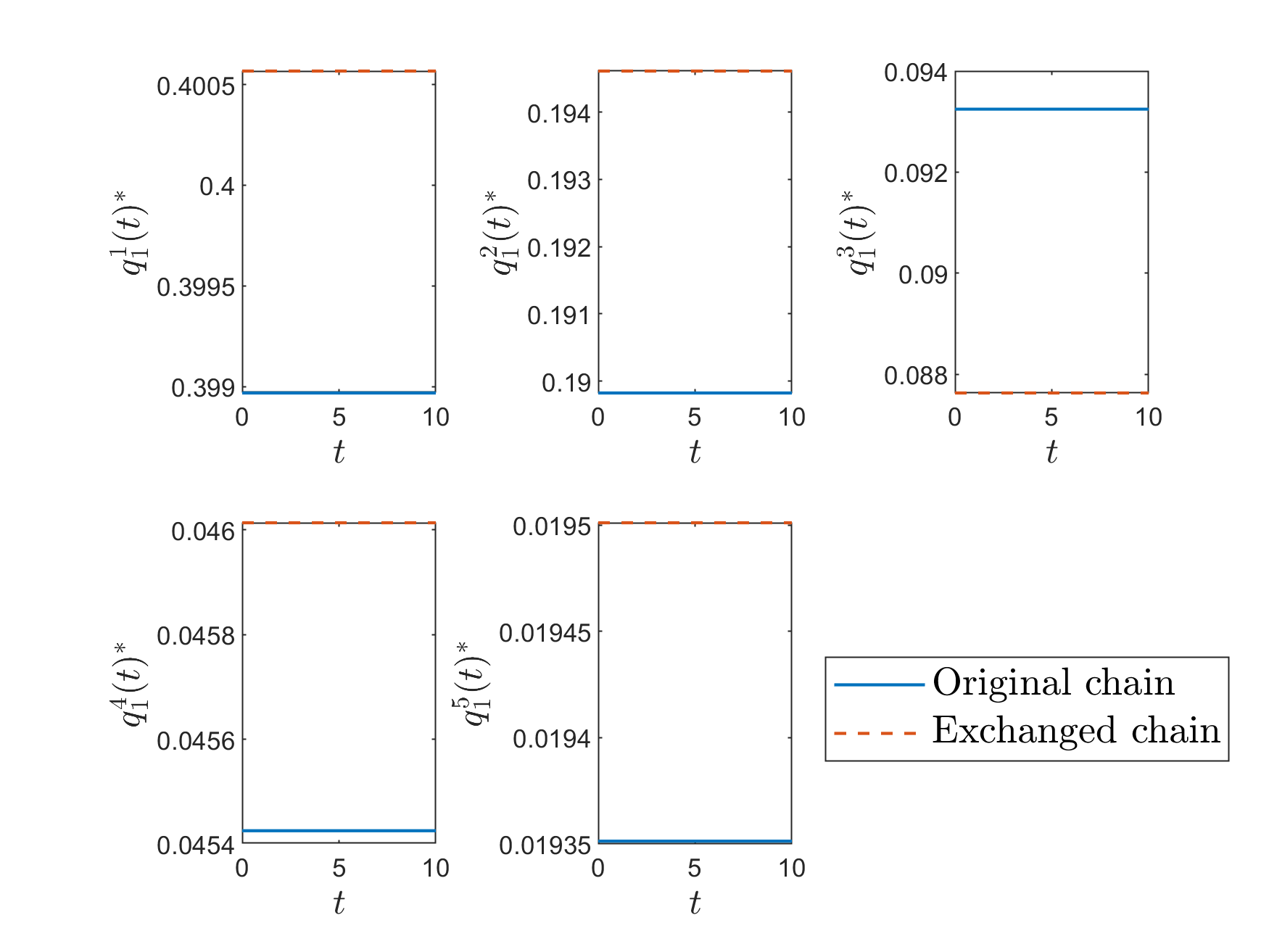}
	   		\caption{First chain}
	   		\label{fig:Reinsurance-chainchange-Reinsurance1}
	   	\end{subfigure}
	   	\begin{subfigure}[b]{0.5\textwidth}
	   		\includegraphics[width=\textwidth]{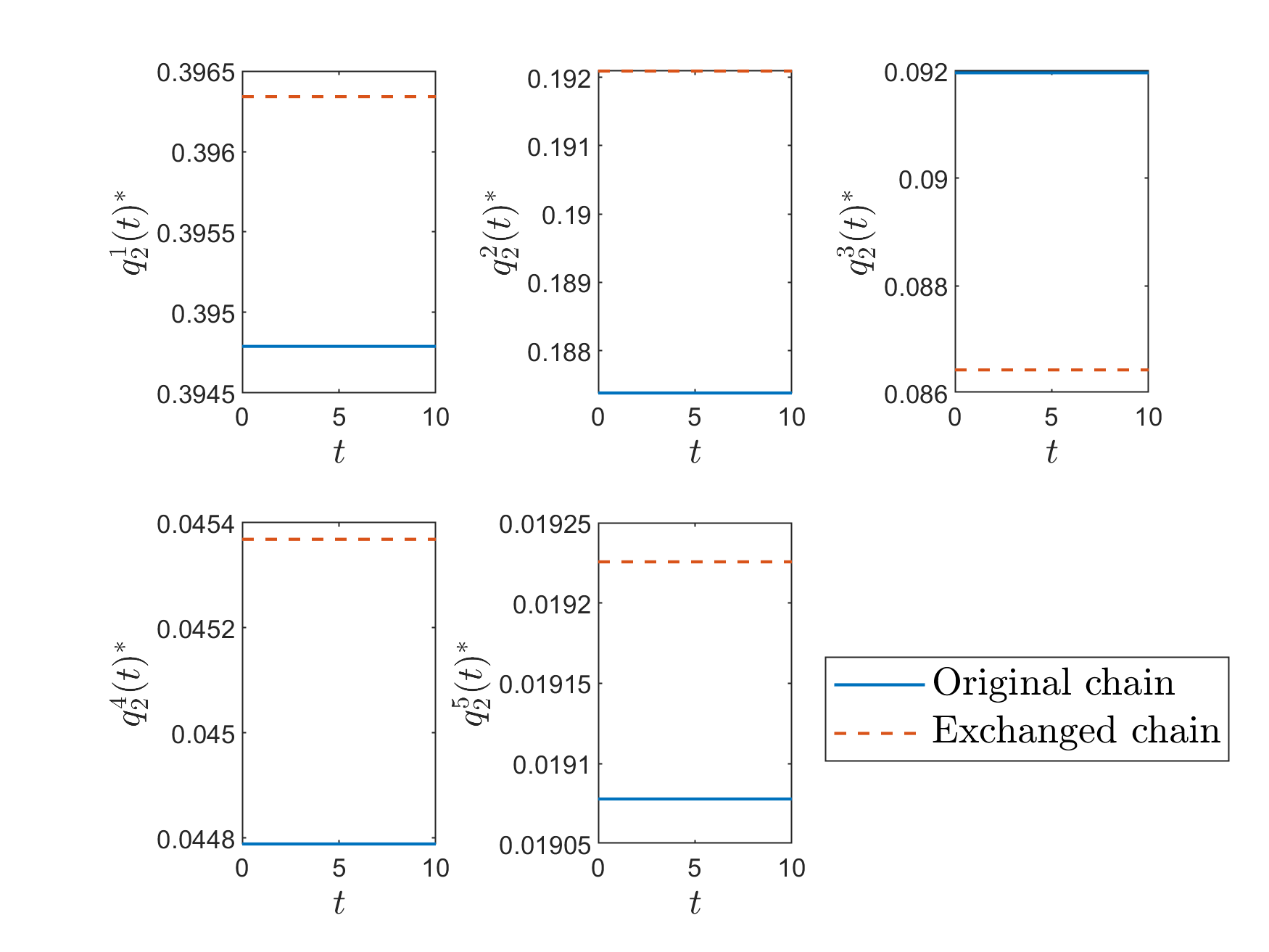}
	   		\caption{Second chain}
	   		\label{fig:Reinsurance-chainchange-Reinsurance2}
	   	\end{subfigure}		
	   	\caption{The impact of swapping the positions of  reinsurer $R_2$ and $R_3$ on the reinsurance strategies}	
	   	\label{fig:Reinsurance-chainchange-Reinsurance}
	   \end{figure}

	  \subsubsection{The effect of increasing the length of the reinsurance chain on the equilibrium strategy}
	  
	  \begin{figure}[h]
	  	\begin{subfigure}[b]{0.5\textwidth}
	  		\includegraphics[width=\textwidth]{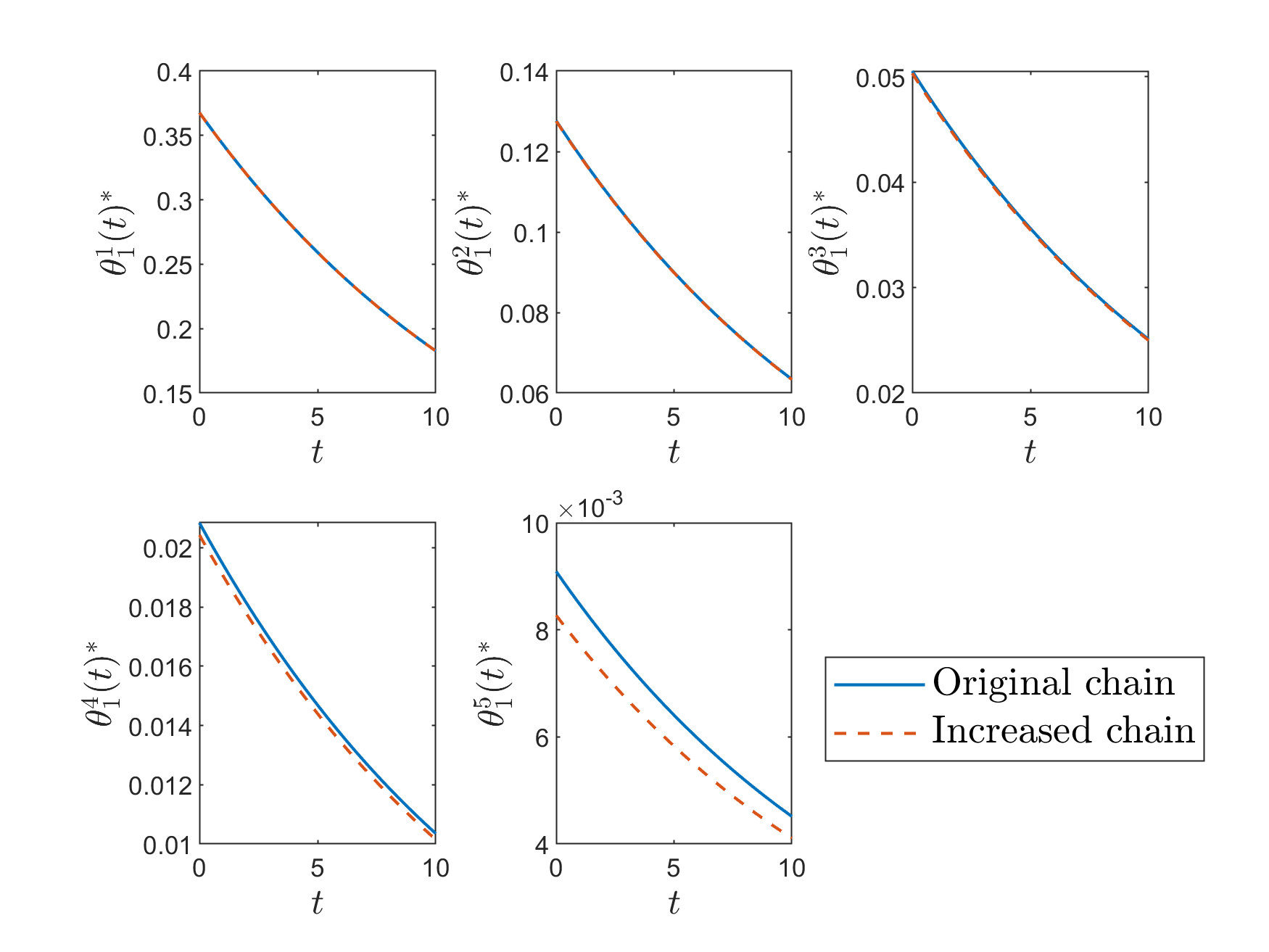}
	  		\caption{First chain}
	  		\label{fig:Reinsurance-chainplus-Safeloadings1}
	  	\end{subfigure}
	  	\begin{subfigure}[b]{0.5\textwidth}
	  		\includegraphics[width=\textwidth]{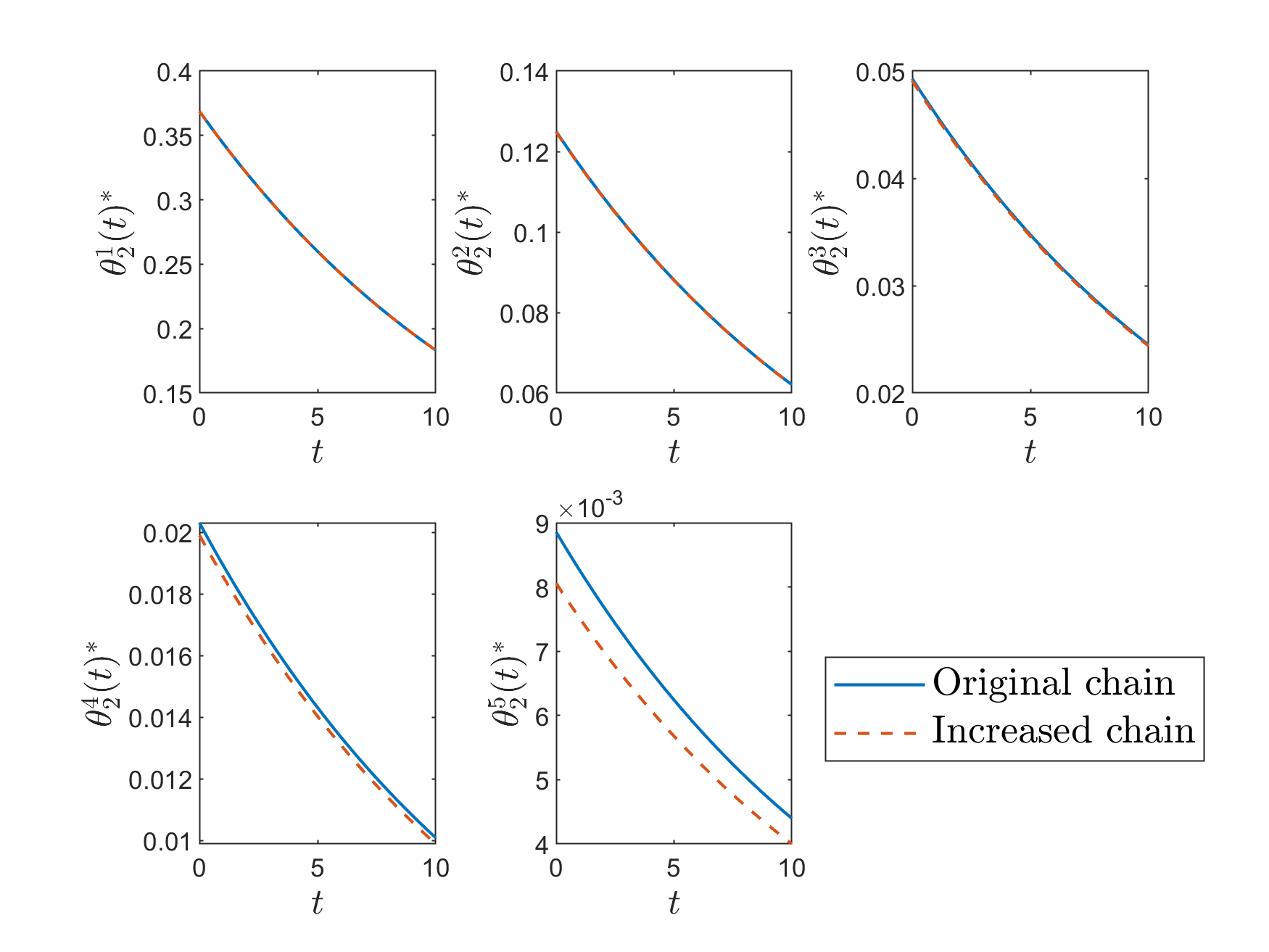}
	  		\caption{Second chain}
	  		\label{fig:Reinsurance-chainplus-Safeloadings2}
	  	\end{subfigure}		
	  	\caption{The effect of increasing the length of the reinsurance chain on the equilibrium safety loading strategy}
	  	\label{fig:Reinsurance-chainplus-Safeloadings}
	  \end{figure}

	  \Cref{fig:Reinsurance-chainplus-Safeloadings} shows that increasing the length of the reinsurance chain generally lowers the safety loadings. 
	  The impact is strongest for reinsurers closest to the newly added reinsurer. 
	  The new reinsurer $R_6$ directly affects $R_5$. 
	  As $R_6$ takes over part of the reinsurance business, $R_5$ assumes less risk and will reduce its safety loading. 
	  This adjustment propagates to other reinsurers, and the effect diminishes as the distance from $R_6$ increases.

	  \Cref{fig:Reinsurance-chainplus-Reinsurance} depicts the impact of increasing reinsurance chain length on equilibrium reinsurance strategies.
	    As shown in  \Cref{fig:Reinsurance-chainplus-Reinsurance}, the equilibrium reinsurance strategy increases as the length of the reinsurance chain extends. When the length of the reinsurance chain increases, the safety loading strategy of reinsurers decreases. At this point, the increased length of the reinsurance chain enhances the overall risk-bearing capacity, prompting insurance companies and reinsurers to expand their reinsurance business.
	  
	 \begin{figure}[H]
		\begin{subfigure}[b]{0.5\textwidth}
			\includegraphics[width=\textwidth]{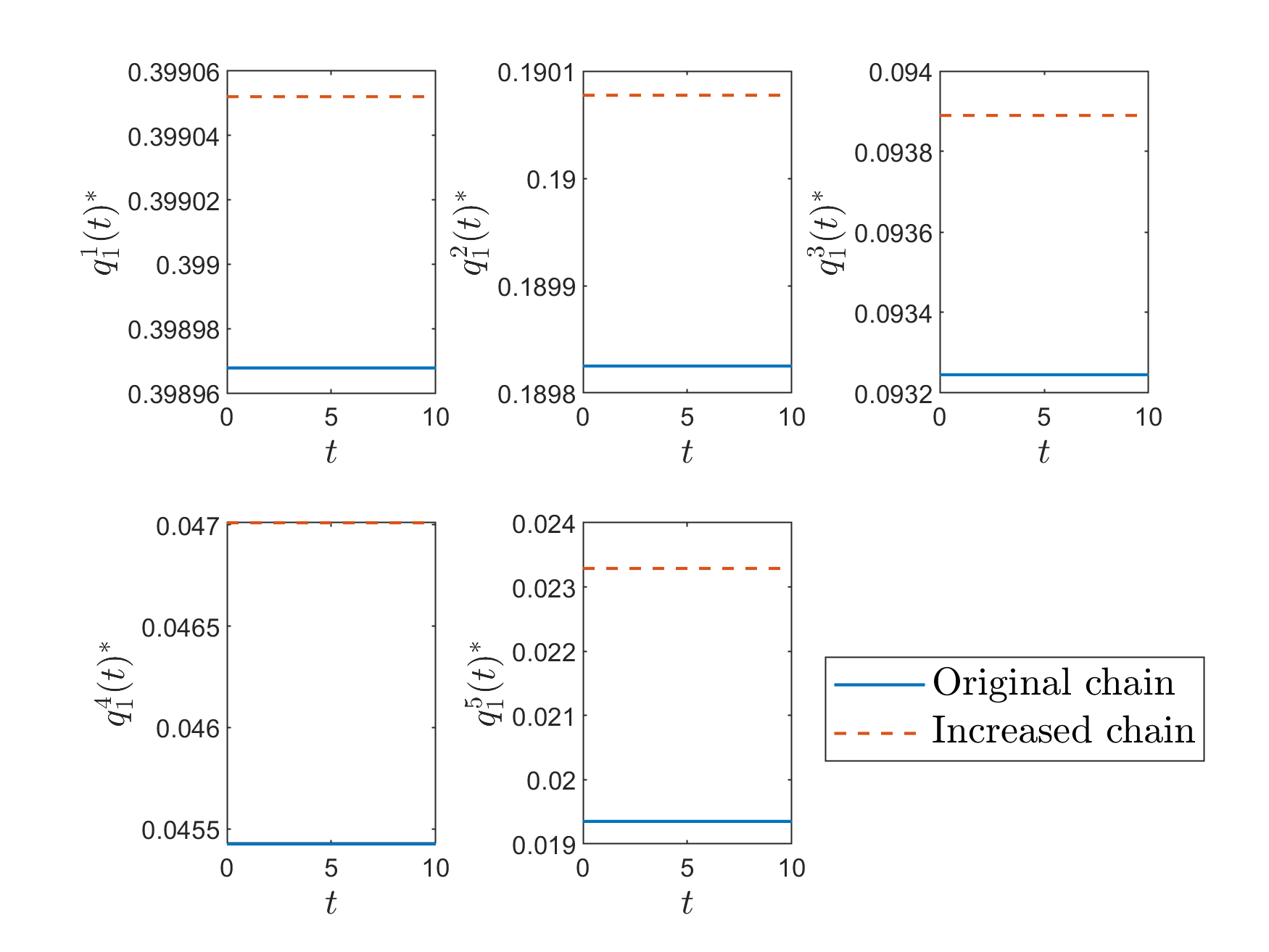}
			\caption{First chain}
			\label{fig:Reinsurance-chainplus-Reinsurance1}
		\end{subfigure}
		\begin{subfigure}[b]{0.5\textwidth}
			\includegraphics[width=\textwidth]{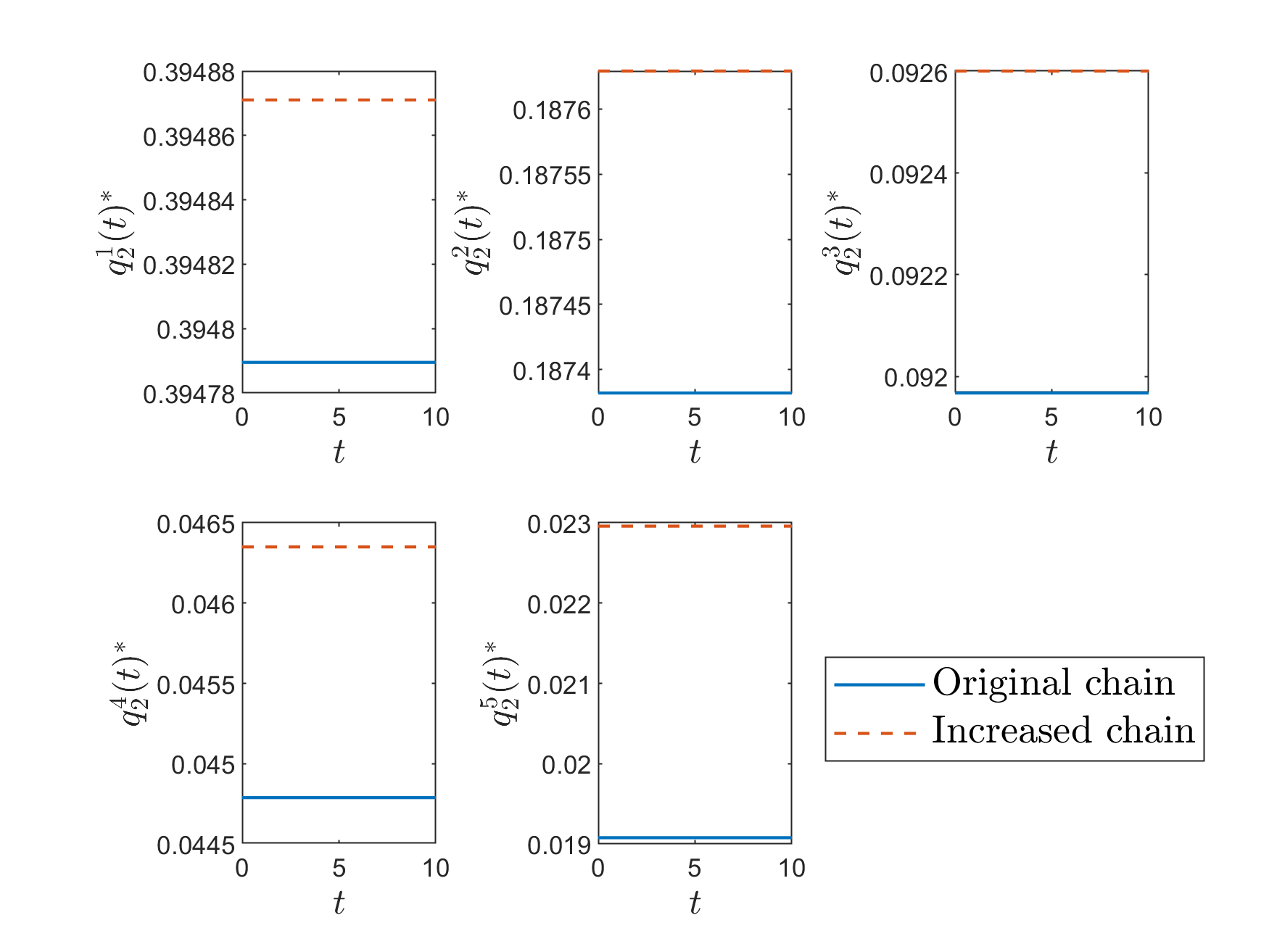}
			\caption{Second chain}
			\label{fig:Reinsurance-chainplus-Reinsurance2}
		\end{subfigure}		
		\caption{The impact of increasing the length of the reinsurance chain on the equilibrium reinsurance strategy}	
		\label{fig:Reinsurance-chainplus-Reinsurance}
	\end{figure}
	
	\subsubsection{The impact of the insurance companies' competition weights 
		$w_{1},w_{2}$
		on the equilibrium strategies}
	
	\Cref{fig:Reinsurance-Competion-Reinsurance} shows that $q^{1}_{i}(t)^{*}$ is directly proportional to $w_k$,  and inversely proportional to $w_i$. Taking $q^1_1(t)^{*}$ as an example, when $w_1$ increases, the competition weight of the first insurance company rises. To maintain its competitive advantage, it tends to retain risk, thus reducing $q^{1}_{1}(t)^{*}$. On the other hand, when $w_2$ rises, $q^{1}_{2}(t)^{*}$ decreases. To attract reinsurance business, reinsurer $R_{1}$ will accordingly reduce the safety load $\theta^{1}_{1}(t)^{*}$ and $\theta^{1}_{2}(t)^{*}$, which leads to an increase in $q^{1}_{1}(t)^{*}$. 
	
	\begin{figure}[H]
		\begin{subfigure}[b]{0.5\textwidth}
			\includegraphics[width=\textwidth]{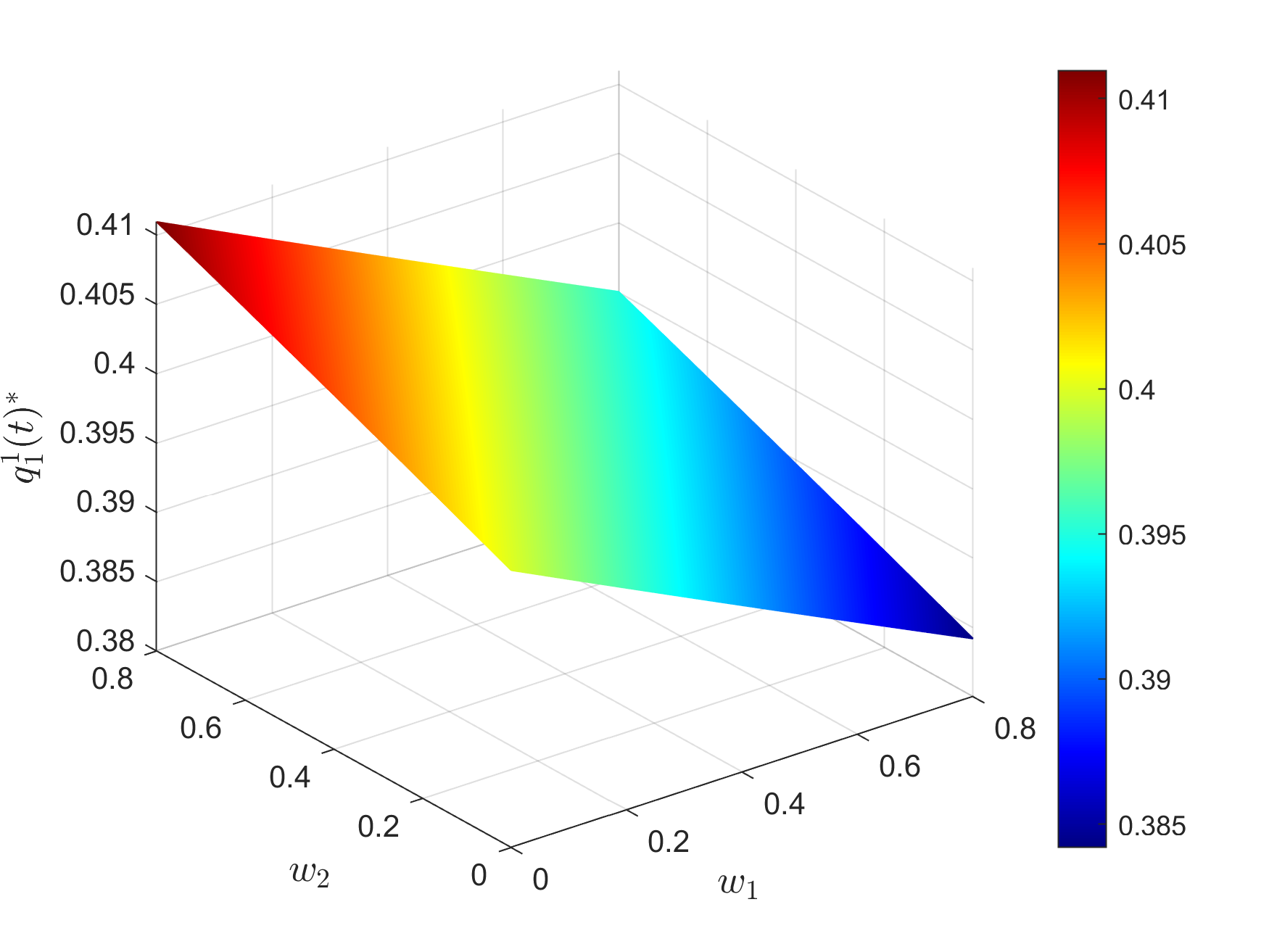}
			\caption{}
			\label{Reinsurance-Competion-Reinsurance-q11}
		\end{subfigure}
		\begin{subfigure}[b]{0.5\textwidth}
			\includegraphics[width=\textwidth]{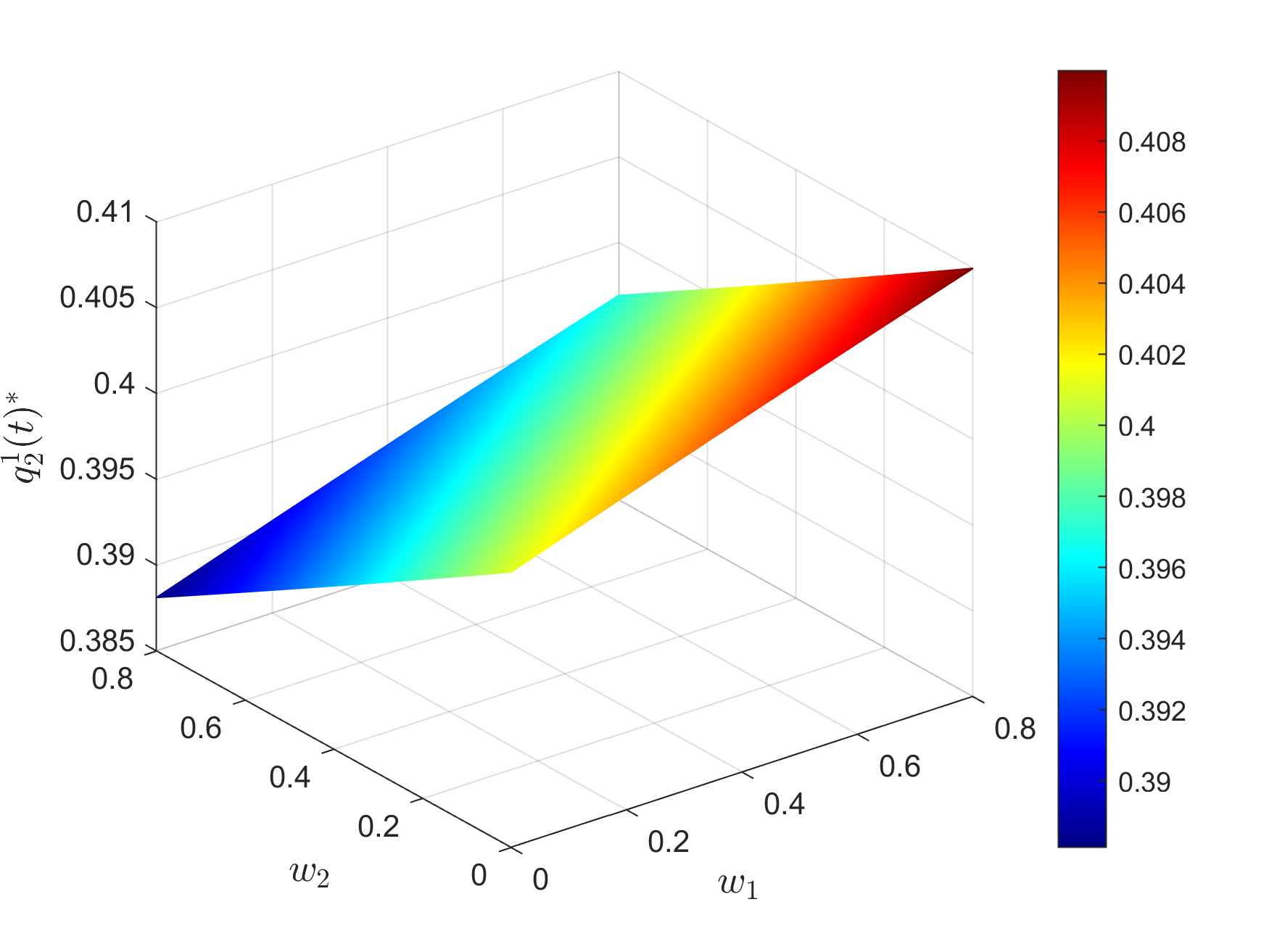}
			\caption{}
			\label{Reinsurance-Competion-Reinsurance-q12}
		\end{subfigure}
				
		\caption{The impact of competitive weights $w_1$ and $w_2$ on reinsurance strategies}	
		\label{fig:Reinsurance-Competion-Reinsurance}
	\end{figure}
	
	\begin{figure}[h]
		\begin{subfigure}[b]{0.5\textwidth}
			\includegraphics[width=\textwidth]{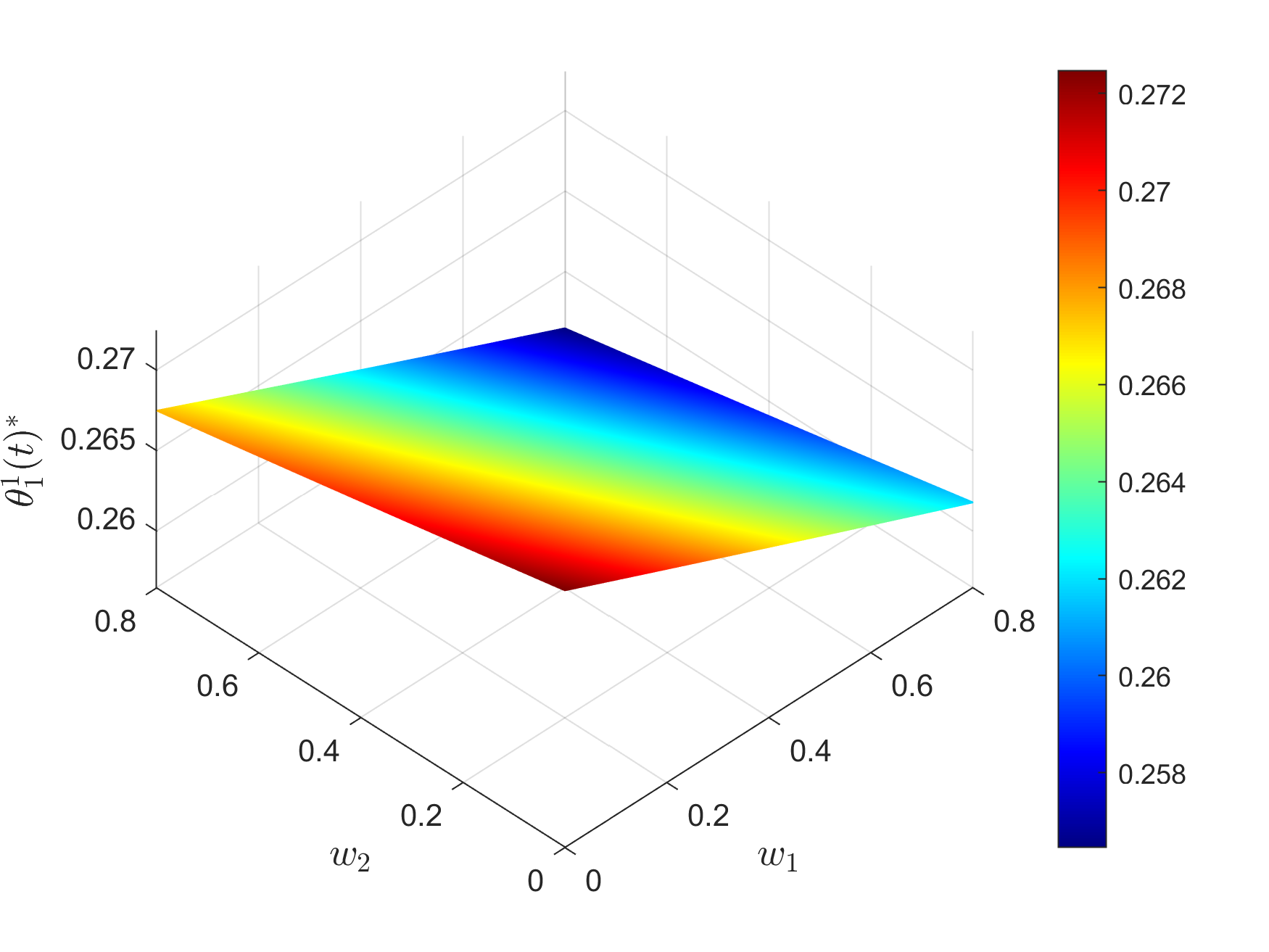}
			\caption{}
			\label{Reinsurance-Competion-Safeloadings1}
		\end{subfigure}
		\begin{subfigure}[b]{0.5\textwidth}
			\includegraphics[width=\textwidth]{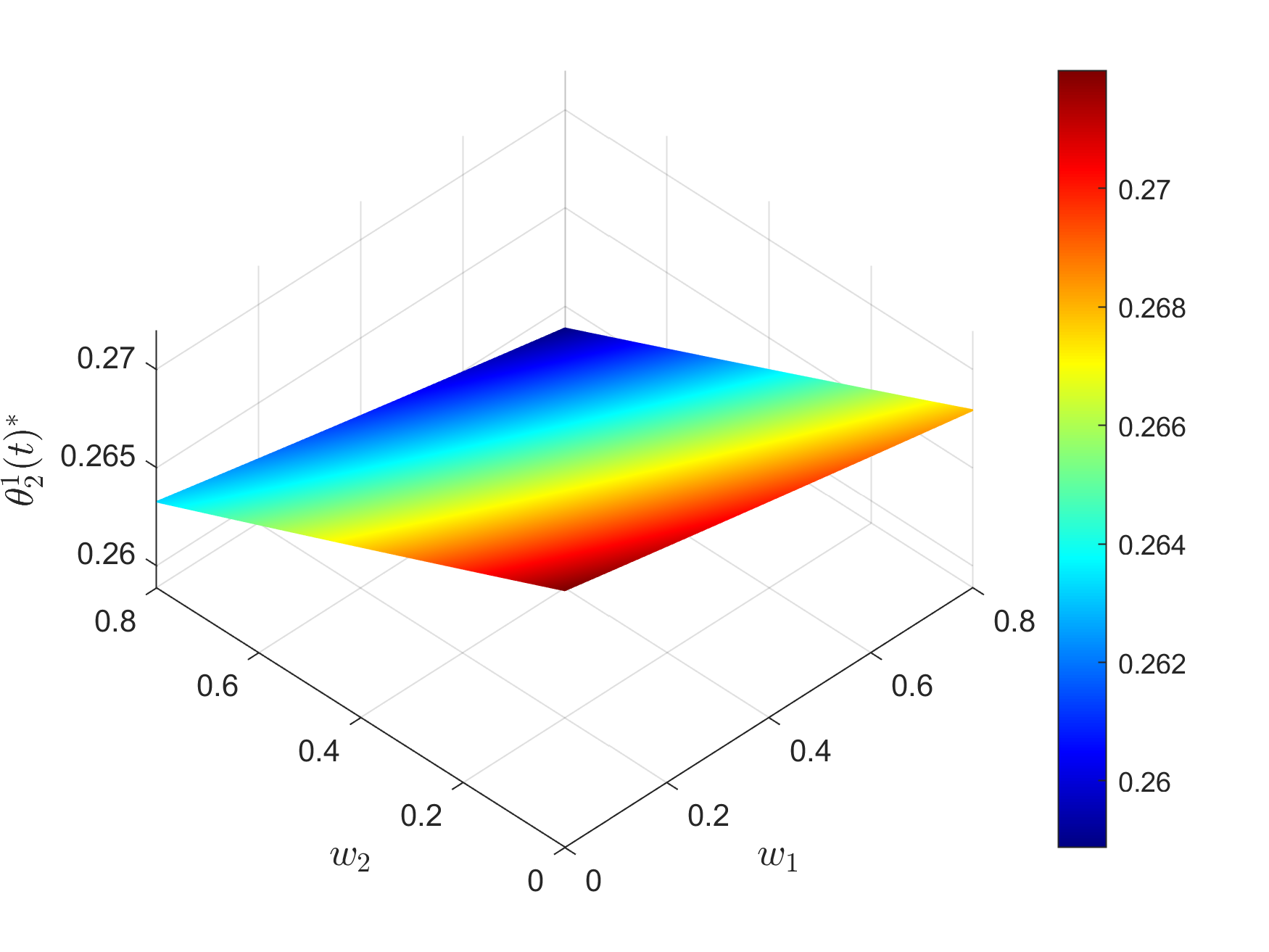}
			\caption{}
			\label{Reinsurance-Competion-Safeloadings2}
		\end{subfigure}		
		\caption{The impact of the competition weights $w_1$ and $w_2$	​
			on the safety loading strategy}	
		\label{fig:Reinsurance-Competion-Safeloadings}
	\end{figure}

	\Cref{fig:Reinsurance-Competion-Safeloadings} shows that the safety load strategy is inversely proportional to the values of 
	$w_1$ and $w_2$
	. When the insurance market is in a highly competitive environment, companies tend to take on more risk to gain an advantage in the market, which reduces reinsurance business. To attract reinsurance business, reinsurance companies at all levels in the reinsurance chain will lower their safety load strategies. The more intense the competition in the market, the greater the reduction in safety load.

	\begin{figure}[h]
		\begin{subfigure}[b]{0.5\textwidth}
			\includegraphics[width=\textwidth]{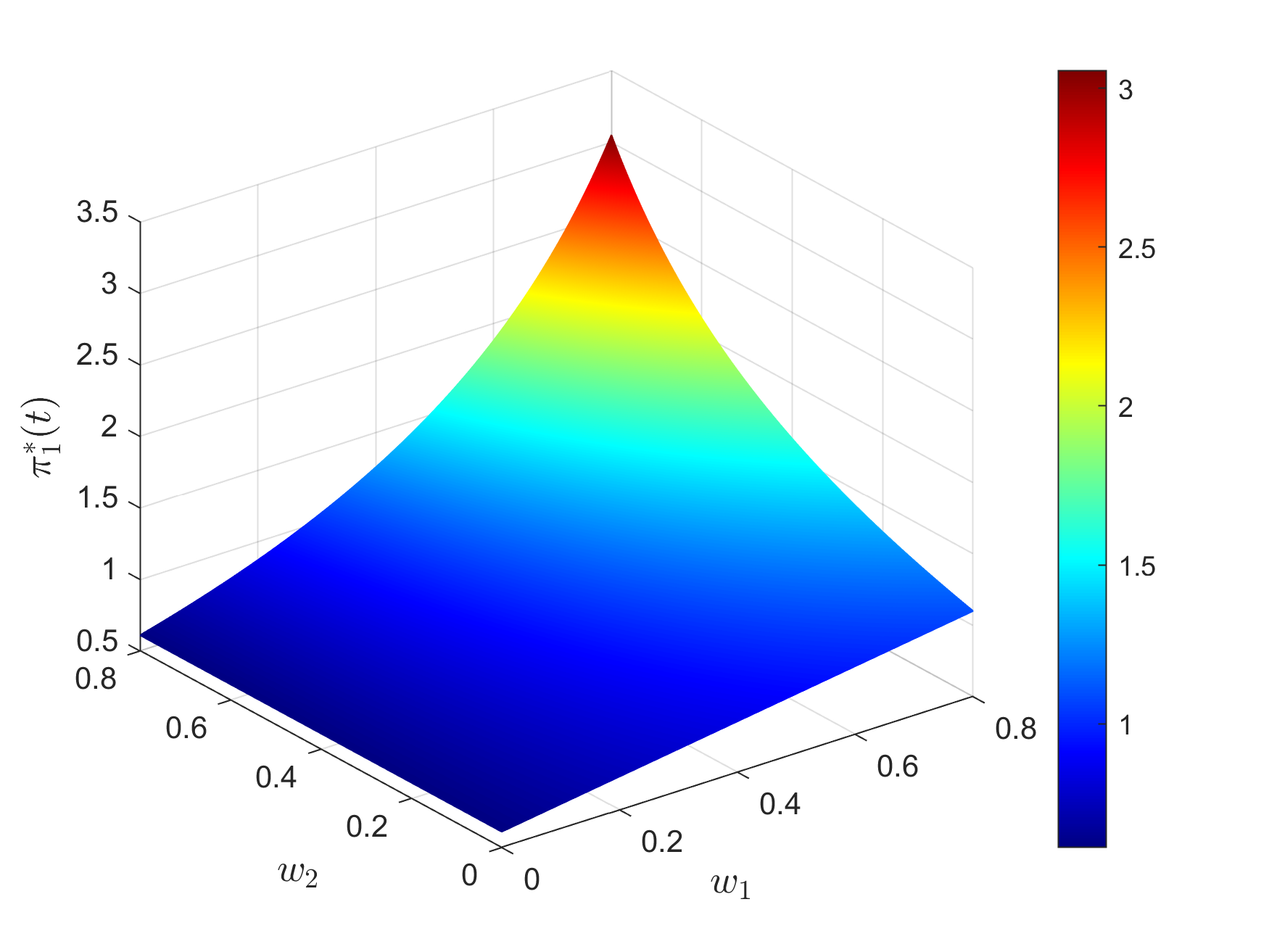}
			\caption{}
			\label{fig:Investment-Competion-pi1}
		\end{subfigure}
		\begin{subfigure}[b]{0.5\textwidth}
			\includegraphics[width=\textwidth]{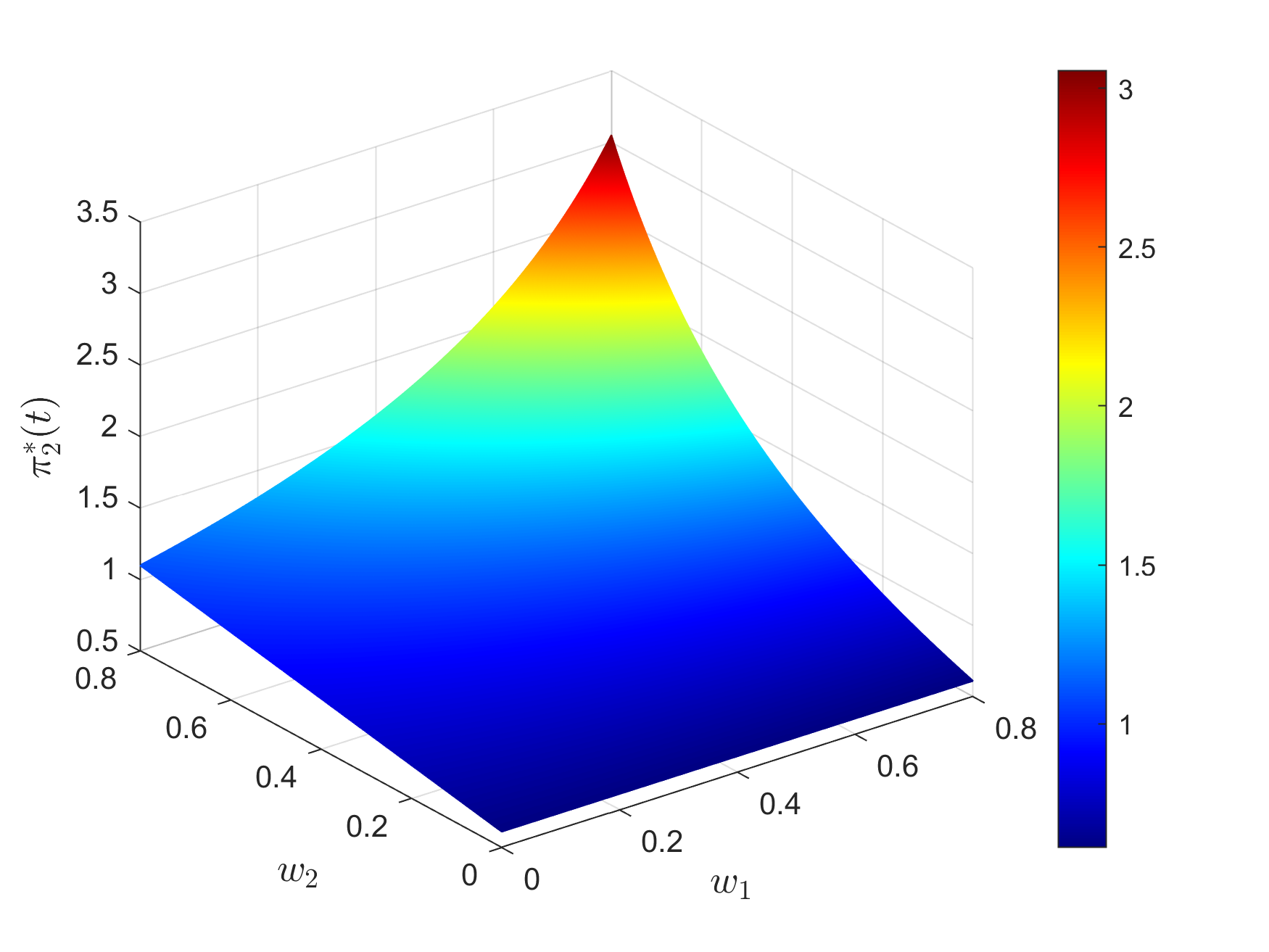}
			\caption{}
			\label{fig:Investment-Competion-pi2}
		\end{subfigure}		
		\caption{The impact of the competition weights $w_1$ and $w_2$	​
			on the equilibrium investment strategies of insurers}	
		\label{fig:Investment-Competion-pi}
	\end{figure}

	It can be seen from  \Cref{fig:Investment-Competion-pi},  when $w_1,w_2$ are close to 1, $\pi_i^*(t)$ will increase sharply. On
	the one hand, the equilibrium investment strategy $\pi_{i}^*(t)$ is proportional to the competition
	parameter $w_i$. This suggests that when an insurer focuses more on its competitors, it tends to increase its investment in risky assets in an effort to achieve higher returns, thereby aiming to gain a competitive advantage. On the other hand, the equilibrium investment strategy $\pi_{i}^*(t)$ is positively correlated with
	$w_{k}$. When the competition weight of  competitor  increases, the competitor will invest more in risk asset. In order to sustain their competitive advantage, insurance companies tend to increase their investments.  This strategic move is essential for enhancing financial stability, improving profitability, and ensuring long-term growth.

	\subsection{The case under excess-of-loss reinsurance}
	This section presents a numerical illustration to examine the impact of relevant parameters on the equilibrium reinsurance strategy within the reinsurance chain, as discussed in \cref{Excess-of-loss reinsurance-economy analysis}. In this section, unless otherwise stated, the relevant parameters are set as follows: $\gamma_{1}=0.8,\gamma_{R_{1}}=0.7,\gamma_{R_{2}}=0.6, \lambda=3,\lambda_{1}=4,  \mu_{1}=0.1,  b_{1}=0.2, r=0.07, \mu=0.2,\sigma=0.5, \delta_{1}=1$.
	
		\begin{figure}[H]
		\begin{subfigure}[b]{0.5\textwidth}
			\includegraphics[width=\textwidth]{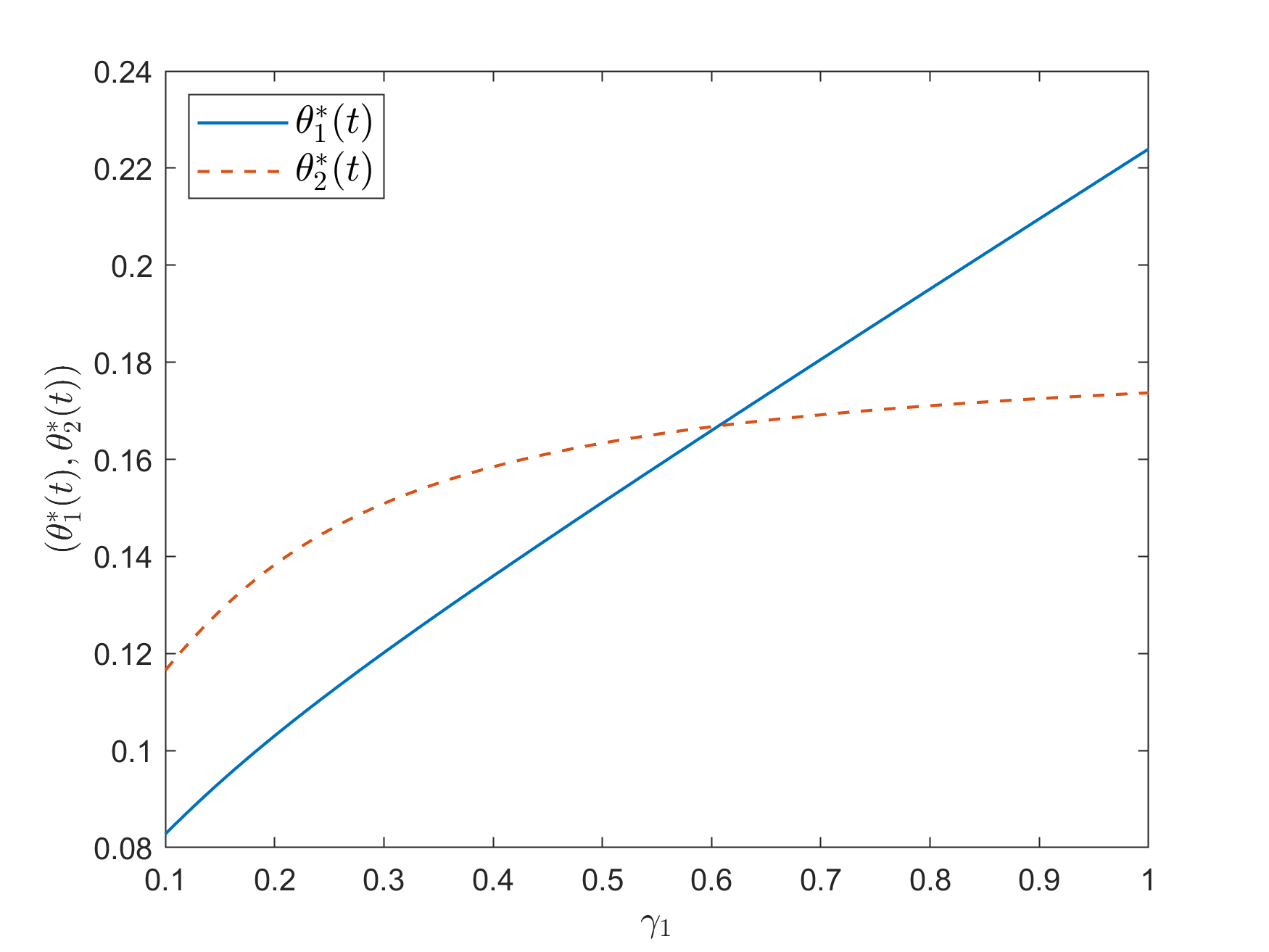}
			\caption{}
			\label{Excess_of_loss_reinsurance-gamma-theta}
		\end{subfigure}
		\begin{subfigure}[b]{0.5\textwidth}
			\includegraphics[width=\textwidth]{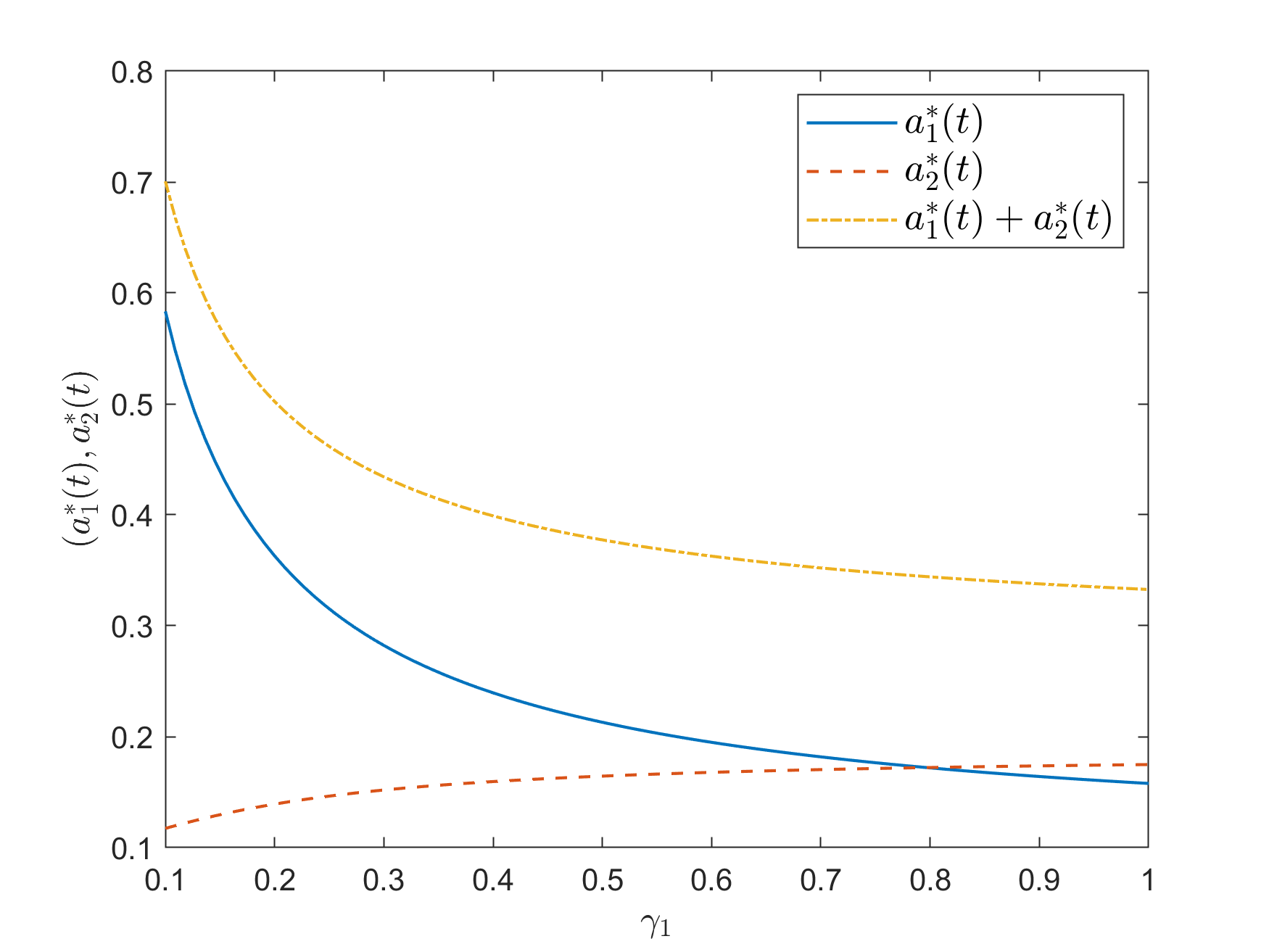}
			\caption{}
			\label{Excess_of_loss_reinsurance-gamma-a}
		\end{subfigure}		
		\caption{Impact of $\gamma_{1}$ on optimal reinsurance contract}	
		\label{fig:Excess_of_loss_reinsurance-gamma}
	\end{figure}

	 \Cref{fig:Excess_of_loss_reinsurance-gamma} illustrates the influences of $\gamma_{1}$ on the optimal reinsurance contract $\theta_1^*(t), \theta_2^*(t),a_1^*(t)$, $a_2^*(t)$, respectively.
	It can be seen from \Cref{fig:Excess_of_loss_reinsurance-gamma}, when $\gamma_{1}$ rises, $\theta_1^*(t), \theta_2^*(t),  a_1^*(t),a_2^*(t),  a_1^*(t)+a_2^*(t)$ are  increasing, increasing, decreasing, increasing, decreasing, respectively. The more risk-averse the insurer $1$ is, the less risk it retains, which means a smaller $a_1^*(t)$. In this case, reinsurer $R_1$ adopts a higher $\theta_1^*(t)$ to hedge risks from more reinsurance. Due to receiving reinsurance at higher premium rates, reinsurer $R_1$ tends to have a higher excess-of-loss retention level $a_2^*(t)$. As a result of the reduction in $a_1^*(t)+a_2^*(t)$, reinsurer $R_2$  will reduce $\theta_2^*(t)$ in order to attract more reinsurance business.

    \begin{figure}[h]
    	\begin{subfigure}[b]{0.5\textwidth}
    		\includegraphics[width=\textwidth]{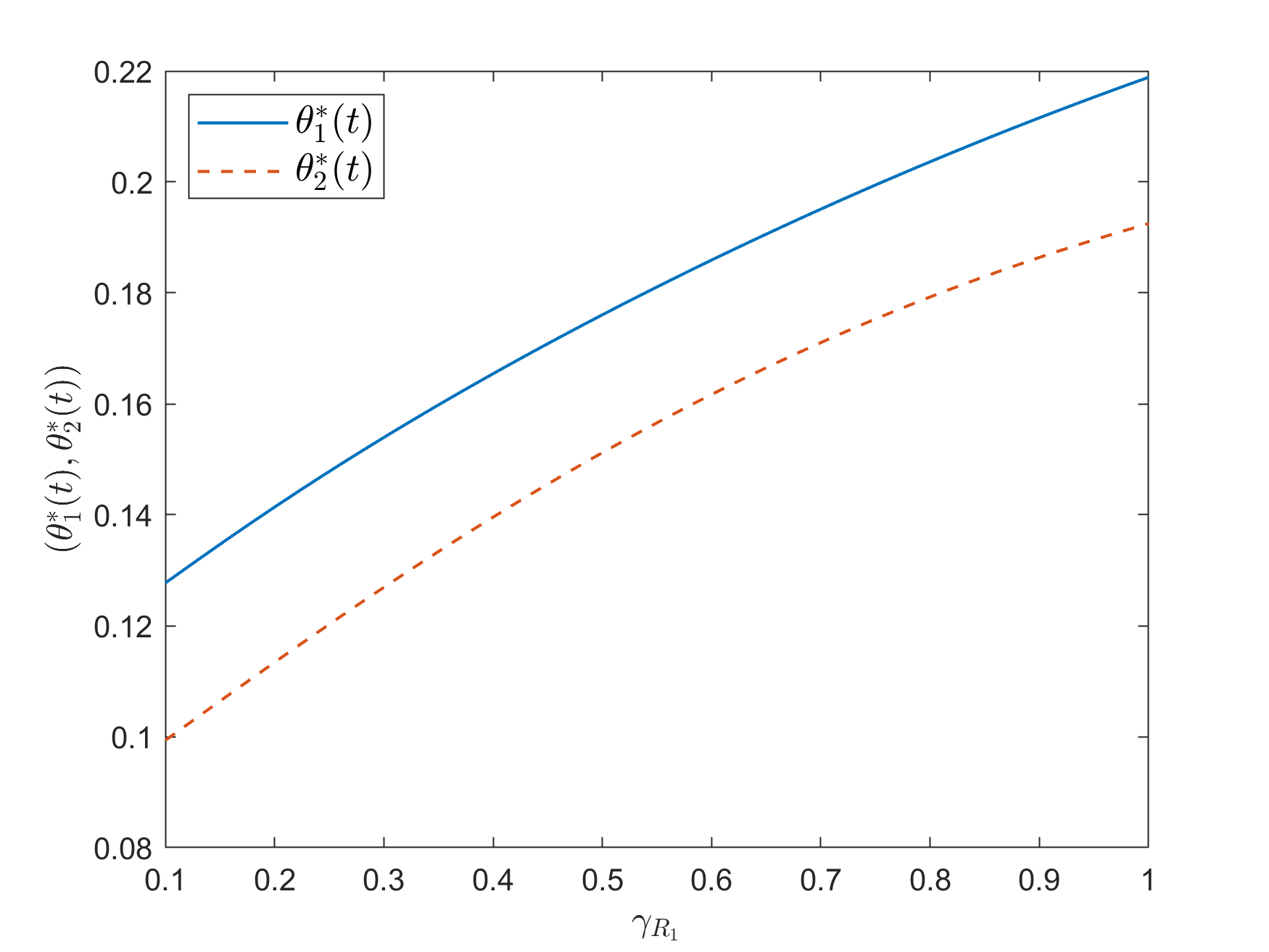}
    		\caption{}
    		\label{Excess_of_loss_reinsurance-gammaR1-theta}
    	\end{subfigure}
    	\begin{subfigure}[b]{0.5\textwidth}
    		\includegraphics[width=\textwidth]{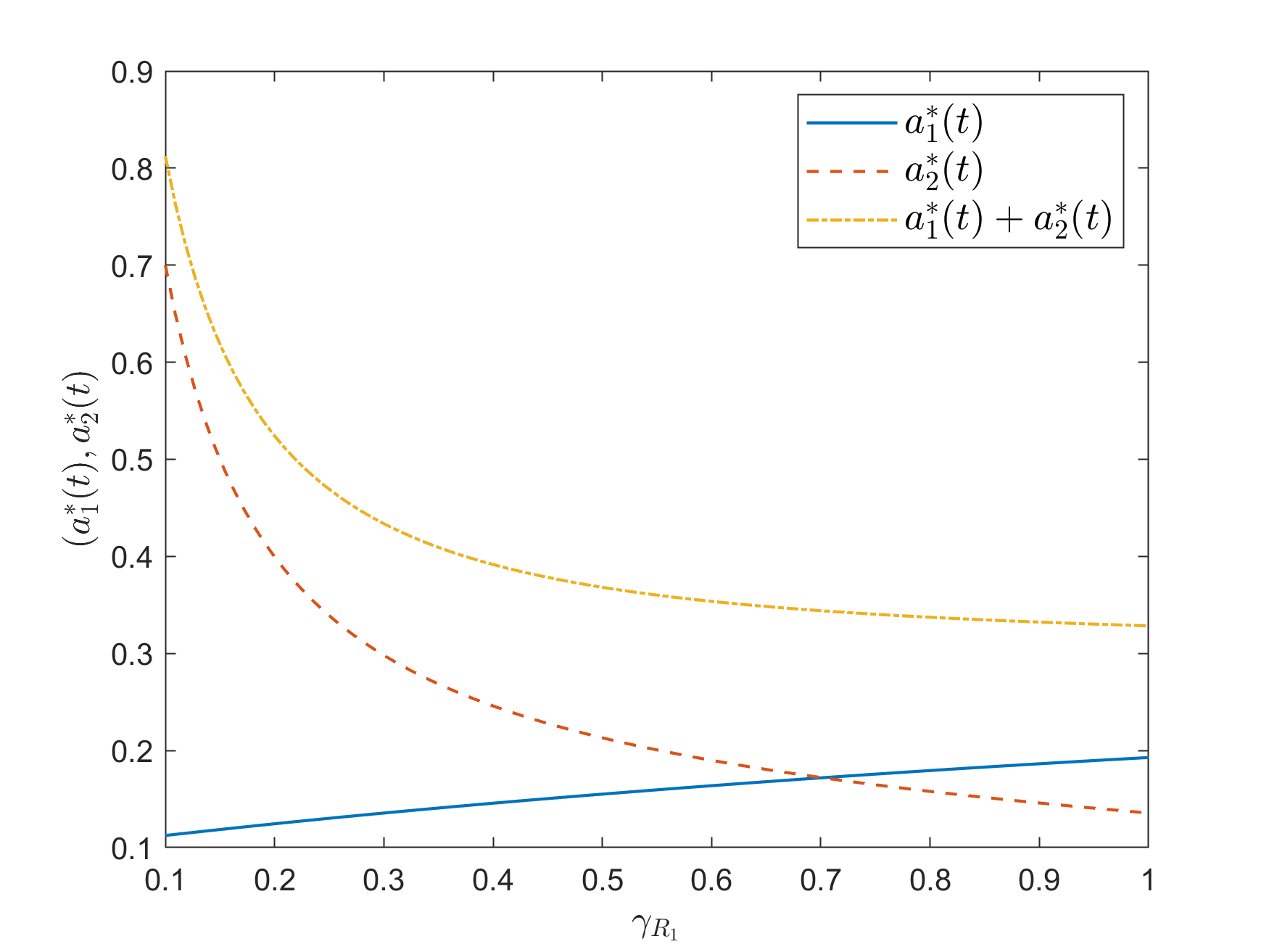}
    		\caption{}
    		\label{Excess_of_loss_reinsurance-gammaR1-a}
    	\end{subfigure}		
    	\caption{Impact of $\gamma_{R_{1}}$ on optimal reinsurance contract}	
    	\label{fig:Excess_of_loss_reinsurance-gammaR1}
    \end{figure}

    \Cref{fig:Excess_of_loss_reinsurance-gammaR1} depicts the effect of $\gamma_{R_{1}}$ on the optimal reinsurance
    contract $\theta_1^*(t), \theta_2^*(t),a_1^*(t)$, $a_2^*(t)$, respectively. As shown
    in \Cref{fig:Excess_of_loss_reinsurance-gammaR1}, when $\gamma_{{R_1}}$ rise, reinsurer $R_{1}$ will choose a higher safety loading $\theta_{1}^*(t)$ and a lower excess-of-loss retention level $a_{2}^*(t)$ to decrease risk. At this point, a higher  excess-of-loss retention level $a_1^*(t)$ will be chosen by insurer $1$. 
    Observing the decrease of  $a_1^*(t)+a_2^*(t)$, reinsurer $R_{2}$ has a higher $\theta_{2}^*(t)$. 
    
    \begin{figure}[h]
    	\begin{subfigure}[b]{0.5\textwidth}
    		\includegraphics[width=\textwidth]{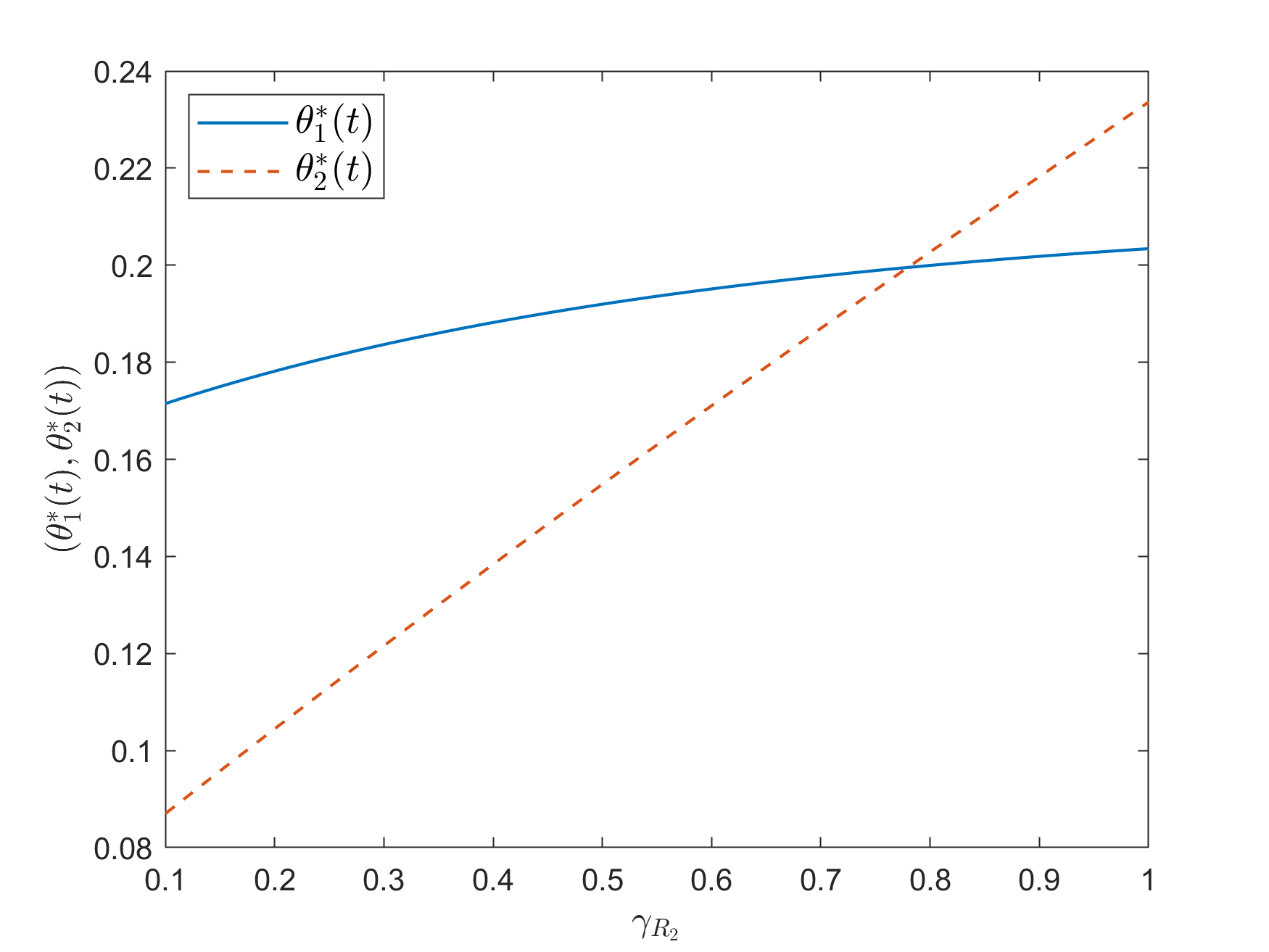}
    		\caption{}
    		\label{fig:Excess_of_loss_reinsurance-gammaR2-theta}
    	\end{subfigure}
    	\begin{subfigure}[b]{0.5\textwidth}
    		\includegraphics[width=\textwidth]{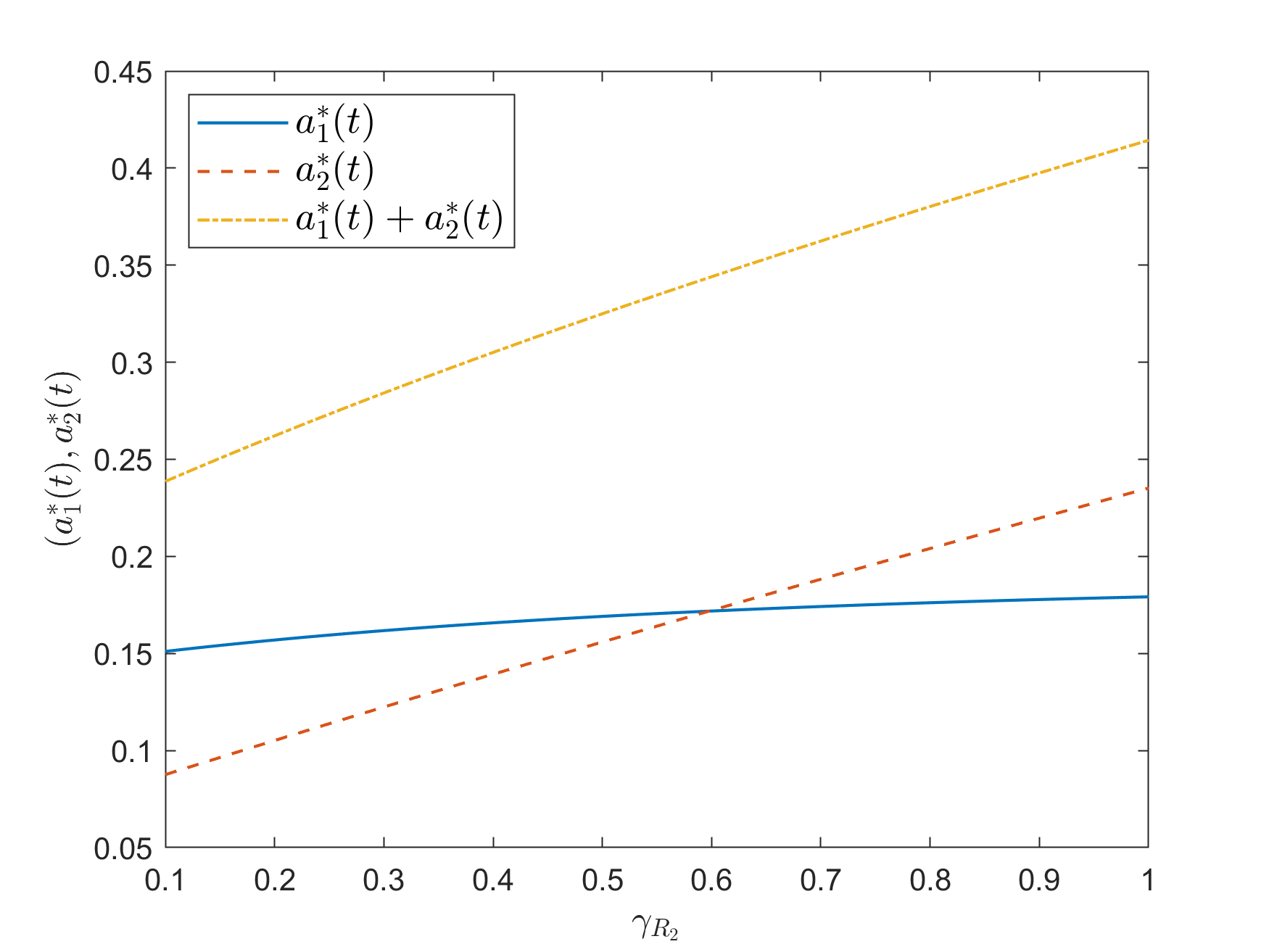}
    		\caption{}
    		\label{fig:Excess_of_loss_reinsurance-gammaR2-a}
    	\end{subfigure}		
    	\caption{Impact of $\gamma_{R_{2}}$ on optimal reinsurance contract}	
    	\label{fig:Excess_of_loss_reinsurance-gammaR2}
    \end{figure}
    
    \Cref{fig:Excess_of_loss_reinsurance-gammaR2} depicts the effect of $\gamma_{R_{2}}$ on the optimal reinsurance
    contract $\theta_1^*(t), \theta_2^*(t),a_1^*(t)$, $a_2^*(t)$, respectively. As shown
    in \Cref{fig:Excess_of_loss_reinsurance-gammaR2}(a), when $\gamma_{{R_2}}$ rise, reinsurer $R_{2}$ will choose a higher safety loading $\theta_{2}^*(t)$ to manage risk. 
    Observing a higher $\theta_{2}^*(t)$, reinsurer $R_{1}$ tends to increase excess-of-loss retention level $a_{2}^*(t)$. By retaining a larger share of the claims risk, the reinsurer $R_{1}$ will increase the safety loading $\theta_{1}^*(t)$ to cover the risk.
    At this point, a higher  excess-of-loss retention level $a_1^*(t)$ will be choose by insurer $1$.

	\section{Conclusion}
	  This paper studies the optimal investment and reinsurance problem involving $m$ insurers and 
	  $n$ reinsurers under a reinsurance chain structure. The competition among the insurance companies is described by a non-zero-sum game, while the interactions between agents at different levels of the reinsurance chain are modeled as the Stackelberg game. 
	  The claims processes of the insurers are modeled by dependent compound Poisson processes, and both the insurers and reinsurers adopt the expected value premium principle. All insurers and reinsurers invest in the same risky and risk-free assets. 
	  Under the mean-variance  criterion, the extended HJB equation is derived, and the corresponding equilibrium strategies are obtained.
	  
	  Based on the analysis of the results, 
	  the main conclusions of this paper are as follows: 
	  (i) When the expected premium principle is adopted, the optimal structure of the reinsurance chain satisfies:  $\gamma_{R_{1}}<\gamma_{R_{2}}<\cdots<\gamma_{R_{n}}$. In the optimal structure, the reinsurance chain is able to bear more risk, accept more reinsurance business, and provide the insurers with the most favorable safety loading strategy. Reasonably constructing the reinsurance chain is of significant importance for enhancing the market's risk-bearing capacity and maintaining the stability of the market.
	  (ii) When a new reinsurer joins the reinsurance chain as a terminal reinsurance company, it enhances the overall risk-bearing capacity and the stability of the reinsurance chain. Therefore, constructing a  multi-layer reinsurance structure can improve the market's ability to withstand risks.
	  (iii) As market competition intensifies, insurers exhibit a tendency to retain more risk in pursuit of higher returns, which in turn induces reinsurer $R_1$ to reduce safety loadings to attract reinsurance business. The propagation effect along the chain further synchronizes safety loading strategies across all tiers.
	  (iv) 
	  In a fiercely competitive insurance market, insurers tend to take on more risk to gain a competitive advantage, which leads to a reduction in reinsurance business.
	  To attract reinsurance business, reinsurers at all levels of the reinsurance chain lower their safety load strategies. The more intense the market competition, the greater the reduction in the safety load. Therefore, maintaining a stable insurance market is crucial for enhancing the market's ability to operate stably.

	  The results of the model reveal a pronounced strategic coupling among insurers: each insurer’s reinsurance decisions are shaped not only by its own risk constraints but also directly by the actions of the other $m-1$ competitors. Unlike \textcite{cao2023-two-reinsurance}, this study does not compare reinsurance chains with reinsurance trees; instead, it extends the dimensionality of the model to $m$ insurers and examines how this expansion influences the reinsurance chain. The findings show that such interdependencies and competitive behaviors propagate upward through the reinsurance chain, affecting contracts at every level and remedying the limitations of traditional isolated-decision models in describing competitive market dynamics. Compared with \textcite{wang2024optimal}, this study advances beyond single-layer Stackelberg games by analyzing how insurer competition cascades through subsequent layers of the reinsurance chain. This strategic coupling effect, absent in single-insurer models, underscores the enhanced applicability and generality of the proposed framework.

	 There are several important issues of reinsurance chain that remain to
	 	be investigated: (i) The existence and uniqueness of solutions for the reinsurance chain using excess-of-loss reinsurance in the case of $n$ reinsurers remains to be discussed.
	 	(ii) In this paper, one reinsurance chain is considered, which can be extended to multiple
	 	 reinsurance chains or a reinsurance network.
	 	(iii) Insurers can choose only reinsurer $R_1$. In reality, insurers are often faced with different reinsurers. The situation where there are multiple reinsurance chains can be considered.
	 	Addressing these issues will
	 	contribute to a more comprehensive understanding of the functioning of the insurance market
	 	and implementation of reinsurance contracts.

\begin{appendices}

		\section{Technical proofs}
		
		\label{Appendix-A}
		
		\setcounter{equation}{0}
		\setcounter{subsection}{0}
		\renewcommand{\theequation}{A.\arabic{equation}}
		\renewcommand{\thesubsection}{A.\arabic{subsection}}
		\vspace{0.2cm}

		\subsection{Proof of \cref{Theorem-Proportional reinsurance-equilibrium strategy}} 
		\label{proof-Proportional reinsurance-equilibrium strategy}
		\begin{proof}
				The extended HJB equation for the insurer $i$ is as follows
		\begin{equation}
			\label{insurer-extended HJB equation}
			\left\{
			\begin{aligned}
				&
				\sup_{u_{i}\in\mathcal{A}_{i}}\mathcal{L}^{u_{i}, u_{i-}}V^{i}(t, x_{i})-\frac{\gamma_{i}}{2}\mathcal{L}^{u_{i}, u_{i-}}g^{i}(t, x_{i})^2+\gamma_{i}
				g^{i}(t, x_{i})\mathcal{L}^{u_{i}, u_{i-}}g^{i}(t, x_{i})=0,
				\\
				&
				V^{i}(T, x_{i})=x_{i},\quad  
				g^{i}(T, x_{i})=x_{i},\quad  
				\mathcal{L}^{u^*_{i}, u^*_{i-}}g^{i}(t, x_{i})=0.  
			\end{aligned}
			\right.
		\end{equation}
		
		Assume $V^{i}(t, x_{i})$ and $g^{i}(t, x_{i})$ take the following forms.
		\begin{equation}
			\label{insurer-value function-assumption}
			\begin{aligned}
				V^{i}(t, x_{i})=\mathrm{e}^{r(T-t)}x_{i}+D_{i}(t),\quad
				g^{i}(t, x_{i})=\mathrm{e}^{r(T-t)}x_{i}+d_{i}(t).
			\end{aligned}
		\end{equation}

		 Substituting (\ref{insurer-value function-assumption}) into (\ref{insurer-extended HJB equation}) yields
			\begin{align*}
				\frac{\partial D_{i}}{\partial t} &+\mathrm{e}^{r(T-t)}\Biggl[(\mu-r)\left(\pi_{i}(t)
				-
				\omega_{m,i}\sum_{k\neq i}^{m} \pi_{k}(t)\right)
				\\
				&
				+
				\left(\eta_{i}-q^{1}_{i}(t)\theta^{1}_{i}(t)\right)o_{i}
				-
				\omega_{m,i}\sum_{k\neq i}^{m} \left(\eta_{k}-q^1_k(t)\theta^{1}_{k}(t)\right)o_k
				\Biggr]
				\\
				&-\frac{\gamma_{i}}{2}\mathrm{e}^{2r(T-t)}
				\Biggl[
				\left(1-q^{1}_{i}(t)\right)^2\sigma_{i}^2
				-
				2\omega_{m,i} \sum_{k\neq i}^{m} \rho_{ik}\left(1-q^{1}_i(t)\right)\left(1-q^{1}_k(t)\right)\sigma_i\sigma_k
				\\
				&
				+
				2 \omega_{m,i}^2 \sum_{k\neq i, p\neq i,k<p}^{n} \rho_{kp}\left(1-q^{1}_k(t)\right)\left(1-q^{1}_p(t)\right)\sigma_k\sigma_p
				+
				\omega_{m,i}^2 \sum_{k\neq i}^{m} \left(1-q^{1}_k(t)\right)^2\sigma_{k}^2\Biggr]
				\\
				&
				-
				\frac{\gamma_{i}}{2}\mathrm{e}^{2r(T-t)}
				\left[\left(\pi_{i}(t)-\omega_{m,i}\sum_{k\neq i}^{m}\pi_k(t)\right)^2\sigma^2\right]=0.
			\end{align*}
		
		Taking the first-order condition with respect to $\pi_{i}(t)$ gives
		\begin{equation*}
			\begin{aligned}
				\mathrm{e}^{r(T-t)}(\mu-r)-\gamma_{i}\mathrm{e}^{2r(T-t)}\left(\pi_{i}(t)-\omega_{m,i}\sum_{k\neq i}^{m}\pi_k(t)\right)\sigma^2=0,
			\end{aligned}
		\end{equation*}
		
		which simplifies to
		\begin{equation}
			\label{insurer-solution-investment}
			\pi_{i}(t)=\frac{\omega_{m,i}}{1+\omega_{m,i}}\sum_{v=1}^{m}\pi_v(t)+\frac{1}{1+\omega_{m,i}}\cdot\frac{\mu-r}{\gamma_{i}\mathrm{e}^{r(T-t)}\sigma^2}.
		\end{equation}

		Summing  $(\ref{insurer-solution-investment})$ over all insurers $i=1,\cdots, m$ yields
		\begin{equation}
        \label{insurer-solution-investment-1}
			\sum_{v=1}^{m}\pi_v(t)=
			\frac{\sum_{v=1}^{m}\frac{1}{(1+\omega_{m,v})\gamma_{v}}}{1-\sum_{v=1}^{m}\frac{\omega_{m,v}}{1+\omega_{m,v}}}\cdot
			\frac{\mu-r}{\mathrm{e}^{r(T-t)}\sigma^2}. 
		\end{equation}
		
		Substituting   $(\ref{insurer-solution-investment-1})$ into $(\ref{insurer-solution-investment})$ gives the optimal investment strategy
		
		\begin{equation}
			\pi_{i}^*(t)=\left[
			\frac{
				\frac{\omega_{m,i}}{1+\omega_{m,i}}\sum_{v=1}^{m}
				\frac{1}{(1+\omega_{m,v})\gamma_{v}}
				}
			{1-\sum_{v=1}^{m}\frac{\omega_{m,v}}{1+\omega_{m,v}}}
			+
			\frac{1}{(1+\omega_{m,i})\gamma_{i}}
			\right]
			\cdot
			\frac{\mu-r}{\mathrm{e}^{r(T-t)}\sigma^2}.
		\end{equation}
		
		Next, the first-order condition with respect to $q^{1}_{i}(t)$ leads to
		\begin{equation}
			\label{insurer-solution-reinsurance}
			\begin{aligned}
				-\theta^{1}_{i}(t)o_i
				+
				\mathrm{e}^{r(T-t)}\gamma_i\left[\left(1-q^{1}_i(t)\right)\sigma_i^2-\frac{w_i}{m-1}\sum_{k\neq i}^{m}\rho_{ik}\left(1-q^{1}_k(t)\right)\sigma_i\sigma_k\right]&=0,
			\end{aligned}
		\end{equation}
		
		Simplifying  (\ref{insurer-solution-reinsurance}) yields

		\begin{equation}
			\label{insurer-solution-reinsurance-moment}
			\begin{aligned}
				\left(1-q^{1}_i(t)\right)
				=
				\Lambda^{1}_{i}(t)\theta^{1}_{i}(t)
				+ 
				\Phi^{1}_{i}
				\sum_{k=1}^{m} \left(1-q^{1}_k(t)\right)\mu_k,
			\end{aligned}
		\end{equation}
	where
	\begin{equation*}
		\left\{
		\begin{aligned}
			&\Lambda^{1}_{i}(t)=\frac{1}{\left(1+\omega_{m,i} \frac{\lambda \mu_{i}^2}{\sigma_{i}^2}\right)} 
			\cdot
			\frac{ o_i}{\sigma_{i}^2 \gamma_i \mathrm{e}^{r(T-t)}},
			\\
			&\Phi^{1}_{i}=\frac{1}{\left(1+\omega_{m,i} \frac{\lambda \mu_{i}^2}{\sigma_{i}^2}\right)}
			\cdot
			\omega_{m,i}
			\frac{\lambda \mu_{i}}{\sigma_{i}^2}.
		\end{aligned}
		\right.
	\end{equation*}
	
		Multiply both sides of the $(\ref{insurer-solution-reinsurance-moment})$ by $\mu_{i}$ and summing  over $i=1,\cdots, m$ yields
		\begin{equation*}
			\begin{aligned}
				\sum_{v=1}^{m}\left(1-q^{1}_v(t)\right) \mu_{v}
				=
				\sum_{v=1}^{m} 
				\mu_{v} \Lambda^{1}_{v}(t) \theta^{1}_{v}(t)
				+ 
				\sum_{v=1}^{m} 
				\mu_{v} \Phi^{1}_{v}
				\sum_{k=1}^{m} \left(1-q^{1}_k(t)\right)\mu_k, 
			\end{aligned}
		\end{equation*}
		Solving for $\sum_{v=1}^{m}\left(1-q^{1}_v(t)\right) \mu_{v}$ yields
		\begin{equation*}
			\begin{aligned}
				\sum_{v=1}^{m}\left(1-q^{1}_v(t)\right) \mu_{v}
				=
				\frac{\sum_{v=1}^{m} 
					\mu_{v} \Lambda^{1}_{v}(t) \theta^{1}_{v}(t)}
				{1-\sum_{v=1}^{m}
					\mu_{v} \Phi^{1}_{v}}
			.
			\end{aligned}
		\end{equation*}

		The solutions to (\ref{insurer-solution-reinsurance}) are given by
        
        \begin{equation}
        	\label{insurer-solution-reinsurance-simplification}
        	\begin{aligned}
        		q^{1}_{i}(t)^*
        		&=1-G^{1}_{i}(t)-\sum_{v=1}^{m}H^{1}_{iv}(t)\theta^{1}_{v}(t),
        	\end{aligned}
        \end{equation}	
     or in vector form
        \begin{equation*}
        	\begin{aligned}
        		\vec{q}_{1}(t)^*
        		&=\mathbf{1}-\vec{G}_{1}(t)-\mathbf{H}_{1}(t)\vec{\theta}_{1}(t),
        	\end{aligned}
        \end{equation*}
		where
			\begin{align*}
				&G^{1}_{i}(t)=0, 
				\\
				&H^{1}_{ii}(t)=\Lambda^{1}_{i}(t)+ \Phi^{1}_{i} 
				\frac{ 
					\mu_{i} \Lambda^{1}_{i}(t)}
				{1-\sum_{v=1}^{m} 
					\mu_{v} \Phi^{1}_{v}}
				, 
				\\
				&H^{1}_{ik}(t)=
				\Phi^{1}_{i} 
				\frac{ 
					\mu_{k} \Lambda^{1}_{k}(t)}
				{1-\sum_{v=1}^{m} 
					\mu_{v} \Phi^{1}_{v}}. 
			\end{align*}

		 The extended HJB equation of reinsurer $R_1$ is given by  
		\begin{equation}
			\left\{
			\begin{aligned}
				&
				\sup_{u_{R_1}\in\mathcal{A}_{R_1}}\mathcal{L}^{u_{R_1}, u_{R_1-}}V^{R_1}(t, y_1)
				-\frac{\gamma_{R_1}}{2}\mathcal{L}^{u_{R_1}, u_{R_1-}}g^{R_1}(t, y_1)^2
				\\
				&\qquad\qquad
				+\gamma_{{R_1}}
				g^{R_1}(t, y_{1})\mathcal{L}^{u_{R_1}, u_{R_1-}}g^{R_1}(t, y_{1})=0,
				\\
				&
				V^{R_1}(T, y_1)=y_1,\quad
				g^{R_1}(T, y_1)=y_1,\quad
				\mathcal{L}^{u^*_{R_1}, u^*_{R_1-}}g^{R_1}(t, y_1)=0. 
			\end{aligned}
			\right.
		\end{equation}
		
		Assume that  $V^{R_1}(t, y_1)=\mathrm{e}^{r(T-t)}y_1+D_{R_1}(t), g^{R_1}(t, y_1)=\mathrm{e}^{r(T-t)}y_1+d_{R_1}(t)$. Substituting these expressions into the extended HJB equation of reinsurer $R_1$ yields
		\begin{equation*}
			\begin{aligned}
				\frac{\partial D_{R_1}}{\partial t} &+\left[\left(\mu-r\right)\pi_{R_1}(t)+\left(\sum_{k=1}^{m}\theta^{1}_{k}(t)q^{1}_k(t)-\theta^{2}_{k}(t)q^{2}_k(t)\right)o_k
				\right]
				\\
				&-\frac{\gamma_{R_1}}{2}\mathrm{e}^{r(T-t)}
				\left[\sum_{k=1}^{m}\left(q^{1}_k(t)-q^{2}_k(t)\right)^2\sigma_k^2
				+2\sum_{i< k}^{m}\rho_{ik}\left(q^{1}_i(t)-q^{2}_i(t)\right)\left(q^{1}_k(t)-q^{2}_k(t)\right)\sigma_i\sigma_k\right]
				\\
				&-\frac{\gamma_{R_1}}{2}\mathrm{e}^{r(T-t)}
				\left[\pi_{R_1}^2(t)\sigma^2\right]=0,
			\end{aligned}
		\end{equation*}
		
		Taking the first-order condition with respect to $\pi_{R_{1}}(t)$ gives
		\begin{equation}
			\pi_{R_1}^*(t)=\frac{\mu-r}{\gamma_{R_1}\mathrm{e}^{r(T-t)}\sigma^2}.
		\end{equation}
		
		The first-order condition for $q^{2}_{i}(t)$ is
		\begin{equation}
			\label{reinsurer-solution-reinsurance}
			\begin{aligned}
				-\theta^{2}_{i}(t)o_i+\mathrm{e}^{r(T-t)}\gamma_{{R_1}}\left[\left(q^{1}_i(t)-q^{2}_i(t)\right)\sigma_i^2
				+\sum_{k \neq i}^{m}\rho_{ik}\left(q^{1}_k(t)-q^{2}_k(t)\right)\sigma_i\sigma_k\right]&=0. 
			\end{aligned}
		\end{equation}
		
			Simplifying (\ref{reinsurer-solution-reinsurance}) leads to
			\begin{equation}
				\label{reinsurer-solution-reinsurance-moment}
				\begin{aligned}
					\left(q^{1}_i(t)-q^{2}_i(t)\right)
					=
					\Lambda^{2}_{i}(t)\theta^{2}_{i}(t)
					+
					\Phi^{2}_{i} 
					\sum_{k =1 }^{m}\left(q^{1}_k(t)-q^{2}_k(t)\right)\mu_{k},
				\end{aligned}
			\end{equation}
			where
			\begin{equation*}
				\left\{
				\begin{aligned}
				&\Lambda^{2}_{i}(t)
				=
				\frac{1}{\left(1-\frac{\lambda \mu_{i}^2}{\sigma_{i}^2}\right)}
				\frac{o_i}{\sigma_{i}^2\gamma_{{R_1}}\mathrm{e}^{r(T-t)}},
				\\
				&\Phi^{2}_{i}
				=
				-
				\frac{1}{\left(1-\frac{\lambda \mu_{i}^2}{\sigma_{i}^2}\right)}
				\frac{\lambda \mu_{i}}{\sigma_{i}^2}.
				\end{aligned}
				\right.
			\end{equation*}
		
		Multiplying both sides of the $(\ref{reinsurer-solution-reinsurance-moment})$ by $\mu_{i}$ and summing  over $i=1,\cdots, m$ yields
		\begin{equation*}
			\begin{aligned}
				\sum_{v=1}^{m}\left(q^{1}_{v}(t)-q^{2}_v(t)\right) \mu_{v}
				=
				\frac{\sum_{v=1}^{m} 
					\mu_{v} \Lambda^{2}_{v}(t) \theta^{2}_{v}(t)}
				{1-\sum_{v=1}^{m}
					\mu_{v} \Phi^{2}_{v}}
				.
			\end{aligned}
		\end{equation*}
		
	    The solutions to (\ref{reinsurer-solution-reinsurance-moment}) can be written as
	    \begin{equation}
	    	\label{Second Reinsurance}
	    	\begin{aligned}
	    		\vec{q}_{2}(t)
	    		=\vec{q}_{1}(t)
	    		-
	    		\mathbf{\Psi}_{2}(t)
	    			\vec{\theta}_{2}(t),
	    	\end{aligned}
	    \end{equation}
	    where
	    \begin{equation*}
	    	\left\{
	    	\begin{aligned}
	    		&\Psi^{2}_{ii}(t)=\Lambda^{2}_{i}(t)
	    		+
	    		\frac{ 
	    			\mu_{i} \Phi^{2}_{i} \Lambda^{2}_{i}(t) }
	    		{1-\sum_{v=1}^{m}
	    			\mu_{v} \Phi^{2}_{v}}
	    		, 
	    		\\
	    		&\Psi^{2}_{ik}(t)=
	    		\Phi^{2}_{i} 
	    		\frac{ 
	    			\mu_{k} \Lambda^{2}_{k}(t)}
	    		{1-\sum_{v=1}^{m} 
	    			\mu_{v} \Phi^{2}_{v}}. 
	    	\end{aligned}
	    	\right. 
	    \end{equation*}

		The first-order condition for $\theta^{1}_{i}(t)$ yields
		\begin{equation}
        \label{Second Reinsurance-1}
			\begin{aligned}
				\begin{aligned}
					&\left[
					1-G^{1}_{i}(t)-2H^{1}_{ii}(t)\theta^{1}_{i}(t)
					-
					\sum_{k\neq i}^{m}
					H^{1}_{ik}(t)\theta^{1}_{k}(t)
					\right]o_i
					+\left[-\sum_{k\neq i}^{m}H^{1}_{ki}(t)\theta^{1}_{k}(t)\right]o_k
					\\
					&+\mathrm{e}^{r(T-t)}\gamma_{R_1}
					\Bigg[
					\left(q^{1}_{i}(t)-q^{2}_{i}(t)\right)
					\left(H^{1}_{ii}(t)\sigma_i^2
					+
					\sum_{k\neq i}^{m}
					\rho_{ki} H^{1}_{ki}(t)\sigma_i \sigma_k\right)
					\\
					&
					+
					\sum_{k \neq i}^{m}
					\left(q^{1}_{k}(t)-q^{2}_{k}(t)\right)\left(H^{1}_{ki}(t)\sigma_k^2+
					\sum_{v\neq k}^{m}\rho_{kv}H^{1}_{vi}(t)
                    \sigma_k\sigma_v\right)\Bigg]=0,
				\end{aligned}
			\end{aligned}
		\end{equation}

		Substituting (\ref{Second Reinsurance}) into (\ref{Second Reinsurance-1}) gives 
			\begin{align*}
				&
				2H^{1}_{ii}(t)o_i\theta^{1}_{i}(t)
				+
				\sum_{k\neq i}^{m}
				\left[
				H^{1}_{ik}(t)o_{i}+H^{1}_{ki}(t)o_{k}\right]
				\theta^{1}_{k}(t)
				\\
				=&
				(1-G^{1}_{i}(t))o_i
				\\
				&
				+
				\mathrm{e}^{r(T-t)}\gamma_{R_1}
				\Bigg\{
				\Bigg[ \Psi^{2}_{ii}(t)
				\left(H^{1}_{ii}(t)\sigma_i^2
				+
				\sum_{k\neq i}^{m}
				\rho_{ki} H^{1}_{ki}(t)\sigma_i \sigma_k\right)
				\\
				&
				+ 
				 \sum_{k\neq i}^{m}
				 \Psi^{2}_{ki}(t)
				 	\left(H^{1}_{ki}(t)\sigma_k^2+
				 	\sum_{v\neq k}^{m}\rho_{kv}H_{vi}(t)
				 	\sigma_k\sigma_v\right)
				\Bigg]
				\theta^{2}_{i}(t)
				\\
				&
				+
				\sum_{k \neq i}^{m} 				
				\Bigg[
				\Psi^{2}_{ik}(t)
				\left(H^{1}_{ii}(t)\sigma_i^2
				+
				\sum_{k\neq i}^{m}
				\rho_{ki} H^{1}_{ki}(t)\sigma_i \sigma_k\right)
				\\
				&
				+
				\sum_{v \neq i}^{m} \Psi^{2}_{vk}
				\left(H^{1}_{vi}(t)\sigma_v^2+
				\sum_{l\neq v}^{m}\rho_{vl}H^{1}_{li}(t)
				\sigma_l\sigma_v\right)
				\Bigg]
				\theta^{2}_{k}(t)
				\Bigg\},
			\end{align*}
		which can be written in matrix form as
		\begin{equation*}
			\mathbf{A}_{1}\vec{\theta}_{1}(t)=\vec{C}_{1}+\mathbf{B}_{1}\vec{\theta}_{2}(t),
		\end{equation*}
		where

			\begin{align*}
				&A^{1}_{ii}=2H^{1}_{ii}(t)o_i, 
				\\
				&A^{1}_{ik}=
				\left[
				H^{1}_{ik}(t)o_{i}+H^{1}_{ki}(t)o_{k}\right], 
				\\
				&C^{1}_{i}=(1-G^{1}_{i}(t))o_i, 
				\\
				&B^{1}_{ii}=\mathrm{e}^{r(T-t)}\gamma_{R_1}
				\Bigg[ \Psi^{2}_{ii}(t)
				\left(H^{1}_{ii}(t)\sigma_i^2
				+
				\sum_{k\neq i}^{m}
				\rho_{ki} H^{1}_{ki}(t)\sigma_i \sigma_k\right)
				\\
				&
				\qquad
				+ 
				\sum_{k\neq i}^{m}
				\Psi^{2}_{ki}(t)
				\left(H^{1}_{ki}(t)\sigma_k^2+
				\sum_{v\neq k}^{m}\rho_{kv}H^{1}_{vi}(t)
				\sigma_k\sigma_v\right)
				\Bigg], 
				\\
				&
				B^{1}_{ik}
				=
				\sum_{k \neq i}^{m} 				
				\Bigg[
				\Psi^{2}_{ik}(t)
				\left(H^{1}_{ii}(t)\sigma_i^2
				+
				\sum_{k\neq i}^{m}
				\rho_{ki} H^{1}_{ki}(t)\sigma_i \sigma_k\right)
				\\
				&
				\qquad
				+
				\sum_{v \neq i}^{m} \Psi^{2}_{vk}
				\left(H^{1}_{vi}(t)\sigma_v^2+
				\sum_{l\neq v}^{m}\rho_{vl}H^{1}_{li}(t)
				\sigma_l\sigma_v\right)
				\Bigg]. 
			\end{align*}
		
		It is evident that $\mathbf{A}_{1}$ is a symmetric matrix. Assuming that $\mathbf{A}_{1}$ is invertible, the solution for  $\theta^{1}_{i}(t)$ is  given by
		\begin{equation}
			\label{Chain-first Reinsurance-theta}
			\begin{aligned}
				&
				\vec{\theta}_{1}(t)^*=\left(\mathbf{A}_{1}\right)^{-1}\vec{C}_{1}+\left(\mathbf{A}_{1}\right)^{-1}\mathbf{B}_{1}\vec{\theta}_{2}(t), 
			\end{aligned}
		\end{equation}
        
		Substituting \eqref{Chain-first Reinsurance-theta} into (\ref{Second Reinsurance}) yields
		\begin{equation*}
			\begin{aligned}
				\vec{q}_{2}(t)^*=\mathbf{1}-\vec{G}_{2}(t)-\mathbf{H}_{2}(t)\vec{\theta}_{2}(t), 
			\end{aligned}
		\end{equation*}
		where
		\begin{equation*}
			\left\{
			\begin{aligned}
				&\vec{G}_{2}(t)=\vec{G}_{1}(t)+\mathbf{H}_{1}(t)\left(\mathbf{A}_{1}\right)^{-1}
				\vec{C_{1}},
				\\
				&
				\mathbf{H}_{2}(t)=\mathbf{H}_{1}(t)\left(\mathbf{A}_{1}\right)^{-1}\mathbf{B}_{1}+\mathbf{\Psi_{2}(t)}.
			\end{aligned}
			\right.
		\end{equation*}
		
		Similarly, the optimal investment strategy for  reinsurer $R_j$ is 
		\begin{equation}
			\pi_{R_j}^*(t)=\frac{\mu-r}{\gamma_{R_j}\mathrm{e}^{r(T-t)}\sigma^2}.
		\end{equation}
		
		The first-order condition for $q^{j+1}_{i}(t)$ yields
		
		\begin{equation}
			\label{Reinsurance ratio-recursion formula}
			\begin{aligned}
				q^{j+1}_{i}(t)
				=q^{j}_i(t)
				-
				\Psi^{j+1}_{ii}(t)
				\theta^{j+1}_{i}(t)
				-
				\sum_{k \neq i}^{m} 
				\Psi^{j+1}_{ik}(t)
				\theta^{j+1}_{k}(t). 
			\end{aligned}
		\end{equation}
		where
		\begin{equation*}
			\left\{
			\begin{aligned}
				&\Psi^{j+1}_{ii}(t)=\Lambda^{j+1}_{i}(t)
				+
				\frac{ 
					\mu_{i} \Phi^{j+1}_{i} \Lambda^{j+1}_{i}(t) }
				{1-\sum_{v=1}^{m}
					\mu_{v} \Phi^{j+1}_{v}}
				, 
				\\
				&\Psi^{j+1}_{ik}(t)=
				\Phi^{j+1}_{i} 
				\frac{ 
					\mu_{k} \Lambda^{j+1}_{k}(t)}
				{1-\sum_{v=1}^{m} 
					\mu_{v} \Phi^{j+1}_{v}}. 
			\end{aligned}
			\right. 
		\end{equation*}
		
		The first-order condition for $\theta^{j}_{i}(t)$ yields 
	\begin{equation}
		\label{Chain-reinsurance R_j-theta}
		\vec{\theta}_{j}(t)^*=\left(\mathbf{A}_{j}\right)^{-1}\vec{C}_{j}
		+
		\left(\mathbf{A}_{j}\right)^{-1}\mathbf{B}_{j}\vec{\theta}_{j+1}(t),
	\end{equation}
	where
    	\begin{align*}
    		&A^{j}_{ii}=2H^{j}_{ii}(t)o_i, 
    		\\
    		&A^{j}_{ik}=
    		\left[
    		H^{j}_{ik}(t)o_{i}+H^{j}_{ki}(t)o_{k}\right], 
    		\\
    		&C^{j}_{i}=(1-G^{j}_{i}(t))o_i, 
    		\\
    		&B^{j}_{ii}=\mathrm{e}^{r(T-t)}\gamma_{R_j}
    		\Bigg[ \Psi^{j+1}_{ii}(t)
    		\left(H^{j}_{ii}(t)\sigma_i^2
            +
    		\sum_{k\neq i}^{m}
    		\rho_{ki} H^{j}_{ki}(t)\sigma_i \sigma_k\right)
    		\\
    		&
    		\qquad
    		+ 
    		\sum_{k\neq i}^{m}
    		\Psi^{j+1}_{ki}(t)
    		\left(H^{j}_{ki}(t)\sigma_k^2+
    		\sum_{v\neq k}^{m}\rho_{kv}H^{j}_{vi}(t)
    		\sigma_k\sigma_v\right)
    		\Bigg], 
    		\\
    		&
    		B^{j}_{ik}
    		=
    		\sum_{k \neq i}^{m} 				
    		\Bigg[
    		\Psi^{j+1}_{ik}(t)
    		\left(H^{j}_{ii}(t)\sigma_i^2
    		+
    		\sum_{k\neq i}^{m}
    		\rho_{ki} H^{j}_{ki}(t)\sigma_i \sigma_k\right)
    		\\
    		&
    		\qquad
    		+
    		\sum_{v \neq i}^{m} \Psi^{j+1}_{vk}
    		\left(H^{j}_{vi}(t)\sigma_v^2+
    		\sum_{l\neq v}^{m}\rho_{vl}H^{j}_{li}(t)
    		\sigma_l\sigma_v\right)
    		\Bigg]. 
    	\end{align*}
	  Substituting  $(\ref{Chain-reinsurance R_j-theta})$  into (\ref{Reinsurance ratio-recursion formula})	yields 
		\begin{equation*}
			\begin{aligned}
				\vec{q_{j+1}}(t)^*=\mathbf{1}-\vec{G}_{j+1}(t)-\mathbf{H}_{j+1}(t)\vec{\theta}_{j+1}(t), 
			\end{aligned}
		\end{equation*}
		where
		\begin{equation*}
			\begin{aligned}
				&\vec{G}_{j+1}(t)=\vec{G}_{j}(t)+\mathbf{H}_{j}(t)\left(\mathbf{A}_{j}\right)^{-1}
				\vec{C}_{j},
				\\
				&
				\mathbf{H}_{j+1}(t)=\mathbf{H}_{j}(t)\left(\mathbf{A}_{j}\right)^{-1}\mathbf{B}_{j}+\mathbf{\Psi_{j+1}(t)}.
			\end{aligned}
		\end{equation*}
		
		Similarly, the first-order condition of $\theta^{n}_{i}(t)$ for reinsurer $R_{n}$ yields
				\begin{align*}
					&\left[
					1-G^{n}_{i}(t)-2H^{n}_{ii}(t)\theta^{n}_{i}(t)
					-
					\sum_{k\neq i}^{m}
					H^{n}_{ik}(t)\theta^{n}_{k}(t)
					\right]o_i
					+\left[-\sum_{k\neq i}^{m}H^{n}_{ki}(t)\theta^{n}_{k}(t)\right]o_k
					\\
					&+\mathrm{e}^{r(T-t)}\gamma_{R_n}
					\Bigg[
					q^{n}_{i}(t)
					\left(H^{n}_{ii}(t)\sigma_i^2
					+
					\sum_{k\neq i}^{m}
					\rho_{ki} H^{n}_{ki}(t)\sigma_i \sigma_k\right)
					\\
					&
					+
					\sum_{k \neq i}^{m}
					q^{n}_{k}(t)
					\left(H^{n}_{ki}(t)\sigma_k^2+
					\sum_{v\neq k}^{m}\rho_{kv}H^{n}_{vi}(t)
					\sigma_k\sigma_v\right)\Bigg]=0,
				\end{align*}
		
		Substituting $\vec{q_{n}}(t)^*$ into above equations yields
		\begin{equation*}
			\begin{aligned}
				\vec{\theta}_{n}(t)^*=\mathbf{A}_{n}^{-1}\vec{C}_{n},
			\end{aligned}
		\end{equation*}
		where
			\begin{align*}
				&A^{n}_{ii}=2H^{n}_{ii}(t)o_i+
				\mathrm{e}^{r(T-t)}\gamma_{R_n}
				\Bigg[ H^{n}_{ii}(t)
				\left(H^{n}_{ii}(t)\sigma_i^2
				+
				\sum_{k\neq i}^{m}
				\rho_{ki} H^{n}_{ki}(t)\sigma_i \sigma_k\right)
				\\
				&
				\qquad
				+ 
				\sum_{k\neq i}^{m}
				H^{n}_{ki}(t)
				\left(H^{n}_{ki}(t)\sigma_k^2+
				\sum_{v\neq k}^{m}\rho_{kv}H^{n}_{vi}(t)
				\sigma_k\sigma_v\right)
				\Bigg]
				, 
				\\
				&A^{n}_{ik}=
				\left[
				H^{n}_{ik}(t)o_{i}+H^{n}_{ki}(t)o_{k}\right]
				+
				\mathrm{e}^{r(T-t)}\gamma_{R_n}
				\sum_{k \neq i}^{m} 				
				\Bigg[
				H^{n}_{ik}(t)
				\left(H^{n}_{ii}(t)\sigma_i^2
				+
				\sum_{k\neq i}^{m}
				\rho_{ki} H^{n}_{ki}(t)\sigma_i \sigma_k\right)
				\\
				&
				\qquad
				+
				\sum_{v \neq i}^{m} H^{n}_{vk}
				\left(H^{n}_{vi}(t)\sigma_v^2+
				\sum_{l\neq v}^{m}\rho_{vl}H^{n}_{li}(t)
				\sigma_l\sigma_v\right)
				\Bigg]
				, 
				\\
				&C^{n}_{i}=(1-G^{n}_{i}(t))o_i
				+
				\mathrm{e}^{r(T-t)}\gamma_{R_n}
				\Bigg[
				  \left(1-G^{n}_{i}(t)\right)\left(H^{n}_{ii}(t)\sigma_i^2
				  +
				  \sum_{k\neq i}^{m}
				  \rho_{ki} H^{n}_{ki}(t)\sigma_i \sigma_k\right)
				  \\
				  &
				  \qquad
				  +
				  \sum_{k\neq i}^{m}
				  \left(1-G^{n}_{k}(t)\right)
				  \left(H^{n}_{ki}(t)\sigma_k^2+
				  \sum_{v\neq k}^{m}\rho_{kv}H^{n}_{vi}(t)
				  \sigma_k\sigma_v\right)
				\Bigg]
				.
			\end{align*}

	Finally, we solve for the value functions. The  insurer $i$ satisfies the extended HJB equation, and with the terminal condition $D_{i}(T)=0$, we obtain
			\begin{align*}
				&D_{i}(t)
				\\
				=
				&-\int^{T}_{t}
				\mathrm{e}^{r(T-s)}
				\Bigg[(\mu-r)\left(\pi_{i}(s)
				-
				\omega_{m,i}\sum_{k\neq i}^{m} \pi_{k}(s)\right)
				\\
				&
				+
				\left(\eta_{i}-q^{1}_{i}(s)\theta^{1}_{i}(s)\right)o_{i}
				-
				\omega_{m,i}\sum_{k\neq i}^{m} \left(\eta_{k}-q^1_k(s)\theta^{1}_{k}(s)\right)o_k
				\Bigg]\mathrm{d}s
				\\
				&-\frac{\gamma_{i}}{2}
				\int^{T}_{t} \mathrm{e}^{2r(T-s)}
				\Biggl[
				\left(1-q^{1}_{i}(s)\right)^2\sigma_{i}^2
				-
				2\omega_{m,i} \sum_{k\neq i}^{m} \rho_{ik}\left(1-q^{1}_i(s)\right)\left(1-q^{1}_k(s)\right)\sigma_i\sigma_k
				\\
				&
				+
				2\omega_{m,i}\sum_{k\neq i, p\neq i,k<p}^{m} \rho_{kp}\left(1-q^{1}_k(s)\right)\left(1-q^{1}_p(s)\right)\sigma_k\sigma_p
				+
				\frac{w_i^2}{m-1} \sum_{k\neq i}^{m} \left(1-q^{1}_k(s)\right)^2\sigma_{k}^2\Biggr]
				\mathrm{d}s
				\\
				&-\frac{\gamma_{i}}{2}
				\int^{T}_{t} \mathrm{e}^{2r(T-s)}
				\left[\left(\pi_{i}(s)-\sum_{k\neq i}^{m}\omega_{m,i}\pi_k(s)\right)^2\sigma^2\right]\mathrm{d}s.
			\end{align*}
		
		Similarly, $D_{R_j}(t)$ is given by 
			\begin{align*}
				&D_{R_{j}}(t)
				\\
				=&
				-\int^{T}_{t}
				\mathrm{e}^{r(T-s)}
				\left[\left(\mu-r\right)\pi_{R_j}(s)
				+
				\left(\sum_{k=1}^{m}\theta^{j}_{k}(s)q^{j}_k(s)-\theta^{j+1}_{k}(s)q^{j+1}_k(s)\right)o_k
				\right]
				\\
				&-\frac{\gamma_{R_{j}}}{2}\mathrm{e}^{2r(T-s)}
				\left[\sum_{k=1}^{m}\left(q^{j}_k(s)-q^{j+1}_k(s)\right)^2\sigma_k^2
				+2\sum_{i< k}^{m}\rho_{ik}\left(q^{j}_i(s)-q^{j+1}_{i}(s)\right)\left(q^{j}_k(s)-q^{j+1}_k(s)\right)\sigma_i\sigma_k\right]
				\\
				&-\frac{\gamma_{R_j}}{2}\mathrm{e}^{2r(T-s)}
				\left[\pi_{R_j}^2(s)\sigma^2\right]\mathrm{d}s.
			\end{align*}
\end{proof}

\subsection{Proof of \cref{Lemma 3.1}}
\label{proof of Lemma 3.1}

\begin{proof}
		\begin{align*}
			\mathbb{E}\left[l_j(Z^{i},t)\right]
			&=
			\int_{h^{j-1}_{i}(t)}^{h^{j}_{i}(t)}\left(z-h^{j-1}_{i}(t)\right)f_{i}(z)\mathrm{d}z
			+
			\int_{h^{j}_{i}(t)}^{\infty}a^{j}_{i}(t)f_{i}(z)dz
			\\
			&=\int_{h^{j-1}_{i}(t)}^{h^{j}_{i}(t)}\left(z-h^{j}_{i}(t)\right)f_{i}(z)\mathrm{d}z
			+
			\int_{h^{j-1}_{i}(t)}^{\infty}a^{j}_{i}(t)f_{i}(z)dz
			\\
			&=
			\int_{h^{j-1}_{i}(t)}^{\infty}a^{j}_{i}(t)f_{i}(z)dz
			-
			\int_{h^{j-1}_{i}(t)}^{h^{j}_{i}(t)}
			\left(\int_{z}^{h^{j}_{i}(t)}\mathrm{d}x\right) f_{i}(z)\mathrm{d}z
			\\
			&=
			\int_{h^{j-1}_{i}(t)}^{h^{j}_{i}(t)}\overline{F}_{i}\left(h^{j-1}_{i}(t)\right)\mathrm{d}z
			-
			\int_{h^{j-1}_{i}(t)}^{h^{j}_{i}(t)}
			\left(\int_{h^{j-1}_{i}(t)}^{x}f_{i}(z)\mathrm{d}z\right)
			\mathrm{d}x
			\\
			&=
			\int_{h^{j-1}_{i}(t)}^{h^{j}_{i}(t)}\overline{F}_{i}\left(h^{j-1}_{i}(t)\right)\mathrm{d}z
			-
			\int_{h^{j-1}_{i}(t)}^{h^{j}_{i}(t)}
			\left(F_{i}(x)-F_{i}\left(h^{j-1}_{i}(t)\right)\right)
			\mathrm{d}x
			\\
			&=\int_{h^{j-1}_{i}(t)}^{h^{j}_{i}(t)}\overline{F}_{i}(z)dz .
		\end{align*}

		\begin{align*}
			\mathbb{E}\left[l_j(Z^{i}, t)^2\right]
			&=
			\int_{h^{j-1}_{i}(t)}^{h^{j}_{i}(t)}\left(z-h^{j-1}_{i}(t)\right)^2f(z)\mathrm{d}z
			+
			\int_{h^{j}_{i}(t)}^{\infty}\left(a^j_{i}(t)\right)^2 f(z) \mathrm{d}z
			\\
			&=
			\int_{h^{j-1}_{i}(t)}^{\infty} \left(a^j_{i}(t)\right)^2 f(z)dz
			-
			\int_{h^{j-1}_{i}(t)}^{h^{j}_{i}(t)}
			\left[
			\left(a^j_{i}(t)\right)^2-\left(z-h^{j-1}_{i}(t)\right)^2
			\right]f(z)\mathrm{d}z
			\\
			&=
			\int_{h^{j-1}_{i}(t)}^{h^{j}_{i}(t)}2\left(x-h^{j-1}_{i}(t)\right)
			\mathrm{d}x
			\overline{F}\left({h^{j-1}_{i}(t)}\right)
			\\
			&\quad
			-
			\int_{h^{j-1}_{i}(t)}^{h^{j}_{i}(t)}
			\left(\int_{z}^{h^{j}_{i}(t)}2\left(x-h^{j-1}_{i}(t)\right)
			\mathrm{d}x
			\right)
			f(z)\mathrm{d}z
			\\
			&=
			\int_{h^{j-1}_{i}(t)}^{h^{j}_{i}(t)}2\left(x-h^{j-1}_{i}(t)\right)
			\mathrm{d}x
			\overline{F}\left({h^{j-1}_{i}(t)}\right)
			\\
			&\quad
			-
			\int_{h^{j-1}_{i}(t)}^{h^{j}_{i}(t)}
			2\left(x-h^{j-1}_{i}(t)\right)
			\left(\int_{h^{j-1}_{i}(t)}^{x}
			f(z)\mathrm{d}z
			\right)
			\mathrm{d}x
			\\
			&=
			\int_{h^{j-1}_{i}(t)}^{h^{j}_{i}(t)}2\left(x-h^{j-1}_{i}(t)\right)
			\mathrm{d}x
			\left[
			\overline{F}\left({h^{j-1}_{i}(t)}\right)
			-
			F(x)+F\left(h^{j-1}_{i}(t)\right)
			\right]
			\\
			&=
			\int_{h^{j-1}_{i}(t)}^{h^{j}_{i}(t)}2\left(z-h^{j-1}_{i}(t)\right)\overline{F}(x)
			\mathrm{d}x.
		\end{align*}
\end{proof}

\subsection{Proof of \cref{Theorem-excess-of-loss reinsurance-n reinsurer}}
\label{proof-excess-of-loss reinsurance-n reinsurer}
\begin{proof}
	The extended HJB equation for the insurer $i$ is as follows
	\begin{equation}
		\label{insurer-excess-of-loss reinsurance-HJB-n insurers}
		\left\{
		\begin{aligned}
			&
			\sup_{u_{i}\in\mathcal{A}_{i}}
			\left\{
			\mathcal{L}^{u_{i}, u_{i-}}V^{i}(t, x_{i})
			-\frac{\gamma_{i}}{2}\mathcal{L}^{u_{i}, u_{i-}}g^{i}(t, x_{i})^2+\gamma_{i}
			g^{i}(t, x_{i})\mathcal{L}^{u_{i}, u_{i-}}g^{i}(t, x_{i})
			\right\}=0,
			\\
			&
			V^{i}(T, x_{i})=x_{i},\quad  
			g^{i}(T, x_{i})=x_{i},\quad  
			\mathcal{L}^{u^*_{i}, u^*_{i-}}g^{i}(t, x_{i})=0.  
		\end{aligned}
		\right.
	\end{equation}
	
	Assume
	\begin{equation}
		\label{insurer-excess-of-loss reinsurance-value function}
		V^{i}(t, x_{i})=\mathrm{e}^{r(T-t)}x_{i}+D_{i}(t),\quad g^{i}(t, x_{i})=\mathrm{e}^{r(T-t)}x_{i}+d_{i}(t).
	\end{equation}
	
	For notational simplicity, we denote $a^{j}_{i}(t)$ and $h^{j}_{i}(t)$ by $a^{j}_{i}$ and $h^{j}_{i}$,  respectively. 
	Substituting (\ref{insurer-excess-of-loss reinsurance-value function}) into (\ref{insurer-excess-of-loss reinsurance-HJB-n insurers}) yields
		\begin{align*}
			\frac{\partial M_{i}}{\partial t}
			&+\mathrm{e}^{r(T-t)}\left[(\mu-r)\left(\pi_{i}(t)-\omega_{m,i} \sum_{k\neq i}^{m} \pi_{k}(t)\right)
			\right]
			-
			\frac{\gamma_{i}}{2}\mathrm{e}^{2r(T-t)}
			\left[\left(\pi_{i}(t)-\omega_{m,i}\sum_{k\neq i}^{m}\pi_k(t)\right)^2\sigma^2\right]
			\\
			&
			+\mathrm{e}^{r(T-t)}
			\left[
			\left(\eta_{i}-\theta^{1}_{i}(t)\right)(\lambda+\lambda_{i})\mu_{i}
			+(1+\theta^{1}_{i}(t))(\lambda+\lambda_{i})\int_{0}^{a^{1}_{i}}\overline{F}_{i}(z)\mathrm{d}z
			\right]
			\\
			&
			-\mathrm{e}^{r(T-t)} \omega_{m,i} \sum_{k\neq i}^{m}
			\left[
			\left(\eta_{k}-\theta^{1}_{k}(t)\right)(\lambda+\lambda_{k})\mu_{k}
			+(1+\theta^{1}_{k}(t))(\lambda+\lambda_{k})\int_{0}^{a^{1}_{k}}\overline{F}_{k}(z)\mathrm{d}z
			\right]
			\\
			&-\mathrm{e}^{r(T-t)}\lambda_{i}\int_{0}^{a^{1}_{i}}\overline{F}_{i}(z)\mathrm{d}z
			-\mathrm{e}^{2r(T-t)}\frac{\gamma_{i}}{2}\lambda_{i}\int_{0}^{a^{1}_{i}}2z\overline{F}_{i}(z)\mathrm{d}z
			\\
			&
			+
			\mathrm{e}^{r(T-t)} \omega_{m,i} \sum_{k\neq i}^{m} \lambda_{k}  \int_{0}^{a^{1}_{k}}\overline{F}_{k}(z)\mathrm{d}z
			-
			\mathrm{e}^{2r(T-t)}\frac{\gamma_{i}}{2}
			\sum_{k\neq i}^{m}
			\lambda_{k}\int_{0}^{a^{1}_{k}}2z\overline{F}_{k}(z)\mathrm{d}z
			\\
			&-\mathrm{e}^{r(T-t)}\lambda\left(\int_{0}^{a^{1}_{i}}\overline{F}_{i}(z)\mathrm{d}z-\omega_{m,i}
			\sum_{k\neq i}^{m}
			 \int_{0}^{a^{1}_{k}}\overline{F}_{k}(z)\mathrm{d}z\right)
			 \\
			 &
			-
			\mathrm{e}^{2r(T-t)}\frac{\gamma_{i}}{2}\lambda
			\Biggr[
			\int_{0}^{a^{1}_{i}}2z\overline{F}_{i}(z)\mathrm{d}z
			-
			2\omega_{m,i}
			\sum_{k\neq i}^{m}
			\int_{0}^{a^{1}_{i}}\overline{F}_{i}(z)\mathrm{d}z\int_{0}^{a^{1}_{k}}\overline{F}_{k}(z)\mathrm{d}z
			\\
			&
			+
			2 \omega_{m,i}^2 \sum_{k\neq i, p\neq i,k<p}^{m}
			\int_{0}^{a^{1}_{k}}\overline{F}_{k}(z)\mathrm{d}z\int_{0}^{a^{1}_{p}}\overline{F}_{p}(z)\mathrm{d}z
			+
			\omega_{m,i}^2
			\sum_{k\neq i}^{m}
			\int_{0}^{a^{1}_{k}}2z\overline{F}_{k}(z)\mathrm{d}z
			\Biggr]=0.
		\end{align*}
	
		The first-order condition with respect to $a^{1}_{i}(t)$ yields
	\begin{equation*}
		a^{1}_{i}(t)^*=\frac{\theta^{1}_{i}(t)}{\gamma_{i}\mathrm{e}^{r(T-t)}}+\frac{\lambda \omega_{m,i}}{\lambda+\lambda_{i}}
		\sum_{k\neq i}^{m}
		\int_{0}^{a^{1}_{k}(t)}\overline{F}_{k}(z)\mathrm{d}z. 
	\end{equation*}
	
	The extended HJB equation for  reinsurer $R_{1}$ is as follows
	\begin{equation}
		\left\{
		\label{reinsurer-excess-of-loss reinsurance-HJB}
		\begin{aligned}
			&\sup_{u_{R_1}\in\mathcal{A}_{R_1}}
			\Bigl\{
			\mathcal{L}^{u_{R_1}, u_{R_1-}}V^{R_1}(t, y_1)
			-\frac{\gamma_{R_1}}{2}\mathcal{L}^{u_{R_1}, u_{R_1-}}g^{R_1}(t, y_1)^2
			\\
			&\qquad\qquad
			+
			\gamma_{{R_1}}
			g^{R_1}(t, y_{1})\mathcal{L}^{u_{R_1}, u_{R_1-}}g^{R_1}(t, y_{1})
			\Bigr\}
			=0,
			\\
			&
			V^{R_1}(T, y_1)=y_1,\quad
			g^{R_1}(T, y_1)=y_1,\quad
			\mathcal{L}^{u^*_{R_1},  u^*_{R_1-}}g^{R_1}(t, y_1)
			=0. 
		\end{aligned}
		\right. 
	\end{equation}
	
	Assume 
	\begin{equation}
		\label{reinsurer-excess-of-loss reinsurance-value function}
		V^{R_1}(t, y_1)=\mathrm{e}^{r(T-t)}y_1+D_{R_1}(t), g^{R_1}(t, y_1)=\mathrm{e}^{r(T-t)}y_1+d_{R_1}(t). 
	\end{equation} 
	
	Substituting (\ref{reinsurer-excess-of-loss reinsurance-value function}) into (\ref{reinsurer-excess-of-loss reinsurance-HJB}) yields
		\begin{align*}
			\frac{\partial M_{R_{1}}}{\partial t}
			&+\mathrm{e}^{r(T-t)}
			(\mu-r)\pi_{R_{1}}(t)
			-
			\frac{\gamma_{R_{1}}}{2}\mathrm{e}^{2r(T-t)}
			\pi_{R_{1}}(t)^2\sigma^2
			\\
			&
			+\mathrm{e}^{r(T-t)}
			\sum_{v=1}^{m}
			\Big[
			\left(\theta^{1}_{v}(t)-\theta^{2}_{v}(t)\right)(\lambda+\lambda_{v})\mu_{v}
			-(1+\theta^{1}_{v}(t))(\lambda+\lambda_{v})\int_{0}^{a^{1}_{v}}\overline{F}_{v}(z)\mathrm{d}z
			\\
			&+(1+\theta^{2}_{v}(t))(\lambda+\lambda_{v})\int_{0}^{a^{1}_{v}+a^{2}_{v}}\overline{F}_{v}(z)\mathrm{d}z
			\Big]
			\\
			&-\mathrm{e}^{r(T-t)}\sum_{v=1}^{m}\lambda_{v}
			\left[
			\int_{a^{1}_{v}}^{a^{1}_{v}+a^{2}_{v}}\overline{F}_{v}(z)\mathrm{d}z
			\right]
			-
			\mathrm{e}^{2r(T-t)} \frac{\gamma_{R_1}}{2}
			\sum_{v=1}^{m}
			\lambda_{v}
			\int_{a^{1}_{v}}^{a^{1}_{v}+a^{2}_{v}}2(z-a^{1}_{v})\overline{F}_{v}(z)\mathrm{d}z
			\\
			&-
			\mathrm{e}^{r(T-t)} \lambda \sum_{v=1}^{m}\left(\int_{a^{1}_{v}}^{a^{1}_{v}+a^{2}_{v}}\overline{F}_{v}(z)\mathrm{d}z\right)
			-
			\mathrm{e}^{2r(T-t)}\frac{\gamma_{R_1}}{2}\lambda
			\bigg[
			\sum_{v=1}^{m}
			\int_{a^{1}_{v}}^{a^{1}_{v}+a^{2}_{v}}2(z-a^{1}_{v})\overline{F}_{v}(z)\mathrm{d}z
			\\
			&
			+
			2\sum_{k<p}^{m}
			\int_{a^{1}_{k}}^{a^{1}_{k}+a^{2}_{k}}\overline{F}_{k}(z)\mathrm{d}z
			\int_{a^{1}_{p}}^{a^{1}_{p}+a^{2}_{p}}\overline{F}_{p}(z)\mathrm{d}z
			\bigg]
			=0.
		\end{align*}
	
	The first-order condition with respect to $a^{2}_{i}(t)$ yields
	\begin{equation*}
		\begin{aligned}
			a^{2}_i(t)^*&=
			\frac{\theta^{2}_{i}(t)}{\gamma_{R_1}\mathrm{e}^{r(T-t)}}
			-
			\frac{\lambda}{\lambda+\lambda_{i}}
			\sum_{k\neq i}^{m}
			\int_{a^{1}_{k}(t)}^{a^{1}_{k}(t)+a^{2}_{k}(t)}\overline{F}_{k}(z)\mathrm{d}z.
		\end{aligned}
	\end{equation*}
	
     The first-order condition with respect to $\theta^{1}_{i}(t)$ yields
		\begin{align*}
			&\Bigg[
			(\lambda+\lambda_{i})\mu_{i}-(\lambda+\lambda_{i})\int_{0}^{a^{1}_{i}}\overline{F}_{i}(z)\mathrm{d}z
			\\
			&
			-
			\sum_{v=1}^{m}\theta^{1}_{v}(t)(\lambda+\lambda_{v})
			\overline{F}_{v}(a^{1}_{v})
			\frac{\partial a^{1}_{v}}{\partial \theta^{1}_{i}(t)}
			+
			\sum_{v=1}^{m}
			\theta^{2}_{v}(t)(\lambda+\lambda_{v})\overline{F}_{v}(a^{1}_{v}+a^{2}_{v})
			\frac{\partial a^{1}_{v}}{\partial \theta^{1}_{i}(t)}
			\Bigg]
			\\
			&
			-\mathrm{e}^{r(T-t)}\gamma_{R_1}
			\sum_{v=1}^{m}(\lambda_{v}+\lambda)
			\left[
			a^{2}_{v}\overline{F}_{v}(a^{1}_{v}+a^{2}_{v})
			-\int_{a^{1}_{v}}^{a^{1}_{v}+a^{2}_{v}}\overline{F}_{v}(z)\mathrm{d}z
			\right]\frac{\partial a^{1}_{v}}{\partial \theta^{1}_{i}(t)}
			\\
			&
			-\mathrm{e}^{r(T-t)}\gamma_{R_1}\lambda
			\Bigg[
			\sum_{k<p}^{m}
			\left(
			\overline{F}_{k}(a^{1}_{k}+a^{2}_{k})-\overline{F}_{k}(a^{1}_{k})
			\right)
			\frac{\partial a^{1}_{k}}{\partial \theta^{1}_{i}(t)}
			\int_{a^{1}_{p}}^{a^{1}_{p}+a^{2}_{p}}\overline{F}_{p}(z)\mathrm{d}z
			\\
			&
			+
			\sum_{k<p}^{m}
			\int_{a^{1}_{k}}^{a^{1}_{k}+a^{2}_{k}}\overline{F}_{k}(z)\mathrm{d}z
			\left(
			\overline{F}_{p}(a^{1}_{p}+a^{2}_{p})-\overline{F}_{p}(a^{1}_{p})
			\right)
			\frac{\partial a^{1}_{p}}{\partial \theta^{1}_{i}(t)}
			\Bigg]
			=0. 
		\end{align*}

	Similarly, for  reinsurer $R_{j}(j=1,\cdots,n-1)$, the optimal retention level $a^{j+1}_{i}(t)^*$ satisfies
	\begin{equation*}
		\begin{aligned}
			a^{j+1}_{i}(t)^*
			&=\frac{\theta^{j+1}_{i}(t)}{\gamma_{R_j}\mathrm{e}^{r(T-t)}}
			-
			\frac{\lambda}{\lambda+\lambda_{i}}\int_{h^{j}_{k}(t)}^{h^{j+1}_{k}(t)}\overline{F}_{k}(z)\mathrm{d}z.
		\end{aligned}
	\end{equation*}
	
	The optimal safety loading $\theta^{j}_{i}(t)^*$ satisfies
		\begin{align*}
			&\Bigg[
			(\lambda+\lambda_{i})\mu_{i}-(\lambda+\lambda_{i})\int_{0}^{h^{j}_{i}}\overline{F}_{i}(z)\mathrm{d}z
			\\
			&
			-
			\sum_{v=1}^{m}\theta^{j}_{v}(t)(\lambda+\lambda_{v})
			\overline{F}_{v}(h^{j}_{v})
			\frac{\partial h^{j}_{v}}{\partial \theta^{j}_{i}(t)}
			+
			\sum_{v=1}^{m}
			\theta^{j+1}_{v}(t)(\lambda+\lambda_{v})\overline{F}_{v}(h^{j+1}_{v})
			\frac{\partial h^{j}_{v}}{\partial \theta^{j}_{i}(t)}
			\Bigg]
			\\
			&
			-\mathrm{e}^{r(T-t)}\gamma_{R_j}
			\sum_{v=1}^{m}(\lambda_{v}+\lambda)
			\left[
			a^{j+1}_{v}\overline{F}_{v}(h^{j+1}_{v})
			-\int_{h^{j}_{v}}^{h^{j+1}_{v}}\overline{F}_{v}(z)\mathrm{d}z
			\right]\frac{\partial h^{j}_{v}}{\partial \theta^{j}_{i}(t)}
			\\
			&
			-\mathrm{e}^{r(T-t)}\gamma_{R_j}\lambda
			\Bigg[
			\sum_{k<p}^{m}
			\left(
			\overline{F}_{k}(h^{j+1}_{k})-\overline{F}_{k}(h^{j}_{k})
			\right)
			\frac{\partial h^{j}_{k}}{\partial \theta^{j}_{i}(t)}
			\int_{h^{j}_{p}}^{h^{j+1}_{p}}\overline{F}_{p}(z)\mathrm{d}z
			\\
			&
			+
			\sum_{k<p}^{m}
			\int_{h^{j}_{k}}^{h^{j+1}_{k}}\overline{F}_{k}(z)\mathrm{d}z
			\left(
			\overline{F}_{p}(h^{j+1}_{p})-\overline{F}_{p}(h^{j}_{p})
			\right)
			\frac{\partial h^{j}_{p}}{\partial \theta^{j}_{i}(t)}
		    \Bigg]
			=0.
		\end{align*}
	
    In particular, $\theta^{n}_{i}(t)$ satisfies
		\begin{align*}
			&\Bigg[
			(\lambda+\lambda_{i})\mu_{i}-(\lambda+\lambda_{i})\int_{0}^{h^{n}_{i}}\overline{F}_{i}(z)\mathrm{d}z
			-
			\sum_{v=1}^{m}\theta^{n}_{v}(t)(\lambda+\lambda_{v})
			\overline{F}_{v}(h^{n}_{v})
			\frac{\partial h^{n}_{v}}{\partial \theta^{n}_{i}(t)}
			\Bigg]
			\\
			&
			-\mathrm{e}^{r(T-t)}\gamma_{R_n}
			\sum_{v=1}^{m}(\lambda_{v}+\lambda)
			\left[
			-
			\int_{h^{n}_{v}}^{\infty}\overline{F}_{v}(z)\mathrm{d}z
			\right]\frac{\partial h^{n}_{v}}{\partial \theta^{n}_{i}(t)}
			\\
			&
			-
			\mathrm{e}^{r(T-t)}\gamma_{R_n}\lambda
			\Bigg[
			\sum_{k<p}^{m}
			\left(
			-\overline{F}_{k}(h^{n}_{k})
			\right)
			\frac{\partial h^{n}_{k}}{\partial \theta^{n}_{i}(t)}
			\int_{h^{n}_{p}}^{\infty}\overline{F}_{p}(z)\mathrm{d}z
			\\
			&
			+
			\sum_{k<p}^{m}
			\left(
			-\overline{F}_{p}(h^{n}_{p})
			\right)
			\frac{\partial h^{n}_{k}}{\partial \theta^{n}_{i}(t)}
			\int_{h^{n}_{k}}^{\infty}\overline{F}_{k}(z)\mathrm{d}z
			\Bigg]
			\\
			&
			=0.
		\end{align*}
	
    The solution to the value function follows the same methodology as detailed in Theorem \ref{Theorem-Proportional reinsurance-equilibrium strategy}.
\end{proof}

\subsection{Proof of \cref{Theorem-ChainChange}}
\label{proof-ChainChange}
\begin{proof}
	For $j=1, 2, \cdots ,k-1,$  it follows from (\ref{one insurer-H}) that $H_{1j}(t)=H_{2j}(t)=H_{j}(t)$.
	
	Suppose $\gamma_{R_{k-1}}<\gamma_{R_k}$. Then $H_{1k}(t)>H_{2k}(t)$ and $H_{1j}(t)<H_{2j}(t)$ for $j=k+1, \cdots, n$.

	First, we compare  $\theta_{1k-1}$ and $\theta_{2k-1}$. From the relation
	\begin{equation*}
		\begin{aligned}
			\theta_{1k-1}-\theta_{2k-1}=\frac{1}{2}\left(\theta_{1k}-\theta_{2k}\right),
		\end{aligned}
	\end{equation*}
	it follows that $\theta_{1k-1}-\theta_{2k-1}$ and $\theta_{1k}-\theta_{2k}$
	have the same sign.
		\begin{align*}
			&\quad\theta_{1k}-\theta_{2k}
			\\
			&=\frac{\mathrm{e}^{r(T-t)} \sigma_1^2}{2^k o_1}
			\cdot
			\left[
			\frac{1}{\frac{1}{2^{k-1}\gamma_{1}}+\sum_{l=1}^{k-2}\frac{1}{2^{k-1-l}\gamma_{R_l}}+\frac{1}{\gamma_{R_{k-1}}}}
			-
			\frac{1}{\frac{1}{2^{k-1}\gamma_{1}}+\sum_{l=1}^{k-2}\frac{1}{2^{k-1-l}\gamma_{R_l}}+\frac{1}{\gamma_{R_k}}}
			\right] 
			\\
			&\quad
			+ \frac{\theta_{1k+1}-\theta_{2k+1}}{2}
			\\
			&=\frac{\mathrm{e}^{r(T-t)} \sigma_1^2}{2^k o_1}
			\cdot
			\left[
			\frac{ \frac{1}{\gamma_{R_k}}-\frac{1}{\gamma_{R_{k-1}}} }
			{\left(\frac{1}{2^{k-1}\gamma_{1}}+\sum_{l=1}^{k-2}\frac{1}{2^{k-1-l}\gamma_{R_l}}+\frac{1}{\gamma_{R_{k-1}}}\right)
				\cdot
				\left(\frac{1}{2^{k-1}\gamma_{1}}+\sum_{l=1}^{k-2}\frac{1}{2^{k-1-l}\gamma_{R_l}}+\frac{1}{\gamma_{R_k}}\right)}
			\right] 
			\\
			&\quad
			+ \frac{\theta_{1k+1}-\theta_{2k+1}}{2}.   	
		\end{align*}
	
	When $k=n$, we have $\theta_{1n-1}-\theta_{2n-1}=\frac{\theta_{1n}-\theta_{2n}}{2}$. At this point, $H_{2n}(t)-H_{1n}(t)=
	\left(\frac{1}{\gamma_{R_n}}-\frac{1}{\gamma_{R_{n-1}}}\right)\frac{o_1}{\sigma_1^2\mathrm{e}^{r(T-t)}}<0$ and $\theta_{1n}-\theta_{2n}$ is 
	
	\begin{equation*}
		\begin{aligned}
			&\quad \theta_{1n}-\theta_{2n}
			\\
			&=\frac{1}
			{2^{n-1}
				\left(2H_{1n}(t)o_1+\gamma_{R_n}\mathrm{e}^{r(T-t)}\sigma_1^2H_{1n}(t) ^2
				\right)
				\left(
				2H_{2n}(t)o_1+\gamma_{R_{n-1}}\mathrm{e}^{r(T-t)}\sigma_1^2H_{2n}(t)^2
				\right)
			}
			\times
			\\
			&
			\Bigg[ 
			2o_{1}^2(H_{2n}(t)-H_{1n}(t))
			+
			2(\gamma_{R_n}-\gamma_{R_{n-1}})\mathrm{e}^{r(T-t)}H_{1n}(t)H_{2n}(t)\sigma_{1}^2o_{1}
			\\
			&
			+
			\sigma_{1}^2 o_{1} \mathrm{e}^{r(T-t)}(H_{2n}(t)^2\gamma_{R_{n-1}}-H_{1n}(t)^2\gamma_{R_{n}})
			+
			\gamma_{R_{n-1}}\gamma_{R_{n}}\mathrm{e}^{2r(T-t)}\sigma_{1}^4 H_{1n}(t)H_{2n}(t)
			(H_{2n}(t)-H_{1n}(t))
			\Bigg]
			\\
			&=
			\frac{H_{2n}(t)-H_{1n}(t)}
			{2^{n-1}
				\left(2H_{1n}(t)o_1+\gamma_{R_n}\mathrm{e}^{r(T-t)}\sigma_1^2H_{1n}(t) ^2
				\right)
				\left(
				2H_{2n}(t)o_1+\gamma_{R_{n-1}}\mathrm{e}^{r(T-t)}\sigma_1^2H_{2n}(t)^2
				\right)
			}
			\times
			\\
			&
			\Bigg[ 
			2o_1^2
			+
			\sigma_1^2 o_1 \mathrm{e}^{r(T-t)} \frac{\gamma_{R_{n-1}}H_{2n}(t)^2+
				(\gamma_{R_n}-\gamma_{R_{n-1}})H_{1n}(t)
				H_{2n}(t)-\gamma_{R_n}H_{1n}(t)^2
			}
			{H_{2n}(t)-H_{1n}(t)}
			\Bigg], 
		\end{aligned}
	\end{equation*}
	where  $H_{1n}(t)=\frac{H_{n-1}(t)}{2}+\frac{1}{\gamma_{R_{n-1}}}\frac{o_1}{\sigma_1^2\mathrm{e}^{r(T-t)}}$, $H_{2n}(t)=\frac{H_{n-1}(t)}{2}+\frac{1}{\gamma_{R_{n}}}\frac{o_1}{\sigma_1^2\mathrm{e}^{r(T-t)}}$, and $(H_{2n}(t)-H_{1n}(t))<0$. 
	
	\begin{equation*}
		\begin{aligned}
			&\quad \gamma_{R_{n-1}}H_{2n}(t)^2+
			(\gamma_{R_n}-\gamma_{R_{n-1}})H_{1n}(t)
			H_{2n}(t)-\gamma_{R_n}H_{1n}(t)^2
			\\
			&=\frac{H_{n-1}(t)o_1}{2\sigma_1^2\mathrm{e}^{r(T-t)}}\cdot\frac{\gamma_{R_{n-1}}^2-\gamma_{R_n}^2}{\gamma_{R_n}\gamma_{R_{n-1}}}+
			\left(\frac{o_1}{\sigma_1^2\mathrm{e}^{r(T-t)}}\right)^2
			\cdot
			\frac{\left(\gamma_{R_{n-1}}-\gamma_{R_n}\right)\left(\gamma_{R_{n-1}}^2+\gamma_{R_n}^2\right)}
			{\gamma_{R_{n-1}}^2\gamma_{R_n}^2 }<0,
		\end{aligned}
	\end{equation*}
	which implies that $\theta_{1n}(t)<\theta_{2n}(t)$.
	
	\vspace{0.3cm}
	
	When $k=n-1$,  
	\begin{equation*}
		\begin{aligned}
			\theta_{1n-1}-\theta_{2n-1}
			&=\frac{1}{2^{n-1}}
			\left(
			\frac{1}{H_{1n-1}(t)}-\frac{1}{H_{2n-1}(t)}
			\right)
			+\frac{\theta_{1n}(t)-\theta_{2n}(t)}{2}
			\\
			&=\frac{o_1}{ 2^{n-1}\mathrm{e}^{r(T-t)}\sigma_1^2 }
			\left(\frac{1}{\gamma_{R_{n-1}}}-\frac{1}{\gamma_{R_{n-2}}}\right)
			\frac{1}{H_{1n-1}(t)H_{2n-1}(t)}
			+\frac{\theta_{1n}(t)-\theta_{2n}(t)}{2},
		\end{aligned}
	\end{equation*}
	\begin{equation*}
		\begin{aligned}
			&\quad \theta_{1n}-\theta_{2n}
			\\
			&=\frac{1}{2^{n-1}}
			\Bigg[
			\left(\frac{1}{H_{1n}(t)}-\frac{o_1}{2H_{1n}(t)o_1+\gamma_{R_n}\mathrm{e}^{r(T-t)}\sigma_1^2H_{1n}(t)^2 }\right)
			\\
			&\quad
			-
			\left(
			\frac{1}{H_{2n}(t)}-\frac{o_1}{2H_{2n}(t)o_1+\gamma_{R_n}\mathrm{e}^{r(T-t)}\sigma_1^2H_{2n}(t)^2 }
			\right)
			\Bigg]
			\\
			&=\frac{H_{2n}(t)-H_{1n}(t)}{ 2^{n-1}H_{1n}(t) H_{2n}(t)}
			\times
			\\
			&\quad
			\frac{2o_1^2
				+\gamma_{R_n}\mathrm{e}^{r(T-t)}\sigma_1^2o_1( H_{1n}(t) + H_{2n}(t) ) + \gamma_{R_n}^2\mathrm{e}^{2r(T-t)}\sigma_1^2H_{1n}(t)H_{2n}(t) }
			{4o_1^2
				+2\gamma_{R_n}\mathrm{e}^{r(T-t)}\sigma_1^2o_1( H_{1n}(t) + H_{2n}(t) ) + \gamma_{R_n}^2\mathrm{e}^{2r(T-t)}\sigma_1^2H_{1n}(t)H_{2n}(t) }>0,
		\end{aligned}
	\end{equation*}           
	where       
	\begin{equation*}
		\frac{1}{2}<\frac{2o_1^2
			+\gamma_{R_n}\mathrm{e}^{r(T-t)}\sigma_1^2o_1( H_{1n}(t) + H_{2n}(t) ) + \gamma_{R_n}^2\mathrm{e}^{2r(T-t)}\sigma_1^2H_{1n}(t)H_{2n}(t) }
		{4o_1^2
			+2\gamma_{R_n}\mathrm{e}^{r(T-t)}\sigma_1^2o_1( H_{1n}(t) + H_{2n}(t) ) + \gamma_{R_n}^2\mathrm{e}^{2r(T-t)}\sigma_1^2H_{1n}(t)H_{2n}(t) }<1. 
	\end{equation*}
	Thus we obtain
	
	\begin{equation*}
		\frac{H_{2n}(t)-H_{1n}(t)}{ 2^{n}H_{1n}(t) H_{2n}(t)}
		<
		\theta_{1n}(t)-\theta_{2n}(t)
		<
		\frac{H_{2n}(t)-H_{1n}(t)}{ 2^{n-1}H_{1n}(t) H_{2n}(t)},
	\end{equation*}
	and
	\begin{equation*}
		\begin{aligned}
			&\quad \frac{H_{2n}(t)-H_{1n}(t)}{ 2^{n-1}H_{1n}(t) H_{2n}(t)}
			&= \frac{
				\frac{o_1}{ \mathrm{e}^{r(T-t)} \sigma_1^2}
				\left(\frac{1}{2\gamma_{R_{n-2}}}-\frac{1}{2\gamma_{R_{n-1}}}\right)
			}{2^{n-1}H_{1n}(t) H_{2n}(t)}
			&=\frac{\frac{o_1}{ \mathrm{e}^{r(T-t)}\sigma_1^2} 
				\left(\frac{1}{\gamma_{R_{n-2}}}-\frac{1}{\gamma_{R_{n-1}}}\right)}
			{2^{n}H_{1n}(t) H_{2n}(t)}.
		\end{aligned}
	\end{equation*}
	
	Substituting it into $\theta_{1n-1}(t)-\theta_{2n-1}(t)$ yields
	
	\begin{equation*}
		\begin{aligned}
			&\quad \theta_{1n-1}(t)-\theta_{2n-1}(t)
			\\
			&
			<
			\frac{o_1}{ 2^{n-1}\mathrm{e}^{r(T-t)}\sigma_1^2 }
			\left(\frac{1}{\gamma_{R_{n-1}}}-\frac{1}{\gamma_{R_{n-2}}}\right)
			\left(
			\frac{1}{H_{1n-1}(t)H_{2n-1}(t)}
			-\frac{1}{4H_{1n}(t)H_{2n}(t)}
			\right)
			\\
			&
			=
			\frac{o_1}{ 2^{n-1}\mathrm{e}^{r(T-t)}\sigma_1^2 }
			\left(\frac{1}{\gamma_{R_{n-1}}}-\frac{1}{\gamma_{R_{n-2}}}\right)
			\left(
			\frac{4H_{1n}(t)H_{2n}(t)-H_{1n-1}(t)H_{2n-1}(t)}
			{H_{1n-1}(t)H_{2n-1}(t)H_{1n}(t)
				H_{2n}(t)}
			\right).
		\end{aligned}
	\end{equation*}
	Since $2H_{1n}(t)>H_{1n-1}(t)$ and $2H_{2n}(t)>H_{1n-1}(t)$, it follows that
    $$4H_{1n}(t)H_{2n}(t)-H_{1n-1}(t)H_{2n-1}(t)>0, $$  
    hence $\theta_{1n-1}<\theta_{2n-1}$. At this point, $\forall j\in\left\{1,2,\cdots,n-2\right\}$,  $\theta_{1j}-\theta_{2j}=\frac{1}{2^{n-1-j}}\left(\theta_{1n-1}-\theta_{2n-1}\right)$, so  $\theta_{1j}<\theta_{2j}$.  The cases $k=2, 3, \cdots, n-2$ can be established similarly.
	
\end{proof}

\subsection{Proof of \cref{Theorem-optimal reinsurance chain structure}}
\label{proof of optimal reinsurance chain structure}
\begin{proof}
	Suppose that the optimal reinsurance chain structure satisfies 
	\begin{equation}
		\label{optimal reinsurance chain structure}
		\gamma_{R_1}<\gamma_{R_2}<\cdots<\gamma_{R_n}. 
	\end{equation}
	We prove our conjecture by contradiction.  Assume that there exists another reinsurance chain, different from $(\ref{optimal reinsurance chain structure} )$, in which strategy $\theta_{1}(t)$
	attains the minimum. Then, one can necessarily identify an adjacent pair 
	$\gamma_{k-1}>\gamma_{k}$
	within the chain. By \Cref{Theorem-ChainChange}, exchanging $\gamma_{k-1}$
	and $\gamma_{k}$
	would lead to a smaller value of strategy $\theta_{1}(t)$, which contradicts the assumption. 
\end{proof}

\subsection{Proof of \cref{Theorem-excess-of-loss reinsurance-two reinsurer}}
\label{proof of Excess-of-loss reinsurance-economy analysis}
\begin{proof}
	The extended HJB equation for insurer $1$ is as follows. 
	\begin{equation}
		\label{insurer-excess-of-loss reinsurance-HJB}
		\left\{
		\begin{aligned}
			&\sup_{u_{1}\in\mathcal{A}_{1}}\mathcal{L}^{u_{1}, u_{1-}}V^{1}(t, x_{1})-\frac{\gamma_{1}}{2}\mathcal{L}^{u_{1}, u_{1-}}g^{1}(t, x_{1})^2+\gamma_{1}
			g^{1}(t, x_{1})\mathcal{L}^{u_{1}, u_{1-}}g^{1}(t, x_{1})=0,
			\\
			&V^{1}(T, x_{1})=x_{1},\quad
			g^{1}(T, x_{i})=x_{1},\quad
			\mathcal{L}^{u^*_{1}, u^*_{1-}}g^{1}(t, x_{1})=0.
		\end{aligned}
		\right.
	\end{equation}
	
	Assume that
	\begin{equation*}
		V^{1}(t, x_{1})=\mathrm{e}^{r(T-t)}x_{1}+M_{1}(t), g^{1}(t, x_{1})=\mathrm{e}^{r(T-t)}x_{i}+m_{1}(t).
	\end{equation*}
	
	For notational convenience, we write $a_{1}(t), a_{2}(t)$ for $a_{1}, a_{2}$, respectively. 
	The extended HJB equation for insurer $1$ becomes
	\begin{equation*}
		\begin{aligned}
			\frac{\partial M_{1}}{\partial t}
			&+\mathrm{e}^{r(T-t)}\left[(\mu-r)\pi_{1}(t)
			\right]
			-
			\frac{\gamma_{1}}{2}\mathrm{e}^{2r(T-t)}
			\pi_{1}(t)^2\sigma^2
			\\
			&
			+\mathrm{e}^{r(T-t)}
			\left[
			\left(\eta_{1}-\theta_1(t)\right)o_{1}
			+(1+\theta_1(t))(\lambda+\lambda_{1})\int_{0}^{a_{1}}\overline{F}_{1}(z)\mathrm{d}z
			\right]
			\\
			&
			-
			\mathrm{e}^{r(T-t)}(\lambda+\lambda_{1})\int_{0}^{a_{1}}\overline{F}_{1}(z)\mathrm{d}z
			-
			\mathrm{e}^{2r(T-t)}\frac{\gamma_{1}}{2}(\lambda+\lambda_{1})\int_{0}^{a_{1}}2z\overline{F}_{1}(z)\mathrm{d}z=0.
		\end{aligned}
	\end{equation*}
	
	From the first-order condition of $a_1$, we obtain
	\begin{equation*}
		a_{1}^*(t)=\frac{\theta_1(t)}{\gamma_{1}\mathrm{e}^{r(T-t)}},
	\end{equation*}
	
	The extended HJB equation for  reinsurer $R_{1}$ is given by
	\begin{equation*}
		\begin{aligned}
			&\sup_{u_{R_1}\in\mathcal{A}_{R_1}}\mathcal{L}^{u_{R_1}, u_{R_1-}}V^{R_1}(t, y_1)
			-\frac{\gamma_{R_1}}{2}\mathcal{L}^{u_{R_1}, u_{R_1-}}g^{R_1}(t, y_1)^2
			\\
			&
			+\gamma_{{R_1}}
			g^{R_1}(t, y_{1})\mathcal{L}^{u_{R_1}, u_{R_1-}}g^{R_1}(t, y_{1})=0.
			\\
			&
			V^{R_1}(T, y_1)=y_1,\quad
			g^{R_1}(T, y_1)=y_1,\quad
			\mathcal{L}^{u^*_{R_1}, u^*_{R_1-}}g^{R_1}(t, y_1)=0.
		\end{aligned}
	\end{equation*}
	
	Assume that 
	\begin{equation*}
		V^{R_1}(t, y_1)=\mathrm{e}^{r(T-t)}y_1+M_{R_1}(t), g^{R_1}(t, y_1)=\mathrm{e}^{r(T-t)}y_1+m_{R_1}(t),
	\end{equation*} 
	
	The extended HJB equation for  reinsurer $R_1$  then becomes
	\begin{equation*}
		\begin{aligned}
			\frac{\partial M_{R_{1}}}{\partial t}
			&+\mathrm{e}^{r(T-t)}\left[(\mu-r)\pi_{R_{1}}(t)
			\right]
			-
			\frac{\gamma_{R_{1}}}{2}\mathrm{e}^{2r(T-t)}
			\pi_{R_{1}}(t)^2\sigma^2
			\\
			&
			+
			\mathrm{e}^{r(T-t)}
			\Bigg[
			\left(\theta_1(t)-\theta_2(t)\right)o_{1}
			-
			(1+\theta_1(t)) (\lambda+\lambda_{1}) \int_{0}^{a_{1}} \overline{F}_{1}(z)\mathrm{d}z
			\\
			&
			+
			(1+\theta_2(t))(\lambda+\lambda_{1})\int_{0}^{a_{1}+a_{2}}\overline{F}_{1}(z)\mathrm{d}z
			\Bigg]
			-
			\mathrm{e}^{r(T-t)}(\lambda+\lambda_{1})\int_{a_{1}}^{a_{1}+a_{2}}\overline{F}_{1}(z)\mathrm{d}z
			\\
			&
			-
			\mathrm{e}^{2r(T-t)} \frac{\gamma_{R_{1}}}{2} (\lambda+\lambda_{1}) \int_{a_{1}}^{a_{1}+a_{2}}2(z-a_1) \overline{F}_{1}(z)\mathrm{d}z=0.
		\end{aligned}
	\end{equation*}
	
	Analogous to $a_1^*(t)$, we obtain $a_2^*(t)$ as
	\begin{equation*}
		a_{2}^*(t)=\frac{\theta_2(t)}{\gamma_{R_{1}}\mathrm{e}^{r(T-t)}}.
	\end{equation*}
	
	The first-order condition with respect to $\theta_{1}(t)$ yields
	\begin{equation*}
		\frac{1}{\delta_1}-\frac{1+\theta_{1}}{\gamma_1\mathrm{e}^{r(T-t)}}+\frac{1+\theta_{2}}{\gamma_1\mathrm{e}^{r(T-t)}}\mathrm{e}^{-\delta_1a_2}
		-\frac{\mathrm{e}^{-\delta_1a_2}-1}{\gamma_1\mathrm{e}^{r(T-t)}}
		-\frac{\gamma_{R_1}}{\gamma_1}
		\left(a_2\mathrm{e}^{-\delta_1a_2}-\frac{(1-\mathrm{e}^{-\delta_1a_2})}{\delta_1}\right)=0.
	\end{equation*}
	
	Solving for $\theta_{1}^*(t)$ gives
	\begin{equation*}
		\theta_1^*(t)=\frac{\gamma_1\mathrm{e}^{r(T-t)}}{\delta_1}
		+\theta_2(t)\mathrm{e}^{-\delta_1a_2}
		-\gamma_{R_1}\mathrm{e}^{r(T-t)}
		\left(a_2\mathrm{e}^{-\delta_1a_2}
		-\frac{(1-\mathrm{e}^{-\delta_1a_2})}{\delta_1}\right).
	\end{equation*}
	
	The extended HJB equation for the reinsurer $R_{2}$ is given by
	\begin{equation*}
		\begin{aligned}
			&\sup_{u_{R_2}\in\mathcal{A}_{R_2}}\mathcal{L}^{u_{R_2}, u_{R_2-}}V^{R_2}(t, y_2)
			-\frac{\gamma_{R_2}}{2}\mathcal{L}^{u_{R_2}, u_{R_2-}}g^{R_2}(t, y_2)^2
			\\
			&
			+
			\gamma_{{R_2}}
			g^{R_2}(t, y_{2})\mathcal{L}^{u_{R_2}, u_{R_2-}}g^{R_2}(t, y_{2})=0,
			\\
			&
			V^{R_2}(T, y_2)=y_2,\quad
			g^{R_2}(T, y_2)=y_2,\quad
			\mathcal{L}^{u^*_{R_2}, u^*_{R_2-}}g^{R_2}(t, y_2)=0,
		\end{aligned}
	\end{equation*}
	
	Assume 
	\begin{equation*}
		V^{R_{2}}(t, y_2)=\mathrm{e}^{r(T-t)}y_2+M_{R_2}(t), g^{R_{2}}(t, y_2)=\mathrm{e}^{r(T-t)}y_2+m_{R_2}(t),
	\end{equation*} 
	
	The extended HJB equation for reinsurer $R_2$ then  becomes
	
	\begin{equation*}
		\begin{aligned}
			\frac{\partial M_{R_{2}}}{\partial t}
			&+\mathrm{e}^{r(T-t)}\left[(\mu-r)\pi_{R_{2}}(t)
			\right]
			-
			\frac{\gamma_{R_{2}}}{2}\mathrm{e}^{2r(T-t)}
			\pi_{R_{2}}(t)^2\sigma^2
			\\
			&
			+\mathrm{e}^{r(T-t)}
			\left[
			\left(1+\theta_2(t)\right)o_{1}
			-
			(1+\theta_2(t))(\lambda+\lambda_{1})\int_{0}^{a_{1}+a_{2}}\overline{F}_{1}(z)\mathrm{d}z
			\right]
			\\
			&
			-
			\mathrm{e}^{r(T-t)}(\lambda+\lambda_{1})
			\int_{a_{1}+a_{2}}^{\infty}\overline{F}_{1}(z)\mathrm{d}z
			\\
			&
			-
			\mathrm{e}^{2r(T-t)}\frac{\gamma_{R_{2}}}{2}(\lambda+\lambda_{1})
			\int_{a_{1}+a_{2}}^{\infty}2(z-a_1-a_2)\overline{F}_{1}(z)\mathrm{d}z=0.
		\end{aligned}
	\end{equation*}
	Setting $h_{2}=a_1+a_2$, we obtain 
	\begin{equation*}
		\begin{aligned}
			\frac{\partial h_2}{\partial \theta_2}
			&
			=
			\frac{\mathrm{e}^{-\delta_1a_2}}{\gamma_{1}\mathrm{e}^{r(T-t)}}
			\left[
			1-\theta_{2}(t)\frac{\delta_1}{\gamma_{R_{1}}\mathrm{e}^{r(T-t)}}
			+\delta_1a_{2}
			\right]+\frac{1}{\gamma_{R_{1}}\mathrm{e}^{r(T-t)}}
			\\
			&
			=
			\frac{\mathrm{e}^{-\delta_1a_2}}{\gamma_{1}\mathrm{e}^{r(T-t)}}
			+\frac{1}{\gamma_{R_{1}}\mathrm{e}^{r(T-t)}}.
		\end{aligned}
	\end{equation*}
	
	The first-order condition with respect to  $\theta_2(t)$ yields
	\begin{equation*}
		\begin{aligned}
			&\left[
			o_{1}
			-
			(\lambda+\lambda_{1})\int_{0}^{h_{2}}\overline{F}_{1}(z)\mathrm{d}z
			-
			(1+\theta_2(t))(\lambda+\lambda_{1})\overline{F}_{1}(h_{2})\frac{\partial h_2}{\partial \theta_2}
			\right]
			\\
			&
			+
			(\lambda+\lambda_{1})\overline{F}_{1}(h_{2})\frac{\partial h_2}{\partial \theta_2}
			+
			\mathrm{e}^{r(T-t)}\gamma_{R_{2}}(\lambda+\lambda_{1})\frac{\partial h_2}{\partial \theta_2}
			\int_{h_{2}}^{\infty}\overline{F}_{1}(z)\mathrm{d}z=0.
		\end{aligned}
	\end{equation*}
	
	After simplification, $\theta_{2}(t)$ satisfies
	\begin{equation*}
		\begin{aligned}
			\frac{1}{\delta_1}
			+\left(\frac{\gamma_{R_{2}}\mathrm{e}^{r(T-t)}}{\delta_{1}}-\theta_{2}(t)\right)
			\frac{\partial h_2}{\partial \theta_2}=0.
		\end{aligned}
	\end{equation*}
	
	 We now establish the uniqueness of the solution. To this end, we divide the domain of $\theta_{2}$ into three intervals for analysis: $\left(0,\frac{\gamma_{R_{2}}\mathrm{e}^{r(T-t)}}{\delta_{1}}\right]$,  $\left(\frac{\gamma_{R_{2}}\mathrm{e}^{r(T-t)}}{\delta_{1}},\frac{\gamma_{R_{1}}\mathrm{e}^{r(T-t)}}{\delta_{1}}+
	 \frac{\gamma_{R_{2}}\mathrm{e}^{r(T-t)}}{\delta_{1}}\right)$ 
	 and 
	 $\left[\frac{\gamma_{R_{1}}\mathrm{e}^{r(T-t)}}{\delta_{1}}+
	 \frac{\gamma_{R_{2}}\mathrm{e}^{r(T-t)}}{\delta_{1}}, \infty\right)$. 	Define 
	\begin{equation*}
		K(\theta_2)=\frac{1}{\delta_1}
		+(\frac{\gamma_{R_{2}}\mathrm{e}^{r(T-t)}}{\delta_{1}}-\theta_{2})\frac{\partial h_2}{\partial \theta_2}.
	\end{equation*}

	(i)
	When 
	$\theta_{2}\in\left(0,\frac{\gamma_{R_{2}}\mathrm{e}^{r(T-t)}}{\delta_{1}}\right]$, we have $K(\theta_{2})>0$. 
	
	(ii)
	Assume $\forall x\geq 0.$ At this time, 
	
	\begin{equation*}
		\begin{aligned}
			K\left(x+\frac{\gamma_{R_{1}}\mathrm{e}^{r(T-t)}}{\delta_{1}}+
			\frac{\gamma_{R_{2}}\mathrm{e}^{r(T-t)}}{\delta_{1}}\right)
			&=
			-
			x\frac{\partial h_2}{\partial \theta_2}
			-
			\frac{\gamma_{R_{1}}\mathrm{e}^{r(T-t)}}{\delta_{1}}\cdot\frac{\mathrm{e}^{-\delta_1a_2}}{\gamma_{1}\mathrm{e}^{r(T-t)}}<0.
		\end{aligned}
	\end{equation*}
	
	Consequently, when $\theta_{2}\in\left[\frac{\gamma_{R_{1}}\mathrm{e}^{r(T-t)}}{\delta_{1}}+
	\frac{\gamma_{R_{2}}\mathrm{e}^{r(T-t)}}{\delta_{1}}, \infty\right)$, $K(\theta_{2})<0$. 
	
	(iii) 
	The partial derivatives of $K(\theta_2)$ are given by
	
	\begin{equation}
		\left\{
		\begin{aligned}
			\frac{\partial K(\theta_2)}{\partial \theta_{2}}
			&=
			\frac{-\delta_1\mathrm{e}^{-\delta_1a_2}}{\gamma_1\gamma_{R_1}\mathrm{e}^{2r(T-t)}}\left(\frac{\gamma_{R_{2}}\mathrm{e}^{r(T-t)}}{\delta_{1}}-\theta_{2}(t)\right)
			-\frac{\partial h_2}{\partial \theta_2},
			\\
			\frac{\partial^2 K(\theta_2)}{\partial \theta_{2}^2}
			&=	\frac{\delta_1^2\mathrm{e}^{-\delta_1a_2}}
			{\gamma_1\gamma_{R_1}^2\mathrm{e}^{3r(T-t)}}
			\left(
			\frac{2\gamma_{R_{1}}\mathrm{e}^{r(T-t)}}{\delta_{1}}+
			\frac{\gamma_{R_{2}}\mathrm{e}^{r(T-t)}}{\delta_{1}}
			-
			\theta_{2}(t)
			\right).
		\end{aligned}
		\right.
	\end{equation}
	When 
	$\theta_2(t)=\frac{2\gamma_{R_{1}}\mathrm{e}^{r(T-t)}}{\delta_{1}}+
	\frac{\gamma_{R_{2}}\mathrm{e}^{r(T-t)}}{\delta_{1}}
	$, $\frac{\partial^2 K(\theta_2)}{\partial \theta_{2}^2}=0$. And when 
	$\theta_{2}\in\Big(\frac{\gamma_{R_{2}}\mathrm{e}^{r(T-t)}}{\delta_{1}},\frac{\gamma_{R_{1}}\mathrm{e}^{r(T-t)}}{\delta_{1}}+
	\frac{\gamma_{R_{2}}\mathrm{e}^{r(T-t)}}{\delta_{1}}\Big)$, 
	$\frac{\partial^2 K(\theta_2)}{\partial \theta_{2}^2}>0$. 
	Hence, $\left(\frac{2\gamma_{R_{1}}\mathrm{e}^{r(T-t)}}{\delta_{1}}+
	\frac{\gamma_{R_{2}}\mathrm{e}^{r(T-t)}}{\delta_{1}}\right)$ is the maximizer of $\frac{\partial K(\theta_{2})}{\partial \theta_{2}}$. 
	Moreover,
	\begin{equation*}
		\begin{aligned}
			\frac{\partial K}{\partial \theta_{2}}
			\left( \frac{\gamma_{R_{1}}\mathrm{e}^{r(T-t)}}{\delta_{1}}+
			\frac{\gamma_{R_{2}}\mathrm{e}^{r(T-t)}}{\delta_{1}}\right)= -\frac{1}{\gamma_{R_{1}}\mathrm{e}^{r(T-t)}}<0.
		\end{aligned}
	\end{equation*}
	
	Therefore, for $\theta_{2}\in\left(\frac{\gamma_{R_{2}}\mathrm{e}^{r(T-t)}}{\delta_{1}},\frac{\gamma_{R_{1}}\mathrm{e}^{r(T-t)}}{\delta_{1}}+
	\frac{\gamma_{R_{2}}\mathrm{e}^{r(T-t)}}{\delta_{1}}\right)$, 
	$\frac{\partial K(\theta_2)}{\partial \theta_{2}}<0$. Furthermore, since $K\left(\frac{\gamma_{R_{2}}\mathrm{e}^{r(T-t)}}{\delta_{1}}\right)>0$ and $K\left(\frac{\gamma_{R_{1}}\mathrm{e}^{r(T-t)}}{\delta_{1}}+
	\frac{\gamma_{R_{2}}\mathrm{e}^{r(T-t)}}{\delta_{1}}\right)<0$, 
	it follows that the equation $K(\theta_{2})=0$ has a unique solution in the interval $\left(\frac{\gamma_{R_{2}}\mathrm{e}^{r(T-t)}}{\delta_{1}},\frac{\gamma_{R_{1}}\mathrm{e}^{r(T-t)}}{\delta_{1}}+
	\frac{\gamma_{R_{2}}\mathrm{e}^{r(T-t)}}{\delta_{1}}\right)$.
\end{proof}

\end{appendices}

\section*{Acknowledgements}

This work was supported by the Natural Science Foundation of Xinjiang Uygur Autonomous Region (Grant No. 2024D01C39) and the National Natural Science Foundation of China (Grant Nos. 12261085, 12301149).

\section*{Declarations}

\textbf{Conflict of interest} The authors declare no financial or non-financial interests that are directly or indirectly related to the work submitted for publication.

\printbibliography
\end{document}